\documentclass[12pt]{article}
\usepackage{natbib}
\usepackage{amsmath}
\usepackage{amsthm}
\usepackage{amssymb}
\usepackage{amsthm}
\usepackage{graphicx}
\usepackage{sidecap}
\usepackage{subcaption}
\usepackage{caption}
\usepackage{comment}
\usepackage{float}
\graphicspath{{Graphs/}} 	
\usepackage{tikz}
\usepackage{pgf, pgfplots}  
\pgfplotsset{compat=1.17}
\usepackage{xcolor}
\usetikzlibrary{arrows}
\usetikzlibrary{calc}
\usetikzlibrary{positioning}
\usetikzlibrary{automata}
\usetikzlibrary{shapes}
\usetikzlibrary{shapes.geometric}
\usetikzlibrary{decorations.markings}
\usepackage{setspace}
\usepackage{fullpage}
\usepackage{pdfpages}
\usepackage{import}
\usepackage{rotating}
\renewcommand\fbox{\fcolorbox{white}{white}}
\newtheorem{thm}{Theorem}
\newtheorem{prop}{Proposition}
\newtheorem{cor}{Corollary}
\newtheorem{lem}{Lemma}
\theoremstyle{definition}
\newtheorem{defn}{Definition}
\newtheorem{rem}{Remark}
\theoremstyle{definition}
\theoremstyle{definition}
\newtheorem{ass}{Assumption}

\DeclareMathOperator{\Cov}{\mathrm{Cov}}
\DeclareMathOperator*{\argmax}{argmax}

\begin{document}
\title{Social networks, confirmation bias and shock elections 
\footnote{Thanks to Nizar Allouch, Itai Arieli, Yakov Babichenko, Matthew Elliott, and Marcos Ross Fernandes for helpful comments and suggestions. We also thank seminar participants at the University of Cambridge, University of East Anglia, University of Kent, and University of S\~ao Paulo. This work was supported by the Cambridge Endowment for Research in Finance (CERF) and the Economic and Social Research Council [award reference ES/P000738/1]. A very early version of this work was submitted as an undergraduate dissertation at the University of Cambridge under the title ``Why can't we all just get along? Embedding confirmation bias into a model of social learning''. Any remaining errors are the sole responsibility of the authors. Declaration of interests: none.}}

\author{Edoardo Gallo\footnote{University of Cambridge and Magdalene College, Cambridge CB3 0AG, UK. Email: \emph{edo@econ.cam.ac.uk.}}\and Alastair Langtry\footnote{University of Bristol.  Email: \emph{alastair.langtry@bristol.ac.uk.}} }

\date{\today \\ \vspace{7mm}}

\maketitle

\begin{abstract}
This paper links the increasing prominence of social networks in politics to recent shock election outcomes. In our set-up, agents learning on a network suffer from confirmation bias -- the tendency to ignore contrary evidence to one's own belief. Our main result is that confirmation bias slows down learning in a wide class of networks. For any network, consensus fails to emerge when confirmation bias is above a threshold. We explore the implications for elections and media markets. Confirmation bias increases the likelihood of shock elections, and encourages media organizations already at the fringe to take on even more extreme views. \\



\end{abstract}

\newpage
\begin{quote}
\singlespacing
\emph{\footnotesize{
The strongest bias in American politics is not a liberal bias or a conservative bias; it is a confirmation bias, or the urge to believe only things that confirm what you already believe to be true. Not only do we tend to seek out and remember information that reaffirms what we already believe, but there is also a ``backfire effect'' which sees people doubling down on their beliefs after being presented with evidence that contradicts them.}}
\flushright{
\footnotesize{``Your facts or mine?'' by E. Roller, \emph{NYT}, Oct 25th, 2016.}}
\end{quote}

\setstretch{1.3}
\vspace{6mm}

Social networks are increasingly becoming the primary channel for people to acquire information and form opinions. In an experiment involving 60m Facebook users prior to the 2010 US elections, \citet{bond201261} showed they could generate 340,000 additional votes using a social message that informed a user about friends that had voted, compared to an informational message without social network information. Unlike traditional media or gatherings in the local church, club or pub, online social networks make it very easy for a user to ``unfollow'' someone who does not share their opinion. In this way, they exacerbate the role of confirmation bias -- people's tendency to ignore information contrary to their view and reinterpret it as agreeing with their own \citep{pariser2011filter}. A natural question is whether the heightened role of confirmation bias affects the democratic process and whether it is a driver of election results such as Trump's win or Brexit.

The aim of this paper is to examine how confirmation bias affects the process of learning from others, and its consequences for elections and the media. Individuals are endowed with an initial belief at time $0$ and learn according to the well-known \cite{degroot1974} behavioural rule: they update their beliefs by taking weighted averages of their neighbours' beliefs.
We model confirmation bias as causing individuals to disconnect from those whose beliefs are ``too different'' to their own, and give correspondingly more weight to their own belief instead. Specifically, after the assignment of initial beliefs, each individual cuts connections with others who have beliefs further away from their own than a threshold. The individual reassigns the weight of these severed connections to themself, and they never reinstate these links after they have been cut. Mathematically, understanding the effect of confirmation bias in the model reduces to a comparison of the learning processes on the original network and the new sparser network.

Our main result is that, when agents pay enough attention to themselves, confirmation bias slows down learning in \emph{any} reversible network -- a superset of the class of undirected networks. In the proof, we relate the strength of confirmation bias to the whole spectrum of eigenvalues of the network, which governs the rate of convergence to a consensus. The result relies critically on a key feature of confirmation bias known as the ``backfire effect'' -- individuals do not just ignore contradictory evidence, but double down on their own belief when exposed to it. In our set-up, the ``backfire effect'' is the assumption that individuals reassign the weight of severed connections to themselves. If we were to remove this crucial feature of confirmation bias by assuming that the weight of the severed links is partly reassigned to other surviving links then there are counterexamples in which learning is faster in a sparser network. Our analytic result applies to reversible networks, but we show using simulations that it largely holds in more general networks.

We examine the consequences of social learning affected by confirmation bias in two applications to elections and media markets. In the first application, we embed our social learning framework into a two-candidate voting model in which sincere voting is a weakly dominant strategy, so whenever there is an election individuals vote for the candidate closer to their current belief. We restrict the distribution of beliefs to focus on the interesting case when a society would vote for the same candidate before learning takes place and at the end of the learning process. We define a society as having \emph{shock elections} if the other candidate wins at any point in time during the learning process. We prove that a society never has shock elections without confirmation bias, but it can do so if confirmation bias is high enough to remove some connections.

The second application embeds our social learning framework in a Hotelling-style model of a media market. We model media organizations as agents who choose an ideology (i.e., an initial belief) and stick to it forever. Their sole objective is to maximize their audience. Individuals follow exactly one media organization due to limited attention, and choose the one closest to their own initial belief. Confirmation bias prevents individuals from following a media organization whose ideology is ``too far'' from their own initial opinion. We find that the ideology of the most extreme media organization becomes more extreme as the strength of confirmation bias increases. 

Finally, we present a series of simulation exercises to test the robustness of our main insights that confirmation bias slows down consensus and creates shock elections. These exercises corroborate the main insights by relaxing the assumptions imposed for tractability in the analytic results. We also use these exercises to explore how confirmation bias differs from homophily, and how it interacts with the initial beliefs. We show that confirmation bias slows learning over and above its effect via increases in homophily, and find suggestive evidence that it has a larger impact when agents have more extreme initial beliefs.

\textbf{Literature review.} This paper sits at the intersection of multiple areas of inquiry both within and outside economics. We review each in turn, highlighting papers that are particularly relevant to this work.

\textbf{Psychology.} The study of confirmation bias has a long history in psychology; a comprehensive review by \citet{nickerson1998confirmation} shows its relevance to a large range of issues, including judicial outcomes \citep{kuhn1994well}, policymaking \citep{tuchman2011march}, and medical decisions \citep{elstein1979psychology}.\footnote{\citet{nickerson1998confirmation}'s opening paragraph states: ``If one were to attempt to identify a single problematic aspect of human reasoning that deserves attention above all others, the confirmation bias would have to be among the candidates for consideration. Many have written about this bias, and it appears to be sufficiently strong and pervasive that one is led to wonder whether the bias, by itself, might account for a significant fraction of the disputes, altercations, and misunderstandings that occur among individuals, groups, and nations.'' The careful reader will note that the pervasiveness of confirmation bias may extend to the quote itself.} A difficulty posed by the vastness of this literature is, quoting Nickerson, that ``confirmation bias has been used in the psychological literature to refer to a variety of phenomena'' (p. 175). The core idea present throughout the psychology literature is that people are biased against information that conflicts with their own beliefs.

\citet{nickerson1998confirmation}'s working definition is ``unwitting selectivity in the acquisition and use of evidence'' (p. 175) and he puts special emphasis on the non-deliberate nature of the bias which emerges as a heuristic to quickly process information. A key feature is therefore that the selectivity in the acquisition of evidence leads to \emph{selective exposure} to beliefs that are close to the individual's stance \citep{hart2009feeling}. One of the potential sources of this selectivity is \emph{information avoidance} to prevent exposure to contrary evidence \citep{sweeny2010information}. 

A further important feature of what Nickerson labels ``use of evidence'' in confirmation bias is a \emph{backfire effect}. This is where encountering contradictory evidence causes people to become entrenched and increase support for their original stance. Conflicting evidence is not simply ignored -- rather it perversely has the same effect as confirmatory evidence \citep[Section 3.B]{rabin1998psychology}. This backfire effect has been documented in a range of contexts. Examples include support for the death penalty \citep{lord1979}, nuclear weapons \citep{plous1991}, tax cuts, the Iraq war \citep{nyhan2010corrections}, and climate change mitigation policies \citep{hart2012boomerang}.

\textbf{Behavioural economics.} In the economics literature, \citet{rabinschrag1999} formulate a model of how confirmation bias affects individual decision-making. In their set-up, there are two states of the world. Each time a new signal arrives, the agent performs Bayesian updating, with the twist that when a signal runs counter to the agent's current hypothesis then there is a probability $q$ that the agent misinterprets it as actually confirming her hypothesis, i.e., $q$ is the strength of confirmation bias. The main result of their paper is that confirmation bias leads to overconfidence.\footnote{\citet{epstein2006axiomatic} presents a model of non-Bayesian updating, where the agent is `tempted' to change their belief after receiving a signal. This model can nest a version of confirmation bias by choosing an appropriate specification of how the agent is tempted to update their belief.}

While \citet{rabinschrag1999} focus on biased \emph{processing} of information, our paper proposes an alternative and complementary approach to model confirmation bias rooted in the notion of \emph{selective exposure} in psychology -- confirmation bias reduces the range of opinions an individual is willing to listen to. This translates to ignoring information in a different, and complementary, way to the single agent in \cite{rabinschrag1999}. In our framework, confirmation bias stems from selective exposure in the process of information acquisition due to information avoidance. In \cite{rabinschrag1999}, instead, confirmation bias distorts the belief formation process directly leading to biased processing of information.\footnote{The Online Appendix shows that our key result generalizes to a model where agents pay \emph{less} attention to neighbours whose beliefs are further away than a threshold, rather than completely removing the link.} An advantage of our approach is its tractability to analyse the effects of confirmation bias in the context of information diffusion with multiple agents who learn from each other through their social connections. This is arguably becoming more relevant, as individuals increasingly learn by sharing information on social media rather than individually processing information from a media source.

\textbf{Social learning.} A large branch of the literature assumes that agents are non-Bayesian and use a behavioural rule to learn from others.\footnote{Another branch focuses on Bayesian agents. Seminal contributions from \citet{banerjee1992simple} and \citet{bikhchandani1992theory} have agents move in sequence. Subsequent work has embedded agents in a network and relaxed the sequential moves assumption.  
See \citet{golub2016networks} for a comprehensive review.}
\citet{degroot1974} proposed the simple rule that agents update their beliefs by taking a weighted average of their neighbours' opinions, and showed that the process reaches a consensus under mild regularity conditions. This set-up gained traction in economics with \citet{demarzo2003persuasion}, who obtain novel results on convergence speed and relate each agent's contribution to the consensus to their respective network position. More recently, \citet{golub2010naive} give new results on network structures that lead to society correctly aggregating information.\footnote{Other papers using the DeGroot framework include \citet{acemoglu2010spread}, \citet{jadbabaie2012}, \citet{gallo2014social} and \cite{vohra2020strategic}.} Experimental evidence shows that this DeGroot learning rule is a good predictor of how people learn from others, and it is particularly suitable to model learning in a network where it is unrealistic to assume individuals are Bayesian updating and process information conditioning on the network structure.\footnote{See, e.g., \citet{corazzini2012influential} and \citet{chandrasekhar2020testing}. Section 17.2.5 in \citet{choi2016networks} reviews experimental evidence on social learning.}  

The primary contribution of our paper to the social learning literature is to study how confirmation bias affects the outcomes of the learning process, including convergence to a consensus, speed of learning, and the influence of agents. To our knowledge, this is among the first papers to examine how a psychological bias affects DeGroot-type learning.\footnote{\cite{fernandes2023confirmation} also examines confirmation bias in a network setting, but considers bias in interpreting binary public signals and only examines the long-run consensus. In contrast, our work examines bias in the \emph{social learning} component and can generate precise results on learning speed, among other features.} Aside from its intrinsic interest and applications, it also provides a check on whether existing DeGroot-type learning results are robust to a bias that is ubiquitous in reality. 
The objective is analogous to \citet{golub2012homophily}: they examine how homophily, a ubiquitous feature of social networks, affects the speed of social learning and consensus. 

\textbf{Bounded confidence.} Seminal contributions by \cite{hegselmann2002opinion} and \cite{deffuant2000mixing} generated an extensive literature on ``bounded confidence'' in physics that examines learning models where agents choose to ignore others who are too different from themselves. The basic set-up is that, at any point in time, agents ignore information from neighbours whose belief is further away from their own than a threshold. As a consequence, the network varies over time because an agent will start listening again to previously ignored neighbours once/if the difference in their beliefs falls below the threshold. This is a non-homogeneous Markov chain, so tractability is a challenge and the only known analytic results are about the conditions under which there is a consensus. The majority of contributions are simulation-based studies to investigate specific types of networks such as the complete network \citep{hegselmann2002opinion}, the star network \citep{fortunato2005consensus}, the square lattice \citep{weeks2017online} and Erd\H{o}s-R\'{e}nyi random graphs \citep{fortunato2004universality}.\footnote{The literature is extensive, and a review is beyond the scope of this paper. See \cite{lorenz2007continuous}, \citet[pages 608-613]{castellano2009statistical}, and \cite{mastroeni2019agent} for reviews. There is also a related strand of the literature which considers formation of new links based on agents' opinions. See \cite{holme2006nonequilibrium} for an example, and \cite{gross2008adaptive} for a review.}

The key difference in our set-up is that agents only cut links before learning starts, rather than changing them as the learning process unfolds. Therefore the process is described by a homogeneous Markov chain. This makes it tractable and allows us to prove novel analytic results about the speed of convergence. Compared to simulation-based studies, an analytic approach allows us to prove results for a large class of networks as well as an extensive coverage of the parameter space. Moreover, we are not restricted to a single distribution of initial beliefs or a given number of nodes, nor do we have to specify a particular set of link weights. Additionally, we propose a behavioural foundation -- confirmation bias -- for the link-cutting process, along with applications of economic interest.

\textbf{Voting and media.} In recent years political experts have been surprised by several electoral outcomes, including the first election of Donald Trump in the United States, and the outcome of the UK's Brexit referendum. It may not be entirely coincidental that these events occurred alongside a shift in news consumption from traditional media outlets to online social networks \citep{gottfried2016news}. For instance, an extensive study of 10.1 million US Facebook users by \citet{bakshy2015exposure} shows that people tend to predominantly share news with friends that is in line with the recipient's ideology, and this filtering by friends is more powerful than Facebook's algorithmic selection on the news feed to limit exposure to distant viewpoints.

Alongside this is a growing literature on `fake news' and its potential impact on elections. \citet{allcott2017social} provide an overview of some recent work, and also sketch a model of fake news. They suggest that if agents have a preference for confirmatory news reporting, then news reporting can become distorted, possibly reducing the ability of democracies to choose high-quality candidates. Our work complements this literature by showing that confirmation bias can induce a society to elect the ``wrong candidate'' in a two-candidate race (in the sense that confirmation bias changes the outcome of the race).

The second application adds to the large literature that examines media markets through the Hotelling model \citep{hotelling1929}. Existing literature suggests that greater competition allows media to (collectively) express a wider range of views \citep{eaton1975}, but that advertising can pull them towards the centre ground \citep{behringer2015hotelling}. We show that confirmation bias in the potential audience encourages some media organizations to become more extreme. This provides a possible explanation for increasing polarization in the media \citep{prior2013media}, even for organizations that rely on advertising revenue.

This paper shows how confirmation bias' impact on the way we learn from others can lead to surprising election outcomes. In particular, confirmation bias prevents us from directly learning from other people whose information conflicts with our own views. This means that our immediate friends are unrepresentative of society as a whole. Therefore, we can be swayed one way in the medium term, even though the weight of information on aggregate points in the other direction. A consequence of this in the short/medium term is that a society may vote for a candidate that would not be supported by the majority once long-run information aggregation has occurred. 

\section{Model}\label{model}
\paragraph{Agents and the network.} There are $n$ agents, with typical agent $i \in N$. Agents are embedded in a weighted and directed network, represented by an $n \times n$ non-negative matrix $T$. In an abuse of terminology, we will often refer to $T$ as `the network'. We assume that the network $T$ is: (1) \emph{row-stochastic} ($\sum_{j\in N} T_{i j}=1$ for all $i$), (2) \textit{strongly connected} (there is a directed path from $i$ to $j$ for all $i,j \in N$) and (3) \emph{aperiodic} (the greatest common divisor of all cycle lengths is 1).\footnote{A directed path of length $l$ between $i$ and $j$ is a sequence of links $T_{i k_1} > 0,...,T_{k_{l-1} j} > 0$ such that no two nodes on the path are the same.} 
We say $i$ \emph{listens to} $j$ if $T_{ij}>0$.

Each agent is endowed with a signal $\theta_i \in [0,1]$. All signals are of equal precision and each signal is a random draw from some distribution $F$. Following the standard assumption in the literature, agent $i$'s initial belief is equal to her signal: i.e., $x_{i 0} = \theta_i$.

\paragraph{Learning.} In each period $t\geq1$ each agent updates her belief by taking a weighted average of her current belief and the beliefs of agents she listens to. Formally, $\mathbf{x}_{t} = T \ \mathbf{x}_{t-1}$. Iterating, we have that $\mathbf{x}_{t} = T^{t} \ \mathbf{x}_{0}$ so we can derive agents' beliefs at time $t$ from the initial signals and the network.

\paragraph{Confirmation bias.} In the first step of the learning process, agents truthfully share their signals and therefore an agent learns the initial belief of their neighbours. We assume that when an agent with confirmation bias learns that the difference between a neighbour's belief and her own belief exceeds a threshold, she ignores that neighbour and transfers the weight she would have put on that neighbour's belief to her own belief. Moreover, she never listens again to information from that neighbour for the rest of the learning process.

\begin{defn}\label{core rule}
A society on a network $T$ in which agents have confirmation bias $q$ communicates according to a network $T^*$ such that:
\begin{eqnarray*}
  \text{if} \quad |{x_{i 0} - x_{j 0}}| > (1-q) & \text{then}  & T^{*}_{i j}  = 0  \qquad	T^{*}_{i i} = T_{i i} + \sum_{j : |{x_{i 0} - x_{j 0}}| > (1-q)}  T_{i j} \\
   & \text{otherwise}  & T^{*}_{i j}  = T_{i j}  	\quad	T^{*}_{i i} = T_{i i}
\end{eqnarray*}
\end{defn}

Mathematically, a society in which agents have confirmation bias $q$ communicates through a network $T^*$ that has had links cut compared to $T$, and where the weight of the links that have been cut is redistributed to self-links. Notice that the threshold is defined as $1-q$ so that higher $q$ corresponds to increasing confirmation bias. Understanding the impact of confirmation bias in the model means comparing the learning processes on $T$ and $T^*$. 

\paragraph{Discussion.} Three assumptions embedded in how we model confirmation bias are worth discussing upfront.

\begin{enumerate}
  \item Agents \emph{completely} cut links. This is a realistic assumption for online social networks where users have to make a binary decision on whether to follow or unfollow someone. The results are, however, robust to relaxing this assumption. 
  \item Agents \emph{redistribute} weights of cut links \emph{to themselves}. This is an important feature of confirmation bias from a large body of experimental evidence \citep{nickerson1998confirmation} -- the ``backfire effect'' mentioned in the opening quote. 
  \item Agents \emph{never reinstate} links they have cut. This is consistent with the standard assumption of myopic learning in the DeGroot model, in which the weights an agent gives to different neighbours are fixed at the beginning and never change over time. It also helps in terms of tractability because time-varying weights would significantly increase the complexity of the Markov process.
\end{enumerate}

This means that our modelling of confirmation bias is one of biased exposure to information, rather than of biased processing. Agents refuse to expose themselves to conflicting information (through the process of cutting links), but process the information they are willing to listen to in the usual DeGroot fashion. The additional feature that agents re-route all cut links to themselves captures the notion of doubling down on one's opinion in the face of conflicting information that is a key part of the ``backfire effect''.\footnote{The ``backfire effect'' in our modelling of confirmation bias can be interpreted as a generalization of the notion of stubbornness in  \cite{yildiz2013binary}. In extremis, confirmation bias can make agents ``stubborn'' \`a la \cite{yildiz2013binary} if it induces them to cut all of their links. A key difference is that a stubborn agent is not influenced by the
learning process, while an agent with confirmation bias both influences and is influenced by the evolution of
learning.} 

Confirmation bias in our model has two components: link-cutting and re-routing links to oneself. The link-cutting part generates homophily on initial beliefs -- as links are cut with those who have the most different initial beliefs. But link re-routing -- the backfire effect which is critical to our results -- in itself need not increase homophily. For example, homophily would be higher if agents re-routed links to those most similar to themselves. So in our set-up, confirmation bias creates homophily through selective exposure to people with initially similar beliefs, and also makes people more entrenched in their own beliefs when they are initially exposed to strongly conflicting information. Further, to keep the model tractable, we have abstracted from further link-cutting, or the re-formation of links, after the learning process begins.

\section{How confirmation bias affects learning}\label{main results}
This section examines how confirmation bias affects social learning. Section \ref{consensus} presents results on how confirmation bias affects the emergence of consensus. Section \ref{speed} contains our main result on how confirmation bias varies the speed of the learning process. Finally, section \ref{influence} investigates how it changes who is influential.

\subsection{Consensus} \label{consensus}
The assumption that $T$ is both aperiodic and strongly connected is sufficient to guarantee that, absent any confirmation bias, all agents converge to a consensus where they share the same belief \citep{demarzo2003persuasion}. The following is a formal definition of consensus.

\begin{defn}
    There is a \emph{consensus} if and only if $\lim_{t \to \infty} x_{it} = x$ for some $x \in [0,1]$ and for all $i$.
\end{defn}

Given that confirmation bias causes agents to stop listening to some others, an immediate question is whether a society with confirmation bias converges to a consensus. It turns out that this depends on the strength of confirmation bias. There is a critical threshold for the value of $q$ which determines whether there is a consensus. This threshold in turn depends on both the initial network $T$ and on the initial beliefs $x_0$.

\begin{prop}\label{critical q}
For any network $T$, and set of initial beliefs $x_0 \in [0,1]^{n}$, there exists a critical value $\overline{q}(T,x_0)$ such that there is a consensus if and only if $q < \overline{q}(T,x_0)$.
\end{prop}

The intuition is that above the threshold confirmation bias re-routes too many links onto self-links, and this breaks the strong connectedness of the network because some agents become isolated. Increasing confirmation bias further can only re-route more links, and can never re-connect these agents. 

\subsection{Speed of learning} \label{speed}
Having established whether a society converges to a consensus, we now examine how long it takes to do so. Crucially, we show it takes longer when confirmation bias is stronger. To do this formally, we need a measure of `consensus time' that captures how long it takes to reach a consensus. The Markov chain literature contains a number of different definitions of consensus time. In this paper, we adopt the following definition that takes into account the full spectrum of eigenvalues of the matrix $T$.\footnote{In general, consensus time depends on the full spectrum of eigenvalues. Many authors approximate it using the second largest eigenvalue modulus \citep{montenegro2006}.}

\begin{defn} \label{def:time}
If $q<\overline{q}(T,x_0)$, the \emph{average consensus time} $\tau(\epsilon)$ is:
\begin{align*}
\tau(\epsilon) = \min \left\{ t > 0 : \frac{1}{n} \sum_{i} || T^{t}_{(i, \cdot)}(x_0 , q) - T^{\infty}_{(i, \cdot)}(x_0 , q) ||_{2}^{2} < \epsilon \right\}
\end{align*}
for any $\epsilon > 0$, where $||.||_{2}$ denotes the $\ell_{2}$ norm. Otherwise, $\tau(\epsilon) = \infty$.
\end{defn}

Average consensus time measures how long it takes for the average agent to get ``close'' to their long-run listening structure -- the weight they put on each agent's belief in the long run when consensus has been reached. By ``close'', we mean that the $\ell_2$ distance between their current listening structure and their long-run structure is below a fixed value $\epsilon$. The initial beliefs enter the definition only because they change the network $T$ through the effect of confirmation bias. Aside from their impact on the initial network, average consensus time is agnostic to the actual set of initial beliefs.\footnote{Existing measures of convergence time in the mathematics literature follow this approach. See, e.g., \cite{montenegro2006, part3notes,newpart3notes, levin2009}.} It captures how quickly information is able to spread through the network, and it is therefore a good approximation of how quickly actual consensus of beliefs is achieved when information is widely dispersed through the network. These are the cases where social learning is important in practice because the spread of information is key to achieving a consensus. A useful benchmark case for this is when the initial beliefs are uniformly distributed. In this case, information is very widely distributed and the evolution of beliefs closely tracks the evolution of our measure. We validate this approximation in a simulation study in section \ref{simulations results}.
	
In contrast, our metric can diverge significantly from the time it takes for a consensus of beliefs to be achieved in cases where the spread of information through the network is not the primary driver for learning. These are cases where there is little to be learned from others because the long-run consensus can be reached without aggregating the beliefs of other agents. A benchmark case is where all agents start with the same initial belief. Here, learning from others has no impact, so the time taken to achieve consensus of beliefs is zero. However, even in this case, average consensus time still measures how quickly information can spread through society. Average consensus time will not always closely approximate the time taken to reach a consensus of beliefs, but the cases where it performs poorly are precisely the cases where social learning is not that meaningful.

Our result that confirmation bias always weakly increases consensus time applies subject to two conditions: (i) agents listen mostly to themselves, and (ii) the network is \emph{reversible}. The following is a formal definition of the large class of networks that are reversible.

\begin{defn} \label{defn influence}
    A network is \emph{reversible} if and only if $s_i T_{ij} = s_j T_{ji}$ for all $i,j \in N$, where $s_{i}$ is the $i$\textsuperscript{th} entry in the (unique) left-hand unit eigenvector whose entries sum to 1 associated with the first eigenvalue $\lambda_{1} \equiv 1$: $s \ \lambda_{1} = s \ T$. 
    $s_i$ is called $i$'s \emph{influence}.
\end{defn}

Symmetric networks are a special case of reversible networks because when a network is symmetric (i.e., $T_{i j}=T_{j i}$ for all $i,j$) then all agents have the same influence (i.e., $s_i = s_j$ for all $i,j$). With this definition in place, we can now state our main result formally.

\begin{thm}\label{speed result}
Suppose $T$ is reversible and $T_{ii} \geq \frac{1}{2}$ for all $i$. For any $\epsilon > 0$, the average consensus time $\tau(\epsilon)$ is (weakly) monotonically increasing in the amount of confirmation bias $q$.
\end{thm}

This result formalizes the intuitive insight that confirmation bias slows down how long it takes to get to any given level of agreement. The key driving force is that confirmation bias makes agents more reluctant to change their beliefs. Mathematically, the proof to Theorem \ref{speed result} shows that the average consensus time increases when link weight is moved from the between-agent links to agents' self-links,\footnote{Formally, this is moving weight from the off-diagonal elements to the diagonal elements of the matrix $T$.} and there are \emph{no other changes} to the network. It does not rely on imposing any additional structure to the changes to the network. The statement of Theorem \ref{speed result} is proved as a special case, and therefore our core result may be applicable to a range of other settings.

Using our definition for how quickly agents reach a consensus, $\tau(\epsilon)$, we are able to obtain a precise result about how it changes with confirmation bias. Other authors have opted to use a metric that directly captures the evolution of beliefs, at the cost of obtaining approximate results. Typically, this involves approximations based on the second largest eigenvalue modulus (SLEM).\footnote{When eigenvalues are arranged in order of size, $1 = \lambda_1 \geq \lambda_2 \geq ... \geq \lambda_n$, the SLEM is equal to $\max \{ | \lambda_2 | , | \lambda_n | \}$. Note that some authors in the economics literature refer to this, somewhat imprecisely, as the `second eigenvalue'. The second largest eigenvalue need not always be the same as the SLEM.}  Prominently, \citet[Definition 2, Theorem 1]{golub2012homophily} use a definition of convergence based on the \emph{worst-case} assignment of initial beliefs, and link this approximately to the SLEM. Similarly, \citet[Theorem 5]{demarzo2003persuasion} show that the distance between an agent's belief at time $t$ and the long-run consensus is \emph{on the order of} the SLEM to the power of $t-1$. The details of the proof to Theorem \ref{speed result} show that the SLEM is monotonically increasing in $q$. Therefore, our result is robust to adopting the approaches used by \cite{golub2012homophily} and \cite{demarzo2003persuasion}.

There are two conditions that are needed for the result to hold. First, the restriction to the large class of reversible networks, which captures many types of real social networks including trust, friendship and family networks. While weights can be asymmetric (i.e., $T_{ij} \neq T_{ji}$), all links must be reciprocated (i.e., $T_{ij} > 0 \iff T_{ji} > 0$). This does, however, exclude some prominent types of social networks, such as X/Twitter. The restriction to reversible networks is necessary for tractability because it ensures that the removal of links does not change agents' influence. This simplifies the analysis by disentangling the effect of confirmation bias on convergence time from its effect on the distribution of influence. Simulations in section \ref{simulations results} show that the result is largely robust to relaxing network reversibility.

Second, the restriction that $T_{ii} \geq \frac{1}{2}$ for all $i$ requires that agents exhibit some inertia and cannot give up on their beliefs too quickly. Intuitively, one could think of periods in the model as corresponding to relatively short periods of chronological time. In the context described by the model, it is reasonable to assume that agents do not have large swings in their belief from one period to the next. From a technical standpoint, the restriction 
ensures the Markov chain has no negative eigenvalues.\footnote{This is a standard approach in the mathematics literature. See, e.g., \cite{levin2009, mcnew2011eigenvalue, basu2014characterization}. This assumption is sufficient but not necessary. The Markov chain needs only to be positive semi-definite.} It means that an increase in each eigenvalue, as a result of confirmation bias, leads to an increase in the sum of \emph{squared} eigenvalues, which governs the average consensus time $\tau(\epsilon)$. In other words, it guarantees that we avoid situations where the Markov chain is ``nearly'' periodic. 

Our main result in Theorem \ref{speed result} is \emph{not} driven solely by the fact that confirmation bias removes links, making the network sparser. That reducing the number of links by itself slows down the learning process is a tempting, but not universally correct, intuition. The final part of this section presents a counterexample to show that fewer links alone can in fact speed up the learning process.

Consider an extended model that nests confirmation bias as a special case. As in the current set-up, $i$ removes a link with $j$ if $|{x_{i 0} - x_{j 0}}| > 1-q$. The extension is that $i$ re-routes a fraction $\phi \in [0,1]$ of the severed link to herself, and spreads out a fraction $1- \phi$ across her remaining links in proportion to the strength of each link. Clearly, $\phi=1$ is the special case in which the re-routing of links is determined by confirmation bias.

Now consider the network $T$ and allocation of initial beliefs in Figure \ref{figa:speed}. If $q\in (0, 0.3]$ and $\phi=1$, then the resulting network is $T^*$ in Figure \ref{figb:speed}. Assuming $\epsilon = 10^{-4}$, the convergence time in $T^*$ is 80 periods, which is longer than the 25 periods in $T$ when there is no confirmation bias, as expected from Theorem \ref{speed result}. However, if $\phi=0$ then the resulting network $T^{*}_{0}$ is displayed in Figure \ref{figc:speed} and convergence time is 7 periods. Notice that $T^{*}_{0}$ is \emph{sparser} than the initial network $T$, but convergence in $T^{*}_{0}$ is \emph{faster} than in $T$. It turns out that convergence is faster for most values of $\phi$. Figure \ref{figd:speed} displays network $T^{*}_{0.75}$, for which the convergence time is the same as in $T$. Because convergence time is measured in whole periods, this is true for all $\phi$ between roughly $0.74$ and $0.75$: for any $\phi > 0.75$ the convergence time is slower than in $T$, while for any $\phi \leq 0.73$ it is strictly faster, even though all these networks are sparser than $T$.

\newcommand{\figHeight}{7.2cm}
\newcommand{\figWidth}{5.4cm}

\begin{figure}[ht!]
\caption{}\label{fig:speed}
\captionsetup[subfigure]{oneside,margin={1cm,0cm}}
\begin{subfigure}[t]{0.5\linewidth}
\centering
\caption{Network $T$. Convergence takes $25$ periods.}\label{figa:speed}
\resizebox{\figHeight}{\figWidth}{ \fbox{
\begin{tikzpicture}[->,>=stealth',shorten >=1pt,auto,node distance=4cm,semithick, scale=0.5]
  \tikzstyle{every state}=[fill=white, draw, text=black, circle]
  \node[state, label = {$x_{A0} = 0$} ] 		(A)                {$A$};
  \node[state, draw=none]						(dummy) [right = 2cm of A] {} ;
  \node[state, label = {$x_{B0} = 1$} ] 		(B) [right =5cm of A]  {$B$};
  \node[state, label =left:{$x_{C0} = 0.7$} ] 	(C) [below =4cm of  dummy] {$C$};

  \path (A) edge [loop left]		node {0.55} (A)
  			edge [bend left=10]		node {0.4} (B)
  			edge [bend left=10]		node {0.05} (C)
        (B) edge [loop right]		node {0.55} (B)
        	edge [bend left=10]		node {0.4} (A)
            edge [bend left=10]  	node {0.05} (C)
        (C) edge [loop right]		node {0.9} (C)
        	edge [bend left=10]		node {0.05} (A)
            edge [bend left=10] 	node {0.05} (B);  
\end{tikzpicture}  }  }
\end{subfigure}
\hfill
\begin{subfigure}[t]{0.5\linewidth}
\centering
\caption{Network $T^{*}$, $\phi = 1$. Convergence takes $80$ periods.}\label{figb:speed}
\resizebox{\figHeight}{\figWidth}{ \fbox{
\begin{tikzpicture}[->,>=stealth',shorten >=1pt,auto,node distance=4cm,semithick, scale=0.5]
  \tikzstyle{every state}=[fill=white, draw, text=black, circle]
  \node[state, label = {$x_{A0} = 0$} ] 		(A)                {$A$};
  \node[state, draw=none]						(dummy) [right = 2cm of A] {} ;
  \node[state, label = {$x_{B0} = 1$} ] 		(B) [right =5cm of A]  {$B$};
  \node[state, label =left:{$x_{C0} = 0.7$} ] 	(C) [below =4cm of  dummy] {$C$};

  \path (A) edge [loop left]		node {0.95} (A)
  			edge [bend left=10]		node {0.05} (C)
        (B) edge [loop right]		node {0.95} (B)
            edge [bend left=10]  	node {0.05} (C)
        (C) edge [loop right]		node {0.9} (C)
        	edge [bend left=10]		node {0.05} (A)
            edge [bend left=10] 	node {0.05} (B);  
\end{tikzpicture}  }  }
\end{subfigure}

\vspace*{5mm}

\begin{subfigure}[t]{0.5\linewidth}
\centering
\caption{Network $T^{*}_{0}$, $\phi = 0$. Convergence takes $7$ periods.}\label{figc:speed}
\resizebox{\figHeight}{\figWidth}{  \fbox{
\begin{tikzpicture}[->,>=stealth',shorten >=1pt,auto,node distance=5cm,semithick,scale=0.5]
  \tikzstyle{every state}=[fill=white, draw, text=black, circle]
  \node[state, label = {$x_{A0} = 0$} ] 		(A)                {$A$};
  \node[state, draw=none]						(dummy) [right = 2cm of A] {} ;
  \node[state, label = {$x_{B0} = 1$} ] 		(B) [right =5cm of A]  {$B$};
  \node[state, label =left:{$x_{C0} = 0.7$} ] 	(C) [below =4cm of  dummy] {$C$};

  \path (A) edge [loop left]		node {0.55} (A)
  			edge [bend left=10]		node {0.45} (C)
        (B) edge [loop right]		node {0.55} (B)
            edge [bend left=10]  	node {0.45} (C)
        (C) edge [loop right]		node {0.9} (C)
        		edge [bend left=10]	node {0.05} (A)
            edge [bend left=10] 	node {0.05} (B);  
\end{tikzpicture}  }  }

\end{subfigure}
\hfill
\begin{subfigure}[t]{0.5\linewidth}
\centering
\caption{Network $T^{*}_{0.75}$, $\phi = 0.75$. Convergence takes $25$ periods.}\label{figd:speed}
\resizebox{\figHeight}{\figWidth}{  \fbox{
\begin{tikzpicture}[->,>=stealth',shorten >=1pt,auto,node distance=5cm,semithick, scale=0.5]
  \tikzstyle{every state}=[fill=white, draw, text=black, circle]
  \node[state, label = {$x = 0$} ] 		(A)                {$A$};
  \node[state, draw=none]						(dummy) [right = 2cm of A] {} ;
  \node[state, label = {$x = 1$} ] 		(B) [right =5cm of A]  {$B$};
  \node[state, label =left:{$x = 0.7$} ] 	(C) [below =4cm of  dummy] {$C$};

  \path (A) edge [loop left]		node {0.85} (A)
  			edge [bend left=10]		node {0.15} (C)
        (B) edge [loop right]		node {0.85} (B)
            edge [bend left=10]  	node {0.15} (C)
        (C) edge [loop right]		node {0.9} (C)
        		edge [bend left=10]	node {0.05} (A)
            edge [bend left=10] 	node {0.05} (B);  
\end{tikzpicture}  }  }
\end{subfigure}
\end{figure}


\subsection{Influence} \label{influence}

A single Bayesian agent who aggregates all the information in a society would weight each initial signal equally in their posterior. However, when the society learns through DeGroot social learning, the weight that an agent's initial signal has in the final consensus depends on that agent's network position. Confirmation bias alters the network structure, so it may affect an agent's influence on the consensus. Agents who cut links put more weight on their initial signal and less on others' beliefs. Intuitively, we might expect that they have greater influence on the final consensus. This section examines how confirmation bias changes and redistributes influence, and shows that this intuition does not follow through.

Our first result identifies a class of networks in which confirmation bias does not affect agents' influence on the final consensus.

\begin{rem} \label{equal influence no change}
If $T$ is reversible and $q < \overline{q}(T,x_0)$, confirmation bias does not alter any agent's influence.
\end{rem}

In a reversible network, all links are reciprocated, so if $i$ listens to $j$, then $j$ listens to $i$.  Given that the level of confirmation bias $q$ is the same for all agents, if $i$ cuts the $T_{ij}$ link then $j$ will also cut the $T_{ji}$ link. When $i$ cuts her link with $j$, she listens more to herself. This gives her belief at any given period more weight in the next period, which pushes her influence up. However, $j$ also cuts his link with $i$ so others listen to $i$ less, which pushes her influence down. These two forces are perfectly offsetting in a reversible network because $i$ and $j$'s relative influence is the inverse of the ratio of how much they listen to each other. For example, if $i$ is twice as influential as $j$, then she must listen to $j$ twice as much as $j$ listens to her. Consequently, link removals due to confirmation bias in a reversible network do not affect any agent's influence. Additionally, when $T$ is reversible, $T^*$ will also be reversible because no agent's influence changes, and the weights of the links that are not cut remain unchanged.\footnote{In the special case of symmetric networks, agents will also all have the \emph{same} influence. This is because the first left-hand eigenvector of any symmetric $n \times n$ Markov chain is a vector with all entries equal to $1 / n$.} 

Dropping the reversibility assumption makes it challenging to characterize how confirmation bias redistributes influence. The following remark, however, shows that in the special case when there is just a single agent who cuts links due to confirmation bias, their influence unambiguously rises.

\begin{rem}\label{single agent influence}
If exactly one agent $i$ cuts one or more links, then their influence, $s_{i}$, strictly increases.
\end{rem}

An appealing conjecture is that this extends to a situation where many agents cut links. Unfortunately, the following remark shows that this is not the case.

\begin{rem}\label{multiple agent influence}
If more than one agent cuts at least one link, then their influences can either rise or fall.
\end{rem}

\begin{figure}[ht]
\caption{Original network $T$ on the left and resulting network $T^*$ with confirmation bias in the $(0.45,0.7]$ range.}\label{fig:influence}
\begin{subfigure}[t]{0.45\linewidth}
\resizebox{8cm}{7cm}{ \fbox{
\begin{tikzpicture}[->,>=stealth',shorten >=1pt,auto,node distance=5cm,semithick]
  \tikzstyle{every state}=[fill=white, draw, text=black, circle]
  \node[state, label = {$x = 0.2$} ] 		(A)                {$A$};
  \node[state, label = {$x = 0.5$} ] 		(B) [right of =A]  {$B$};
  \node[state, label =below:{$x = 0.75$} ] 	(C) [below of =B]  {$C$};
  \node[state, label =below:{$x = 0.9$} ] 	(D) [below of =A]  {$D$};

  \path (A) 	edge [bend left=10]	node {0.55} (B)
  			edge             	node {0.25} (C)
  			edge [bend left=10]	node {0.2}  (D)
        (B) edge [bend left=10]	node {0.8} (A)
            edge [bend left=10]  node {0.2} (C)
        (C) edge [bend left=10]	node {0.7} (B)
            edge [bend left=10] 	node {0.3} (D)
        (D) edge [bend left=10]	node {0.7} (A)
            edge [bend left=10]	node {0.3} (C);  
\end{tikzpicture} } }
\vspace{3mm}
\caption{Network $T$: \\ $ s = (0.352 , 0.328 , 0.192 , 0.128)^{\prime}$}
\end{subfigure}
\hfill
\begin{subfigure}[t]{0.45\linewidth}
\resizebox{8cm}{7cm}{ \fbox{
\begin{tikzpicture}[->,>=stealth',shorten >=1pt,auto,node distance=5cm,semithick]
  \tikzstyle{every state}=[fill=white, draw, text=black, circle]
  \node[state, label = {$x_{A0} = 0.2$} ] 		(A)                {$A$};
  \node[state, label = {$x_{B0} = 0.5$} ] 		(B) [right of =A]  {$B$};
  \node[state, label =below:{$x_{C0} = 0.75$} ] 	(C) [below of =B]  {$C$};
  \node[state, label =below:{$x_{D0} = 0.9$} ] 	(D) [below of =A]  {$D$};

  \path (A) 	edge [bend left=10]	node {0.55} (B)
  			edge [loop left] 	node {0.45} (A)
        (B) edge [bend left=10]	node {0.8} (A)
            edge [bend left=10]  node {0.2} (C)
        (C) edge [bend left=10]	node {0.7} (B)
            edge [bend left=10] 	node {0.3} (D)
        (D) edge [loop left]		node {0.7} (D)
            edge [bend left=10]	node {0.3} (C);
\end{tikzpicture} } }
\vspace{3mm}
\caption{Network $T^{*}$: \\ $s^{*} = (0.481 , 0.330 , 0.094 , 0.094)^{\prime}$}
\end{subfigure}
\end{figure}

Figure \ref{fig:influence} demonstrates this result. When the level of confirmation bias is in the range $q\in(0.45, 0.7]$, and initial beliefs are $(0.2, 0.5, 0.75, 0.9)$, the listening structure $T$ on the left changes to $T^*$ on the right because $A$ removes links with $C$ and $D$, and $D$ removes the link with $A$. Among these agents who cut links, $A$'s influence rises in $T^*$, but $D$'s influence \emph{decreases} in $T^*$ despite the fact she has removed a link. The reason for this decrease is that in $T^*$ the most influential agent $A$ does not listen to $D$ any more, and therefore $D$ does not directly affect the beliefs of the most influential agent. Conversely, $B$ does not cut any links but her influence \emph{rises} because she is the only agent in $T^*$ that the influential agent $A$ still listens to. The takeaway message is that understanding the change in an agent $i$'s influence involves keeping track of the change in influence of the agents who listen to $i$, in addition to whether $i$ cuts any links.

\section{Shock elections} \label{voting section}
The previous section has shown that confirmation bias affects the process of learning from others by making it slower and by redistributing individuals' influence. In the past decade, the advent of social media has arguably increased the visibility and weight that information we learn from others has on our views, and this has become increasingly relevant in the context of elections (see, e.g., \citet{kohut2008social}, \citet{braha2017voting} and \citet{weeks2017online}). In order to examine how the impact of confirmation bias on learning affects electoral outcomes, we embed the learning framework in section \ref{model} into a basic voting model. The objective is to understand whether ``shock'' elections are more likely in a world where learning is affected by confirmation bias.

\paragraph{Candidates and voters.} We assume there are two candidates $Y = \{\text{Left}, \text{Right}\}$ whose belief, or \emph{ideology}, is fixed at $x_{Y=\text{Left}} = 0$ and $x_{Y=\text{Right}} = 1$. There are $n$ voters and each voter has an initial belief $x_{i 0} \in \{x_{EL}, x_{CL}, x_S, x_{CR}, x_{ER}\}$ with $x_{EL}<x_{CL}<x_S=\frac{1}{2}<x_{CR}<x_{ER}$. Denote by $f_m$ the fraction of voters assigned initial belief $x_m$, with $0<f_m <1$ and $\sum_{m=EL}^{ER}f_m=1$. Initially, the spectrum of voters' preferences spans ``Extreme Left'' (EL), ``Centre Left'' (CL), ``Swing voters'' (S), ``Centre Right'' (CR), and ``Extreme Right'' (ER).

\paragraph{Learning and voting.} Voters communicate through a network $T$ according to the model in section \ref{model}. Each voter $i$'s utility function $U_{it} = u(|x_{it} - x_{y}|)$ is strictly decreasing in the distance between their own belief (at the time voting takes place) and the ideology of the winning candidate. The winner of an election at time $t$ is determined by a simple majority -- the candidate with the most votes wins and a tie is resolved by a coin toss. This set-up implies that sincere voting is a weakly dominant strategy by application of the standard Median Voter Theorem result.\footnote{In particular, our set-up satisfies the following conditions: (a) there are only two candidates; (b) candidates are chosen by a majority vote; (c) voters have single-peaked preferences; and (d) voting takes place along a single dimension. See \citet[Chapter 21.D]{mas1995microeconomic} for a review of the Median Voter Theorem.} Thus, a voter $i$ facing an election at time $t$ casts a vote $v_{i,t}$ according to the following strategy:
\begin{align*}
&v_{i,t}(x_{i,t}) =
	\begin{cases}
	0		& \text{if } x_{i,t} < 0.5 \\
   	1 		& \text{if } x_{i,t} > 0.5 \\
	\xi 	& \text{if } x_{i,t} = 0.5 \hspace{3mm} \text{with} \hspace{3mm} P(\xi = 0) = P(\xi = 1) = 0.5.
	\end{cases}
\end{align*}

\paragraph{Simplifying assumptions.}
In order to focus on the interesting case, we assume that society votes for the same candidate both before any learning takes place and after a consensus has been reached. Without loss of generality, we assume that the Left candidate is the winner. We also need to place some assumptions on the network structure in order to smooth out the learning process. First, we assume that (i) when there is no confirmation bias all agents put the same weight on their own belief (i.e., $T_{ii} = \phi$ for all $i$), and that this weight is at least $\frac{1}{n}$. Second, we assume that (ii) agents' neighbours are representative of society as a whole, in the sense that how much an agent listens to others is unrelated to those others' initial beliefs.

\begin{ass}\label{ass:1}
\phantom{2} \\
    (i) Equal self-links: $T_{ii} = \phi \geq \frac{1}{n}$ for all $i$; \\
    (ii) Representative neighbourhood: $cov(x_{j0}, T_{ij}) = 0$ for all $i$, where the covariance is taken over all $j\neq i$; \\
    (iii) Left candidate wins: $f_{EL} + f_{CL}> f_{CR} + f_{ER}$ and $\sum_{m=EL}^{ER} x_{m} f_{m}< 0.5$.
\end{ass}

The representative neighbourhood assumption is restrictive because it may be difficult to satisfy when the number of agents is small and/or there are agents with only a few links. Nevertheless, it provides a useful benchmark for tractability purposes. In Section \ref{simulations results} we relax the above assumptions in a numerical exercise, and show that the key insight is robust.

The assumptions above are designed in part to bake in victory by the Left candidate when there is no confirmation bias, and therefore we dub an election a `shock' when the Right candidate wins.

\begin{defn}
A \emph{shock election} occurs if there exists a time $0<t<\infty$ such that the Right wins the election that occurs at $t$.
\end{defn}

The following proposition shows that confirmation bias makes a shock electoral outcome possible in a society that would not experience a shock election when there is no confirmation bias. An intuitive sufficient condition is that each agent is willing to listen to others in her own group and in adjacent groups, but not in any other groups. This pattern of willingness to listen happens when $q \in (q_1, q_2)$ where $q_1 = 1 - \min_{m \in \{EL, CL, S\}} \{ x_{m+2} - x_{m} \}$ and $q_2 = 1 - \max_{m \in \{EL, CL, S, CR\}} \{ x_{m+1} - x_{m} \}$.\footnote{In an abuse of notation we are using $m+1$ to denote ``one group to the right of group $m$.'' So $EL + 1 \equiv CL$, and $S +2 \equiv ER$.}

\begin{prop}\label{voting result}
Suppose Assumption \ref{ass:1} holds. A shock election can occur with confirmation bias, but never occurs without confirmation bias.

A sufficient set of conditions that guarantee at least one shock election is: (i) $q \in (q_1 , q_2)$, (ii) $f_{EL} + f_{CL} < 0.5 $, and (iii) $f_{CR} (x_{CR} - 0.5) > f_{CL} (0.5 - x_{CL})$.
\end{prop}

This result is similar in spirit to \cite{golub2012homophily}. Intuitively, they show that learning can be non-monotonic -- which creates shock elections -- by assuming a specific type of network structure with dense connections within groups and sparse ones across groups. Our paper proposes an endogenous process to create a similarly segregated network structure. Starting from the structure of initial beliefs, confirmation bias is a behavioural mechanism that can create a segregated network that also leads to non-monotonic learning and therefore shock elections. Because the segregation is generated rather than assumed, the conditions under which a shock election can occur are restrictions on the distribution of initial beliefs, rather than on the network. Given the regularity conditions in Assumption \ref{ass:1}, our results speak to which societies are susceptible to shock elections, based only on the initial beliefs and the strength of confirmation bias. A further implication is that the network structure can be reshaped by interventions that are not themselves about the network: either by shifting people's initial opinions, or by making them more willing to listen to those they disagree with (a fall in $q$). In contrast, \cite{golub2012homophily} treat the network as a primitive. The two results are also complementary in terms of practical relevance -- they show that both confirmation bias and a homophilous network are capable of causing shock elections.

More formally, they show that spectral homophily, which is defined as the second largest eigenvalue modulus of the network, can induce shock elections. At a technical level, \cite{golub2012homophily} restrict initial beliefs to be binary, and define two groups that differ in terms of listening structure and their probability of receiving an accurate signal. In contrast, we define five groups exclusively in terms of their initial belief, which has an intuitive interpretation in the context of a typical political landscape. Beliefs are allowed to vary on the unit line, and we do not make any assumptions on the network structure of different groups and/or the accuracy of the signals they receive. These differences mean that a different proof strategy is required in our setting compared to their paper.

\paragraph{An example.} Figures \ref{fig:voting} and \ref{fig:voting2} present an example of how a shock election can occur with confirmation bias. Consider a society made up of five groups of voters (EL, CL, S, CR and ER) with population shares $f = (0.35, 0.1, 0.2, 0.25, 0.1)$, in which every voter listens equally to every voter in society. The weight each voter puts on each group is then equal to that group's population share, so the network between groups is $T$, represented in Figure \ref{fig:voting}. Given the vector of initial beliefs $x_0$ below, it is straightforward to compute what the learning process converges to.\footnote{We can also construct an example using a single-peaked distribution. For example $f = (0.2, 0.23, 0.25, 0.3, 0.02)$, $x = (0, 0.3, 0.5, 0.7, 1)$ and $q = 0.65$ fits the sufficient conditions in Proposition \ref{voting result}.} Clearly the society votes for the Left candidate at time 0 (45\% of voters are to the left of $0.5$ and only 35\% to the right), at time 1, and every period thereafter. 

Now suppose that the level of confirmation bias is $q=0.78$. Figure \ref{fig:voting2} represents the resulting network $T^*$ after the removal of links. The learning process converges to the same outcome as without confirmation bias, so the Left wins after consensus is reached. However, at time $t=1$ there is a shock election in which the Right would win because the swing voters' belief is $x_{i, t=1}=0.52$ for all $i \in S$, and therefore they vote for the Right. Their belief shifts to the Right because they listen more to centre-right voters than to centre-left ones, and, due to confirmation bias, they stop listening to both groups of extreme voters. In the short term, therefore, there is a majority (55\%) to support the Right formed by the extreme right, centre right, and swing voters.

\begin{gather*}
 T =
   \begin{pmatrix}
	 0.35 & 0.1 & 0.2 & 0.25 & 0.1 \\
	 0.35 & 0.1 & 0.2 & 0.25 & 0.1 \\
	 0.35 & 0.1 & 0.2 & 0.25 & 0.1 \\
	 0.35 & 0.1 & 0.2 & 0.25 & 0.1 \\
	 0.35 & 0.1 & 0.2 & 0.25 & 0.1
   \end{pmatrix}
\quad
x_{0} =
    \begin{pmatrix}
	0.15 \\
	0.3 \\
	0.5 \\
	0.65 \\
	0.75
    \end{pmatrix}
\quad
\implies T \cdot x_{0} =
	\begin{pmatrix}
	0.42 \\
	0.42 \\
	0.42 \\
	0.42 \\
	0.42
	\end{pmatrix}=T^{\infty} \cdot x_{0}
\end{gather*}

\begin{figure}[ht]
\caption{Network $T$, without confirmation bias. Notice that links are drawn with a double arrow to avoid cluttering, but the strength of the link differs depending on the direction and it is indicated in \emph{italics} next to the arrowheads.}\label{fig:voting}
\centering
\fbox{	\resizebox{0.75\columnwidth}{!} {
\begin{tikzpicture}[->,>=stealth',shorten >=1pt,auto,node distance=5cm,semithick]
  \tikzstyle{every state}=[fill=white, draw, text=black, circle]
  \node[state, label={right:$x=0.5$} ]
  (C) 								{$S$};
  \node[state, label={below left:$x=0.3$} ]
  (B) [below left=2cm and 3cm of C]  {$CL$};
  \node[state, label={below:$x=0.15$} ]
  (A) [below right=3cm and 1cm of B] {$EL$};
  \node[state, label={below right:$x=0.65$} ]
  (D) [below right=2cm and 3cm of C]	{$CR$};
  \node[state, label={below:$x=0.75$} ]
  (E) [below left= 3cm and 1cm of D] 	{$ER$};

  \path (A) 	edge [loop left]	node[font=\small] {\emph{0.35}} (A)
  			edge            	node[above right, font=\small, pos=0.65] {\emph{0.1}}  (B)
  			edge 			node[left, font=\small, pos=0.8] {\emph{0.2}}  (C)
            edge 			node[below right, font=\small, pos=0.75] {\emph{0.25}} (D)
            edge 			node[above, font=\small, pos=0.65] {\emph{0.1}}  (E)
        (B) 	edge 			node[left, font=\small, pos=0.7] {\emph{0.35}} (A)
  			edge [loop above] node[font=\small] {\emph{0.1}}  (B)
  			edge 			node[above left, font=\small, pos=0.7] {\emph{0.2}}  (C)
            edge 			node[below, font=\small, pos=0.8] {\emph{0.25}} (D)
            edge 			node[above right, font=\small, pos=0.85] {\emph{0.1}}  (E)
        (C) 	edge 			node[left, font=\small, pos=0.8] {\emph{0.35}} (A)
  			edge 			node[above left, font=\small, pos=0.7] {\emph{0.1}}  (B)
  			edge [loop above] node[font=\small] {\emph{0.2}}  (C)
            edge 			node[below left, font=\small, pos=0.75] {\emph{0.25}} (D)
            edge 			node[above right, font=\small, pos=0.9] {\emph{0.1}}  (E)
        (D) 	edge 			node[above left, font=\small, pos=0.8] {\emph{0.35}} (A)
  			edge 			node[above, font=\small, pos=0.8] {\emph{0.1}}  (B)
  			edge 			node[below left, font=\small, pos=0.65] {\emph{0.2}}  (C)
            edge [loop above] node[font=\small] {\emph{0.25}} (D)
            edge 			node[right, font=\small, pos=0.8] {\emph{0.1}}  (E)
        (E) 	edge 			node[above, font=\small, pos=0.65] {\emph{0.35}} (A)
  			edge 			node[above right, font=\small, pos=0.8] {\emph{0.1}}  (B)
  			edge 			node[left, font=\small, pos=0.8] {\emph{0.2}}  (C)
            edge 			node[right, font=\small, pos=0.65] {\emph{0.25}} (D)
            edge [loop right]	 node[font=\small] {\emph{0.1}}  (E);
\end{tikzpicture}  }  }
\end{figure}

\vspace{2mm}
What the example shows is that a shock election can occur in a society that would otherwise make the correct decision if there were no learning and/or no confirmation bias. In particular, the effect of confirmation bias is that it can push a majority to support the Right candidate in the short term even if in the long term the society correctly aggregates the initial signals and votes for the Left candidate. In this example, it happens because of two effects: (i) Swing voters stop listening to the Extreme Left/Right voters due to confirmation bias, and (ii) they put more weight on the views of Centre Right compared to Centre Left voters. This means that, in the short term, the majority of Swing voters vote for the Right candidate who obtains a majority support. In fact, these two effects are sufficient conditions (i) and (iii) in Proposition \ref{voting result}. These two effects guarantee a shock election when the Left voters fail to form a majority on their own.

\begin{gather*}
 T^{*} =
   \begin{pmatrix}
    0.9 & 0.1 & 0 & 0 & 0 \\
    0.35 & 0.45 & 0.2 & 0 & 0 \\
    0 & 0.1 & 0.65 & 0.25 & 0 \\
    0 & 0 & 0.2 & 0.7 & 0.1 \\
    0 & 0 & 0 & 0.25 & 0.75
   \end{pmatrix}
\implies T^{*} \cdot x_{0} =
	\begin{pmatrix}
	0.17 \\ 
	0.29 \\ 
	0.52 \\ 
	0.63 \\ 
	0.73 
	\end{pmatrix}
\implies ... \implies T^{*\infty} \cdot x_{0} =
	\begin{pmatrix}
	0.42 \\
	0.42 \\
	0.42 \\
	0.42 \\
	0.42
	\end{pmatrix}
\end{gather*}

\begin{figure}[ht]
\caption{Network $T^*$, with confirmation bias}\label{fig:voting2}
\centering
\fbox{	\resizebox{\columnwidth}{!}{
\begin{tikzpicture}[->,>=stealth',shorten >=1pt,auto,node distance=3.5cm,semithick]
  \tikzstyle{every state}=[fill=white, draw, text=black, circle]
  \node[state, label =below:{$x=0.15$} ] 	(A)  				{$EL$};
  \node[state, label =below:{$x=0.3$} ] 	(B) [right =2cm of A]  	{$CL$};
  \node[state, label =below:{$x=0.5$} ] 	(C) 	[right =2cm of B]	{$S$};
  \node[state, label =below:{$x=0.65$} ] 	(D) [right =2cm of C]	{$CR$};
  \node[state, label =below:{$x=0.75$} ] 	(E) [right =2cm of D]  	{$ER$};

  \path (A) 	edge [loop above]	node {0.9} 	(A)
  			edge [bend left=20] 	node {0.1}  	(B)
        (B) 	edge [bend left=20]	node {0.35} 	(A)
  			edge [loop above]	node {0.45}	(B)
  			edge [bend left=20]	node {0.2}	(C)
        (C) 	edge [bend left=20]	node {0.1}  (B)
  			edge [loop above]	node {0.65}  (C)
            edge [bend left=20]	node {0.25} (D)
        (D) 	edge [bend left=20]	node {0.2}  (C)
            edge [loop above]		node {0.7} (D)
            edge [bend left=20]	node {0.1}  (E)
        (E) 	edge [bend left=20]	node {0.25} (D)
            edge [loop above]		node {0.75}  (E);
\end{tikzpicture} } }
\end{figure}

\section{Fringe media organizations} \label{section media}
Changes in technology and the emergence of social media have made it easier to find news and information tailored exactly to a user's preferences \citep[Chapter 3]{sunstein2018republic}. This has been accompanied by a polarization of the landscape, with outlets characterized by an extreme ideology achieving a large audience \citep{prior2013media}. The purpose of this section is to examine the role of confirmation bias in the growth of media organizations with extreme ideologies. To do this, we extend the model in section \ref{model} by introducing media organizations which choose their ideology and care about maximizing their audience, and by having agents form a link with a media organization before learning takes place.

Consider a media organization $m \in \{ 1, ... , M \}$ and assume that $M>5$.\footnote{The $M>5$ assumption allows us to focus on the case of a moderately competitive media market.} 
At time $t=-1$, each media organization $m$ picks an \emph{ideology} $\rho_{m}\in [0,1]$ in order to maximize its utility $V_{m}(\rho_{m}) = \left| \{ i \in N : T_{im} > 0 \} \right|$. They do so knowing that the initial belief of agent $i$ is a random draw from $X \sim U [0,1]$. This assumes that media organizations only care about the number of agents who listen to them. Once chosen, the ideology of each media organization $m$ is fixed, so we can represent each $m$ as a node with a self-link of weight 1.\footnote{This is equivalent to \citet{hotelling1929}'s model with fixed prices where the $[0,1]$ line represents the unidimensional ideology space. \citet{eaton1975} presents results for this model in the absence of confirmation bias. \citet{anand2007information} examine media bias under the same assumption that media organizations care only about profits. In contrast to this model, they assume that agents care about some objective \emph{truth}, as well as ideology.}

As in the set-up in section \ref{model}, at time $t=0$ each agent is endowed with a signal $\theta_i \sim U [0,1]$ which is agent $i$'s initial belief $x_{i0}$. 
The novel element here is that each agent also forms a link of weight $\epsilon$ with one and only one media organization at zero cost. Both the assumptions that the link weight is $\epsilon$ and it is costless are for simplicity given that the focus of this section is on the choice of ideology by the media organizations.\footnote{Picking a weight of $\epsilon$ allows us to sidestep the issue of how each agent redistributes the weights of their outgoing links to accommodate this new connection. The assumption that each agent forms one and only one link is for analytical tractability and can be justified behaviourally by limited attention (see, e.g., \citet{gabaix2014sparsity}, and \citet{masatlioglu2012revealed} and references therein).} We assume each agent $i$'s payoff is equal to $\ U_{i} = U_{i}(|x_{i0} - \rho_{m} |)$ with $U^{'}_{i}(\cdot) < 0$, i.e., agents choose to listen to the media organization whose belief is closest to their own. If there is confirmation bias of strength $q$ then an agent $i$ forms a link with a media organization $m$ if and only if $|x_{i0} - \rho_{m} | \leq 1-q$.

Let $\Psi_{M}^{q}$ be the set of equilibria with $M$ firms and confirmation bias of strength $q$.\footnote{An equilibrium always exists by standard existence results. For $M>5$ there are infinitely many equilibria both with and without confirmation bias.} An equilibrium $\psi=\{ \rho_{1}, \ldots, \rho_{M} \} \in \Psi_{M}^{q}$ is a set of ideologies for the media organizations. The focus of our analysis is the media organization with the most extreme ideology  $\rho_{fr}(q, M) = \min_{\rho_{m} \in \Psi_{M}^{q} } \{ \rho_{m} \}$ which we dub the ``fringe'' media organization. Notice that, by the symmetry of our set-up, it suffices to investigate the most extreme Left ideology because for any set of ideologies that constitutes an equilibrium, its mirror image is also an equilibrium.

Our first result shows that the ideology of the fringe media organization is increasing in the competitiveness of the market.

\begin{prop}\label{media monotone m}
The ideology of the fringe media organization becomes (weakly) more extreme as the number of media organizations $M$ increases.
\end{prop}

The intuition is that media organizations try to spread out in ideology space to maximize the number of agents who listen to them, and therefore an increase in the number of outlets in the market pushes the most extreme organization closer to the boundary. This is in agreement with the initial observation that the proliferation of media outlets has led to a polarization of the media landscape. The following statement shows how the presence and strength of confirmation bias affect this polarization.

\begin{prop}\label{media monotone q}
The ideology of the fringe media organization becomes (weakly) more extreme as the strength of confirmation bias $q$ increases.
\end{prop}

The presence of confirmation bias reduces the incentive of a fringe media organization to moderate its ideology. Absent confirmation bias, the most extreme organization would want to moderate its ideology in order to attract more moderate listeners as this would not cost it any extreme listeners. What limits this moderation is the presence of other more moderate organizations that might want to `leapfrog' the fringe media if it becomes too moderate. In the presence of confirmation bias, however, the fringe media organization does not want to moderate its ideology as much, for fear of losing some of its most extreme listeners.

\section{Simulations} \label{simulations results}
Finally, we present a series of simulation exercises. First, section \ref{sec:setup} provides the overall set-up. In section \ref{sec:simulation_tests}, we then show that our main theorem on consensus time and our result about shock elections hold when we relax some of the assumptions we imposed for analytic tractability. Next, section \ref{sec:conf_vs_homophily} compares confirmation bias and homophily, showing how our model of confirmation bias complements existing analysis of the effects of homophily (especially that of \cite{golub2012homophily}). Finally, section \ref{sec:further results} explores the robustness of our results to using different metrics to capture the notions of consensus and shock electoral outcomes, and examines how the convergence dynamics is affected by the distribution of initial beliefs.

\subsection{Set-up} \label{sec:setup}
The first step is to build a network with some key stylized features that are common in real-world networks. We use three very widely studied models of random network formation. First, an Erd\H{o}s-R\'enyi random network, where links form uniformly at random \citep{erdos1959random}. Second, a scale-free random network, where links are formed through preferential attachment \citep{barabasi1999emergence}. Third, the \cite{jackson2007meeting} algorithm (modified to create directed networks), which mixes both random meetings and preferential attachment. We parametrize them so that each network contains 200 agents, and an average degree of 30.

In this section, we report results from the \cite{jackson2007meeting} simulations because these networks best match the main structural characteristics of real social networks. The Online Appendix contains a detailed explanation of these network construction processes, as well as the results for the Erd\H{o}s-R\'enyi and scale-free networks.  


Next, we create a set of initial beliefs. We draw initial beliefs from one of four different distributions: (a) the Beta distribution with parameters $(0.5,0.5)$, (b) the uniform distribution on $[0,1]$, (c) the Beta distribution with parameters $(2,2)$, and (d) a discrete distribution with five mass points that mimics the set-up in section \ref{voting section}. The five points are located at $\{0.15, 0.30, 0.5, 0.65, 0.75\}$, and the respective masses are $\{0.35, 0.1, 0.2, 0.25, 0.1\}$. The first three distributions, (a)-(c) generate successively more moderate initial beliefs, while the fourth one matches the example in Figure \ref{fig:voting}.

Finally, we choose values $q = 0.25, 0.5, 0.75$ for the confirmation bias parameter when we study consensus time, and $q=0.78$ -- the value from the example in Figure \ref{fig:voting} -- when we study shock elections. For clarity, the term \emph{run} indicates a combination of a network generating process and a distribution from which to draw initial beliefs. In each run, we generate 100 different networks, each with 200 different agents, and draw initial beliefs 100 times. This gives 10,000 \emph{instances} where we compare the learning process with and without confirmation bias. 
In each of these instances, starting from the network $T$ and initial beliefs $\{x_{i0}\}_{i \in N}$, we remove any link $T_{ij}$ where $\mid x_{i0}-x_{j0} \mid \ > \ 1-q$ and re-route the weight of that link back onto the self-link $T_{ii}^*$ to form the network $T^*$. The same network and the same initial beliefs are used for every value of $q$ (and, in section \ref{sec:conf_vs_homophily}, of $\phi$), so that differences across treatments are not due to sampling variation.

\subsection{Robustness checks}\label{sec:simulation_tests}
Here we test the robustness of our results on consensus time (Theorem \ref{speed result}) and shock elections (Proposition \ref{voting result}) to relaxing some of the assumptions that we imposed for analytic tractability.

\paragraph{Consensus time.} When using simulation methods it is feasible to consider a more `direct' measure of consensus time -- the earliest time at which the \emph{range} of beliefs falls below a threshold $\epsilon$. Formally, we define the \emph{belief consensus time} as: $C(\epsilon) = \min \{ t : \max_{i,j\in N} |x_{i,t} - x_{j,t}| < \epsilon \}$. Note that, even absent confirmation bias, this is a function of both the network $T$ and the initial beliefs $x_0$. Further, both this definition and our pre-existing definition of average consensus time (Definition \ref{def:time}) are functions of $\epsilon$ -- how close agents' beliefs must be for us to say that beliefs have reached a consensus. Here, we pick $\epsilon = 10^{-4}$ for both definitions. Choosing a smaller value of $\epsilon$ would increase consensus time, but would not affect the insights.

\begin{figure}[b]
    \centering
\begin{subfigure}[t]{0.49\linewidth}
\centering
\includegraphics[scale=0.45]{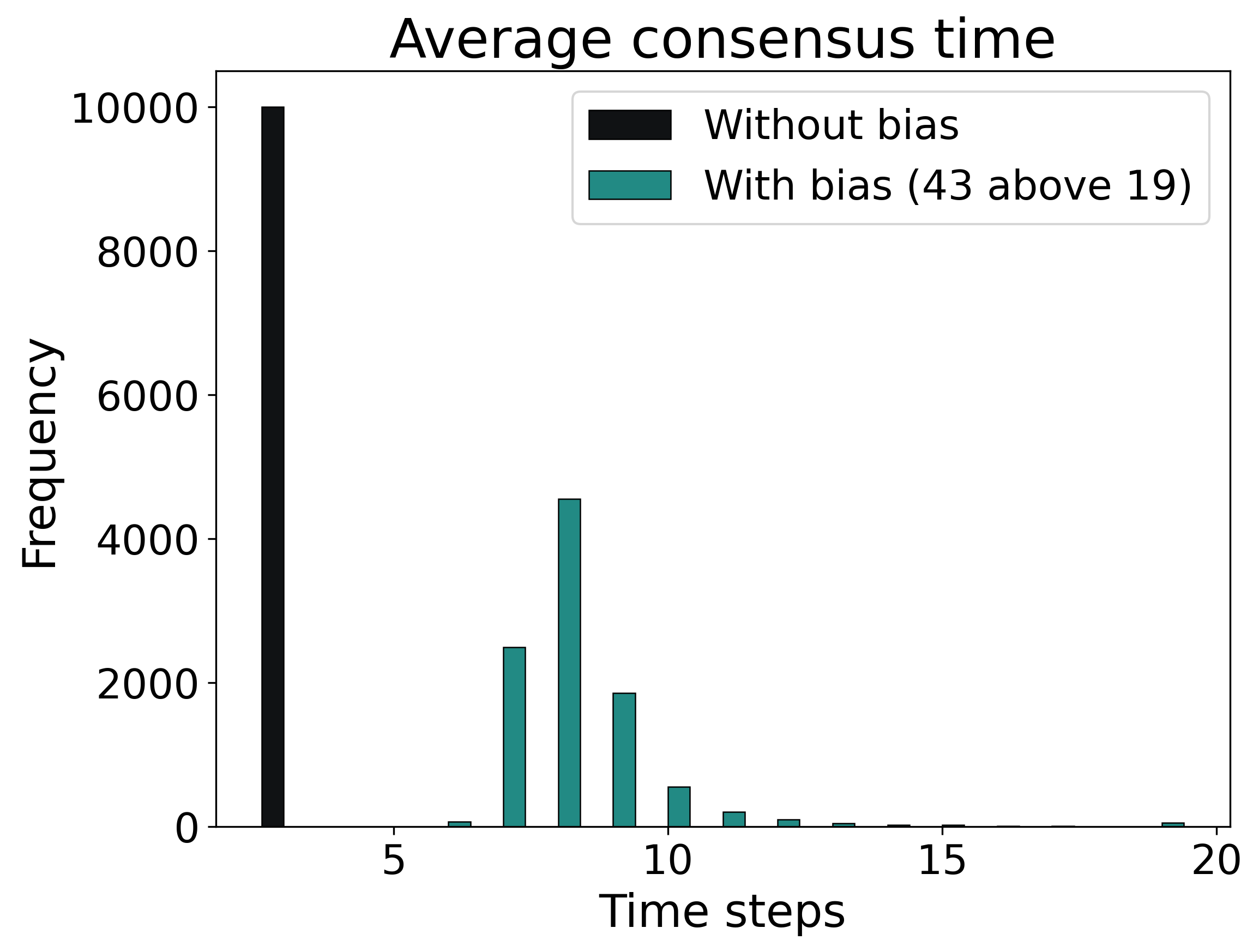}
\caption{ }
\end{subfigure}
\hfill
\begin{subfigure}[t]{0.49\linewidth}
\centering
\includegraphics[scale=0.45]{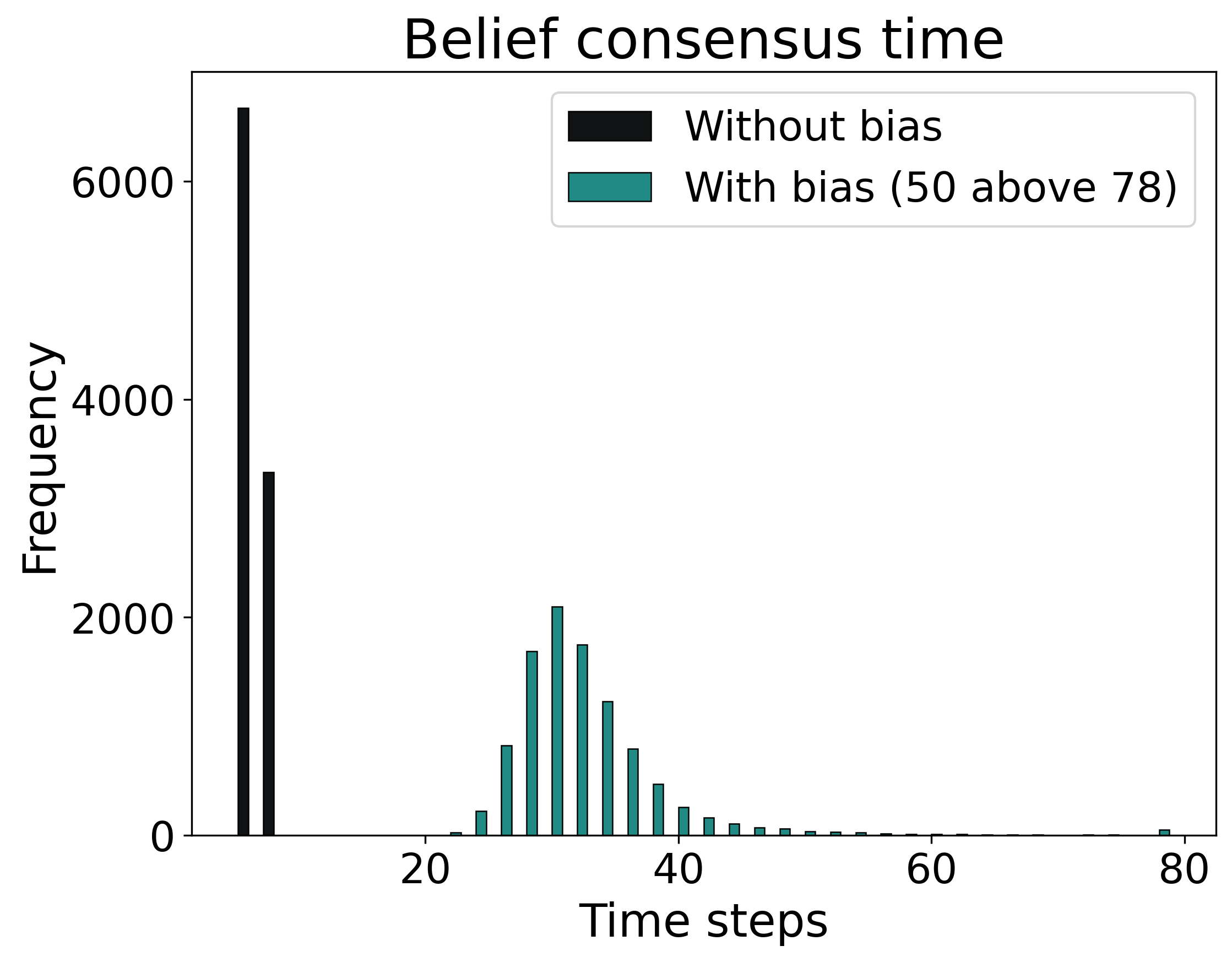}
\caption{ }
\end{subfigure}
    \caption{Average consensus time (left panel) and belief consensus time (right panel) without (black) and with (teal) confirmation bias. Uniform distribution of initial beliefs, $q=0.5$, network constructed using the Jackson-Rogers algorithm.}
    \label{fig:consensus_time_sims}
\end{figure}

Figure \ref{fig:consensus_time_sims} shows the frequency distribution of the average consensus time (Definition \ref{def:time}, to which our analytic results apply) and of the belief consensus time for the networks $T$ (black bars) and $T^*$ (teal bars) for networks constructed via the \cite{jackson2007meeting} algorithm, with a uniform distribution of initial beliefs, and with $q=0.5$. It is clear that confirmation bias shifts the distribution to the right, in line with Theorem \ref{speed result}. This is despite the fact that Theorem \ref{speed result} assumes that \emph{every} link is reciprocated and that $T_{ii} \geq \frac{1}{2}$ for all $i$, while (by construction) \emph{none} of the links in these simulated networks are reciprocated and $T_{ii} = 0$ for all $i$. The mean average consensus time without confirmation bias is $3.0$ periods, while it is $10.8$ with confirmation bias (with $min= 6, max=3481$). The mean belief consensus time without confirmation bias is $6.3$ periods (with $min= 6, max=7$), while it is $36.6$ with confirmation bias ($min=19, max=6449$). There is not a single instance in which either consensus time is shorter with confirmation bias than without.

This suggests that our main result holds when we relax both of the conditions we imposed in Theorem \ref{speed result}. Later in this section, we show the outcomes for other values of $q$ (Figure \ref{fig:effect_of_phi}) and for different distributions of initial beliefs (Figure \ref{fig:distribution_robustness}) that further corroborate the results in Figure \ref{fig:consensus_time_sims}. The Online Appendix shows that the key insights hold for runs with different network formation processes too.

\paragraph{Shock elections.} To examine the incidence of shock elections numerically, we adopt the distribution of initial beliefs and the level of confirmation bias ($q=0.78$) from the example in Figure \ref{fig:voting}. We relax all the assumptions on the structure of the network in Proposition \ref{voting result}, which, taken together,  ruled out many possible network structures. Recall that because of the distribution of initial beliefs we are taking -- with more voters to the left of $0.5$ than to the right, and a mean below $0.5$ -- a shock election here involves voting for the `Right' candidate.

\begin{figure}[b!]
    \centering
    \includegraphics[width=0.75\linewidth]{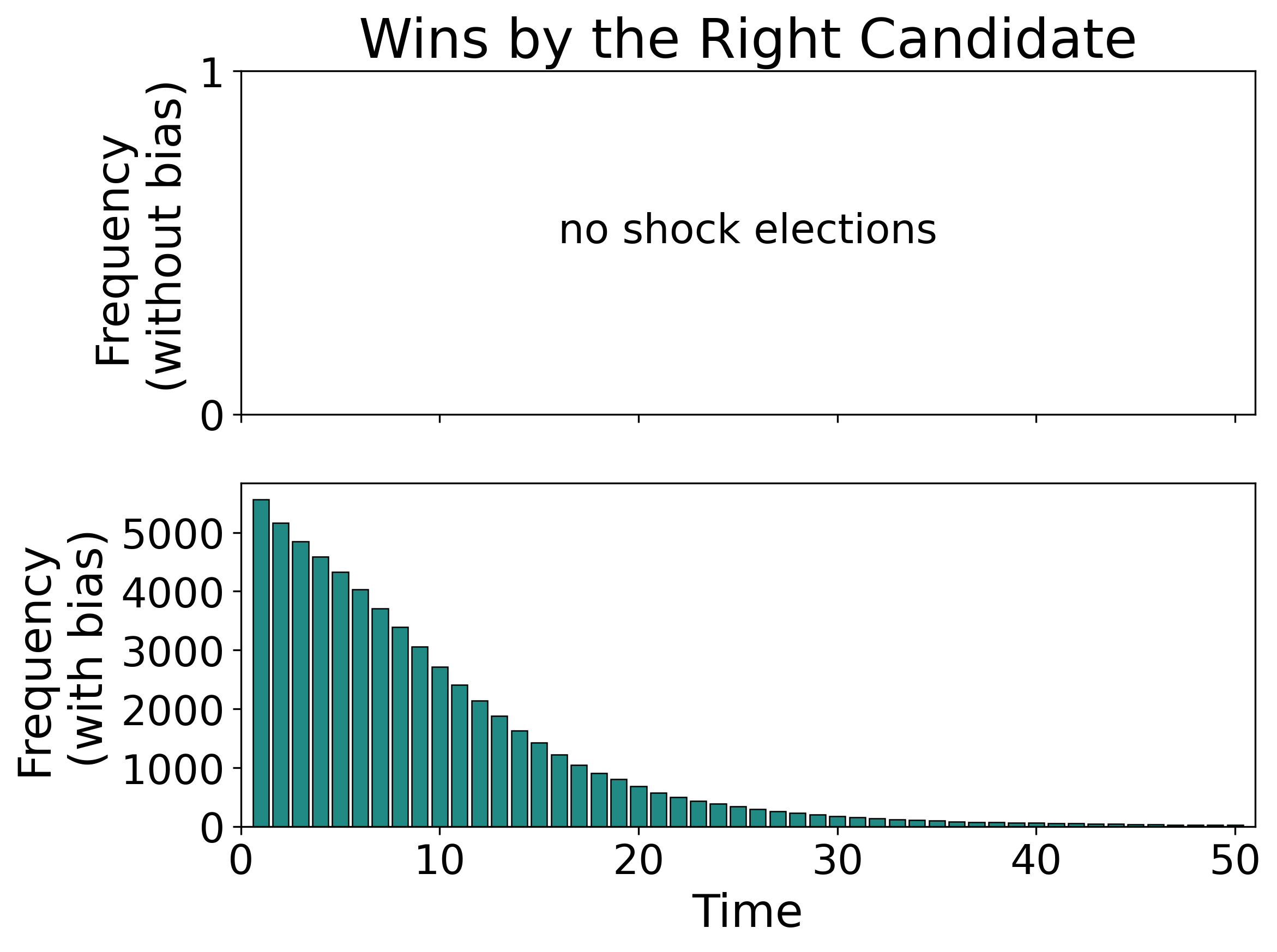}
    \caption{Wins by Right Candidate without (black, top panel) and with (teal, bottom panel) confirmation bias. Each bar counts the number of instances in which the Right candidate would win an election held at that date. Point mass distribution of initial beliefs, $q=0.78$, network constructed using the Jackson-Rogers algorithm.}
    \label{fig:shock_election_sims}
\end{figure}

Assuming there is an election at each point in time, Figure \ref{fig:shock_election_sims} shows the frequency of shock elections with confirmation bias (teal bars, bottom panel) and shows that without confirmation bias there are none (top panel) in this setting. With confirmation bias, a shock election occurs in roughly 60\% of instances. The key takeaway is clear -- shock elections can be common occurrences with confirmation bias, but do not occur without it.\footnote{Following the definition of a shock election, we only count instances in which the Left candidate wins both at $t=0$ and once a consensus is reached. Because initial beliefs are drawn at random, in a small number of instances (about 5\%) the Right candidate happens to win at $t=0$; these instances are excluded. In a further 1.3\% of instances, confirmation bias changes the long-run winner: a majority of long-run beliefs lie above $0.5$ under $T^*$ but not under $T$. As the Right candidate then wins the long-run election, these instances are excluded too.}

\subsection{Confirmation Bias and homophily}\label{sec:conf_vs_homophily}
Here we compare the implications of raising homophily and raising confirmation bias on learning speed and shock elections. 

\paragraph{Link re-routing and learning speed.} First, we show that the process of re-routing links to oneself plays an important role in slowing learning. Figure \ref{fig:effect_of_phi} shows how both the average consensus time and the belief consensus time change for different levels of $\phi \in \{0, 0.5, 1\}$ -- the fraction of weight from cut links that are re-routed to oneself -- and different levels of confirmation bias $q \in \{0.25, 0.5, 0.75\}$, with the no-bias level shown as a dashed line for reference.\footnote{Figure \ref{fig:effect_of_phi} reports the median across instances. At $q=0.75$ the mean is dominated by a minority of instances (about 7\%) in which the network $T^*$ is very weakly connected, so that some agents are only very weakly linked to the rest of society and consensus takes thousands of periods (at $q=0.5$ this happens in $0.1\%$ of instances, and at $q=0.25$ never). The ordering of the lines is the same for the mean.} 
This confirms that reduction in learning speed due to confirmation bias is not only due to the part of the change that increases homophily (i.e., the link-cutting with people whose beliefs are too different). Learning slows in $\phi$ for each value of $q$. Note that we have done this in a setting where link-cutting by itself \emph{does} slow down learning (so the counterexample presented in Figure \ref{fig:speed} does not apply): at $\phi = 0$, learning still slows in $q$. 

\begin{figure}[t]
    \centering
\begin{subfigure}[t]{0.49\linewidth}
\centering
\includegraphics[scale=0.45]{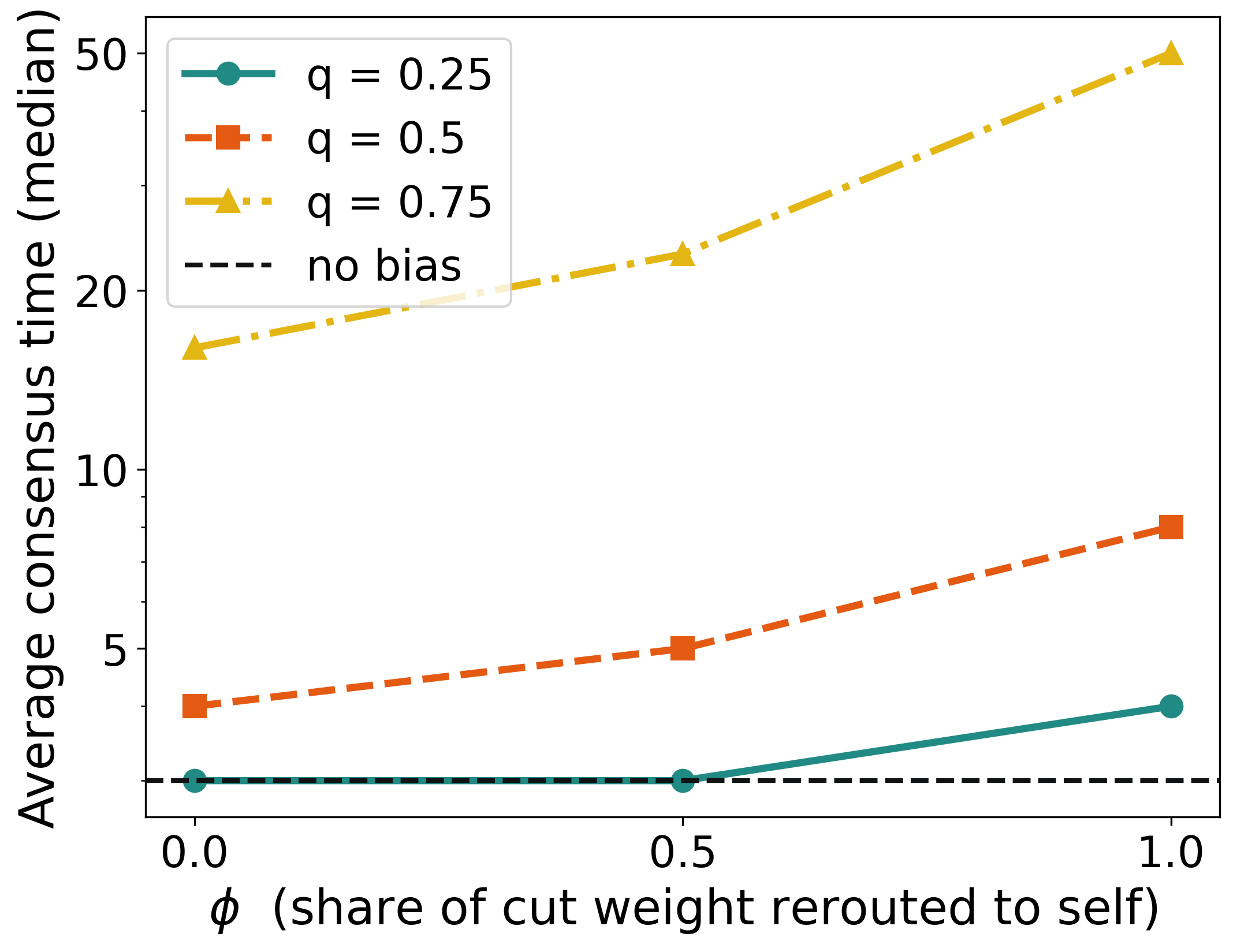}
\caption{ }
\end{subfigure}
\hfill
\begin{subfigure}[t]{0.49\linewidth}
\centering
\includegraphics[scale=0.45]{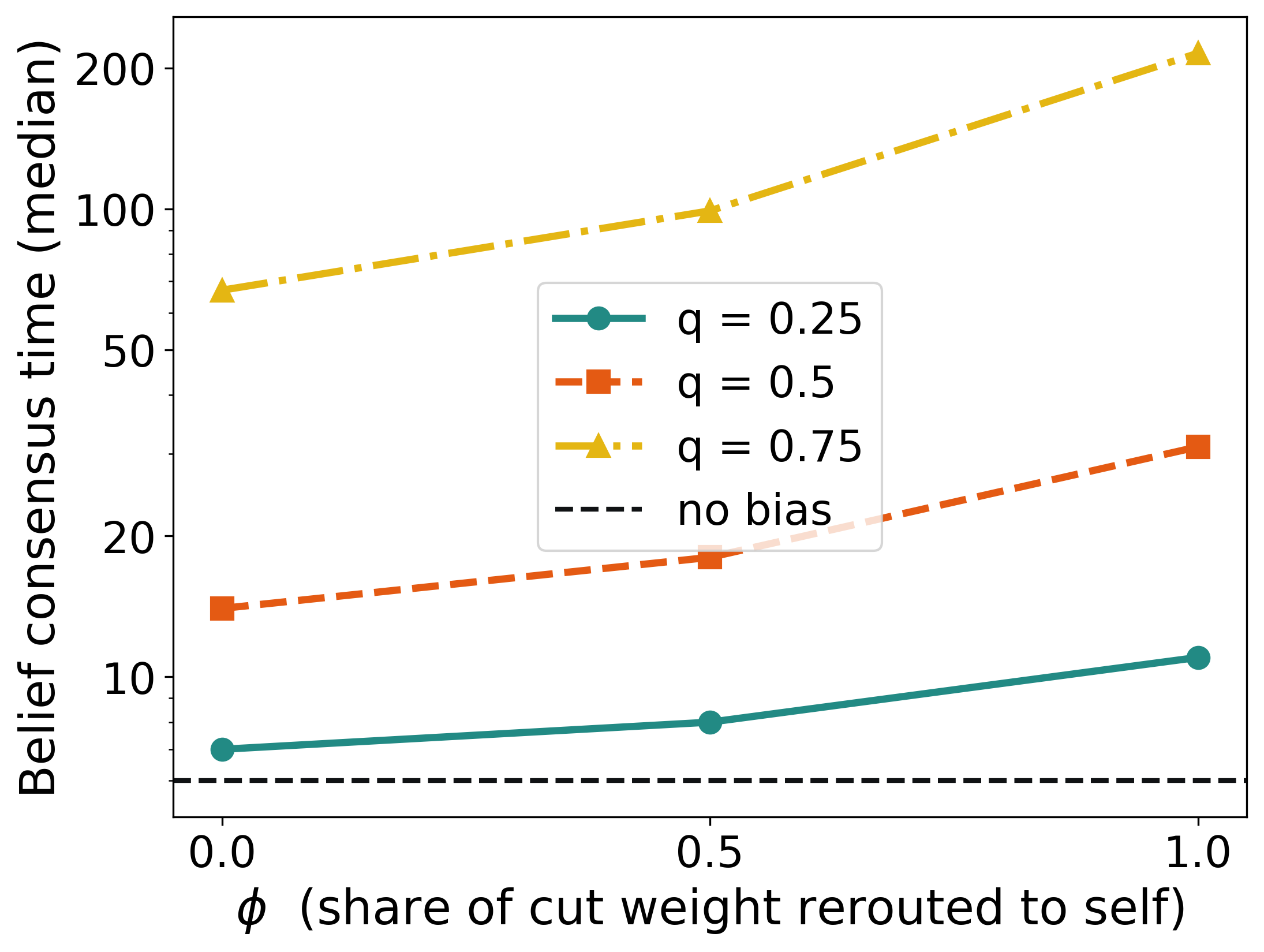}
\caption{ }
\end{subfigure}
    \caption{Average consensus time (left panel) and Belief consensus time (right panel): median across the 10,000 instances, on a logarithmic scale. Uniform distribution of initial beliefs, network constructed using the Jackson-Rogers algorithm.}
    \label{fig:effect_of_phi}
\end{figure}

\paragraph{Beyond impacts on the second eigenvalue.} Next, we examine whether confirmation bias has an impact on learning speed and shock elections over and above its impact on spectral homophily. This is a stringent test, as spectral homophily is \emph{defined} as the SLEM (Second Largest Eigenvalue Modulus), which is known to be very closely linked to learning speed. To do this, we draw random networks according to a stochastic block model with two equally sized blocks, calibrated to have a SLEM of $0.4$ (and, as before, 200 agents and an average degree of 30), where block membership is independent of beliefs. We draw initial beliefs from a uniform distribution, and then apply confirmation bias with $q=0.7$ to obtain the network $T^*$. We then calculate the SLEM of this network. To create a fair comparison, we draw another network according to a stochastic block model -- again with block membership independent of beliefs -- but now calibrated to \emph{match} the SLEM of $T^*$ in expectation. As this is a stochastic process, the SLEM will not match exactly, but we redraw if the difference in the SLEMs is greater than $0.01$. We then run the learning process on this second network starting from the \emph{same} initial beliefs. The left-hand panel of Figure \ref{fig:constantSLEM_learning} plots the values of the SLEMs for our network with confirmation bias (y-axis) and the one calibrated to match (x-axis). We can see that the points fall close to the 45-degree line.

\begin{figure}[t]
    \centering
    \includegraphics[width=\linewidth]{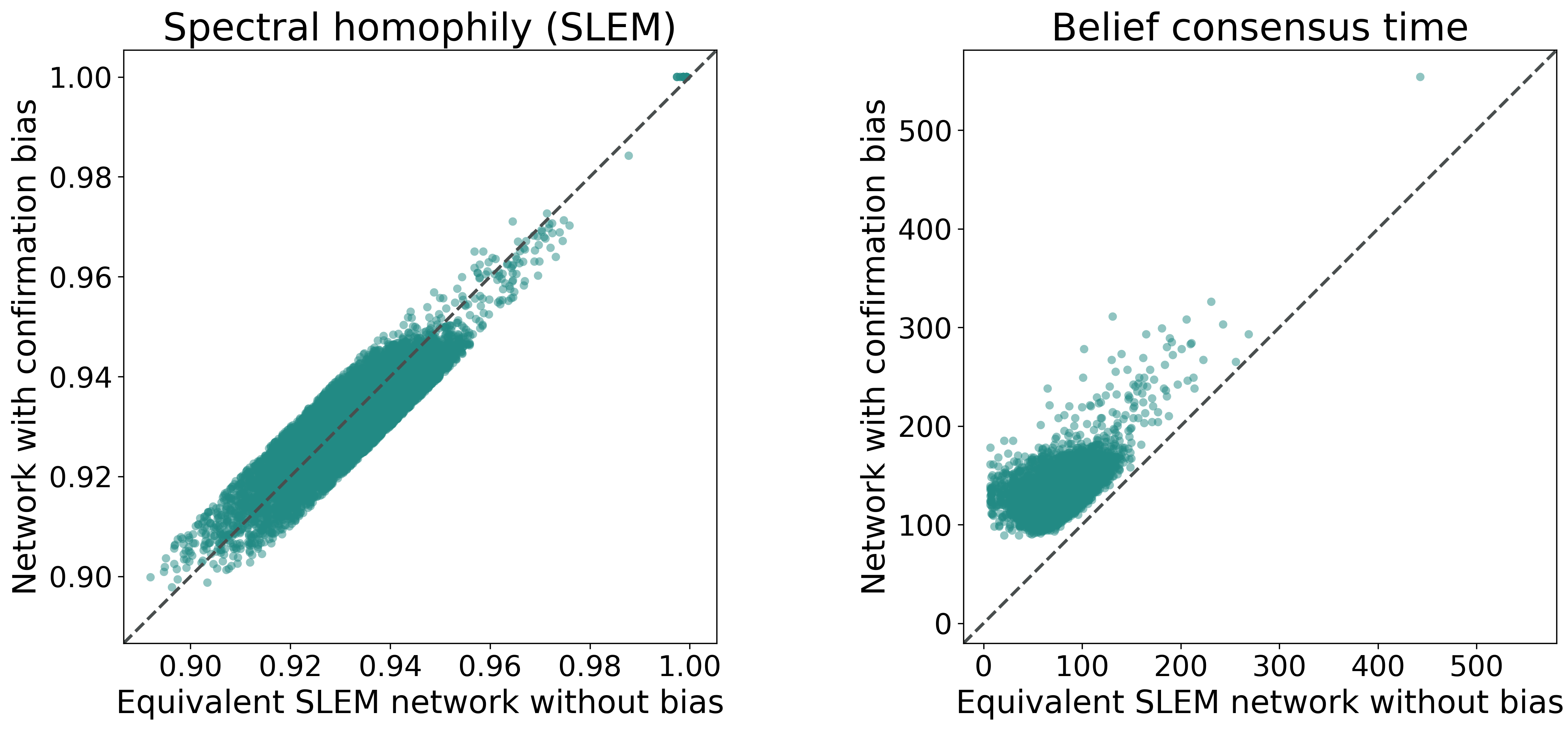}
    \caption{Spectral homophily (left panel) and belief consensus time (right panel) for the network created by confirmation bias ($q=0.7$, y-axis) against a stochastic block model network with the same initial beliefs and the same SLEM (x-axis). Each point is one (network, initial beliefs) instance; the dashed line is the 45-degree line. Uniform distribution of initial beliefs, networks constructed using a two-group stochastic block model (a version of an Erd\H{o}s-R\'enyi random network) with an initial SLEM of $0.4$.}
    \label{fig:constantSLEM_learning}
\end{figure}

Then the right-hand panel of Figure \ref{fig:constantSLEM_learning} shows that, once we account for the actual beliefs by looking at the belief consensus time, learning is significantly slower in the network created by confirmation bias than in a stochastic block model network with the same initial beliefs and the same level of spectral homophily. Again, this suggests that confirmation bias is not simply a behavioural foundation for spectral homophily -- although it does drive higher spectral homophily in our model -- but slows down learning above and beyond that. Intuitively, this happens because confirmation bias induces sorting on beliefs, not just on groups.

\paragraph{Shock elections: segregation on beliefs.} We have shown above, both analytically and using simulations, how confirmation bias can induce ``shock'' election outcomes. But \citet{golub2012homophily} also show how spectral homophily can, by itself, induce the same thing. Here, we examine the relationship between the two.

The headline intuition is that confirmation bias \emph{creates} groups that are segregated along the lines of beliefs. Spectral homophily by itself only creates segregated groups -- that segregation happens to be on beliefs as well as group membership must be assumed. So our model of confirmation bias helps explain why a society might have both the spectral homophily and the segregation in terms of beliefs that together yield shock election outcomes -- confirmation bias drives both at the same time.

In the spirit of \citet{golub2012homophily}, we use a stochastic block model for this exercise. Here, 60\% of agents are in group $A$, and 40\% are in group $B$ (200 agents and an average degree of 30 as before), and the within- and between-group link probabilities are calibrated so that the network's SLEM, before confirmation bias is applied, takes one of three values: $0.5$, $0.7$ or $0.9$. We use the distribution of initial beliefs from Figure \ref{fig:voting} (with 5 point masses) and confirmation bias of $q=0.78$. After drawing the realization of initial beliefs, we initially assign the most right-leaning (i.e., the largest) 60\% to group $A$, and the rest to group $B$, breaking ties at random. Then, we randomly select a fraction $1-\rho$ of agents, and randomly reshuffle their initial beliefs among themselves. Finally, we draw the network from the stochastic block model. This process ensures that, in expectation, the correlation between beliefs and group membership is linear in $\rho$, while holding group sizes and the distribution of initial beliefs fixed. We consider $\rho \in \{0, 0.25, 0.5, 0.75, 1\}$.

Figure \ref{fig:SLEM_shock_elections} shows the prevalence of shock elections with and without confirmation bias. Note that, as in the previous exercise, confirmation bias raises the SLEM of the network it is applied to (to between $0.96$ and $0.99$), so the comparison at each level is between the same starting network with and without confirmation bias. When the initial correlation between beliefs and group membership is low, there are very few -- if any -- shock elections without confirmation bias, but a high prevalence with confirmation bias. As the initial correlation increases, the prevalence of shock elections increases modestly in settings with confirmation bias -- this is because confirmation bias creates high correlation, leaving less room for the initial correlation to have a further effect. In contrast, higher correlation dramatically increases the prevalence of shock elections absent confirmation bias -- so long as spectral homophily (the SLEM) is high enough. When the SLEM is low, there is not enough segregation for correlation in group beliefs to induce shock elections.

\begin{figure}[h]
    \centering
    \includegraphics[width=0.9\linewidth]{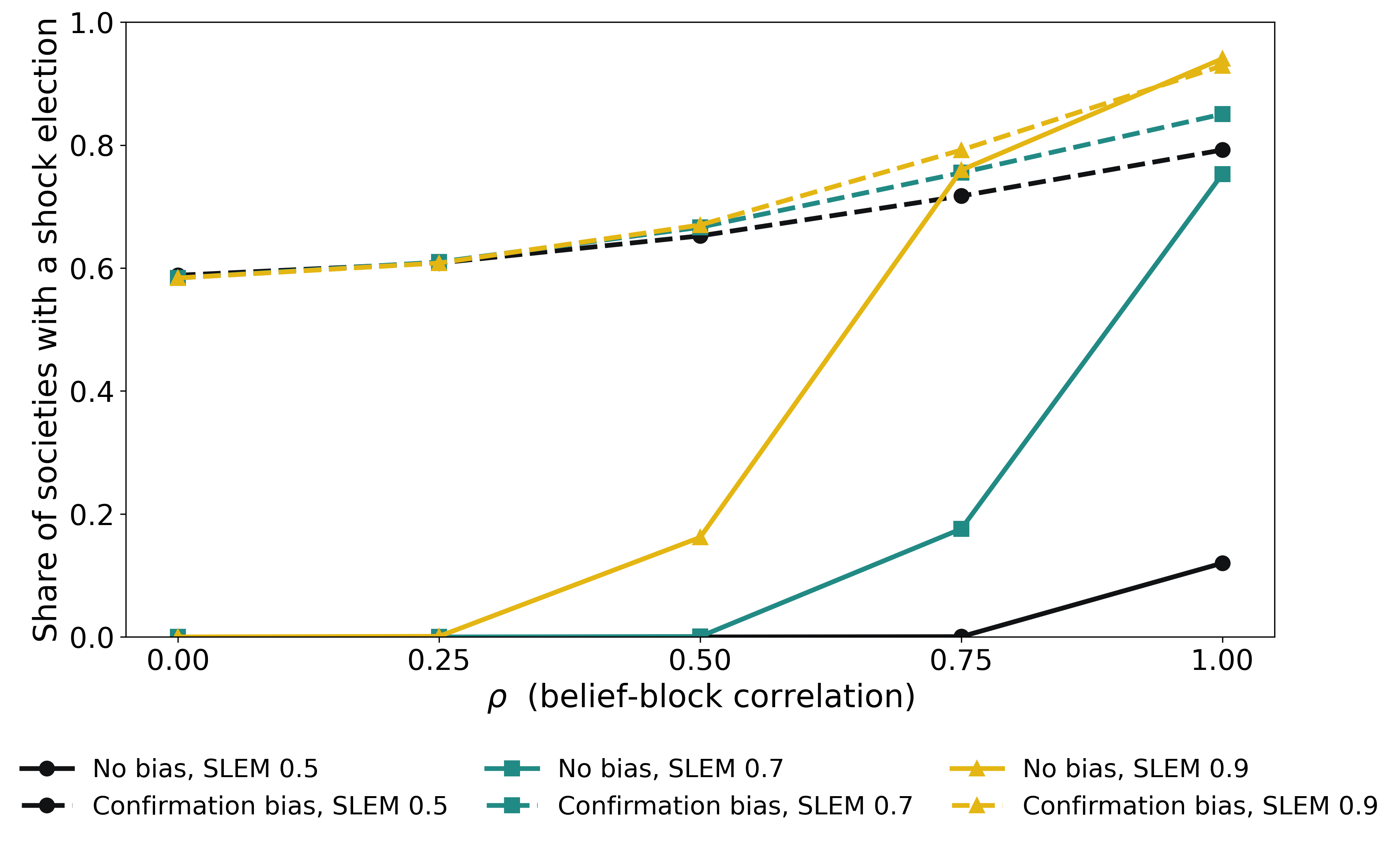}
    \caption{Prevalence of shock elections with and without confirmation bias, for different levels of spectral homophily and of correlation. Solid [resp. dahsed] lines show networks without [resp. with] confirmation bias. Point mass distribution of initial beliefs, $q=0.78$, network constructed using a two-group stochastic block model.}
    \label{fig:SLEM_shock_elections}
\end{figure}

\subsection{Further results} \label{sec:further results}
Here we use the simulations to explore the robustness of our learning speed results to using a different notion of consensus, as well as the make-up of the initial beliefs of a society.

\paragraph{Evolution of beliefs.} In Section \ref{sec:simulation_tests}, we introduced belief consensus time -- a measure of consensus time that accounts for actual beliefs. This is more direct than our analytic definition of consensus time, which is based on the network structure alone. As an alternative way of looking at how the beliefs themselves evolve, we now look at the variance of beliefs over time. This is a very simple measure of the level of consensus at any given point in time.

Figure \ref{fig:polarisation} shows that average variance is higher with confirmation bias than without in every single period. Note, however, that Figure \ref{fig:polarisation} reports averages across instances, and so could in principle mask individual instances in which the variance without bias dips below the variance with bias for some period of the learning process.

\begin{figure}[H]
\centering
\begin{subfigure}[t]{0.49\linewidth}
\centering
\includegraphics[scale=0.45]{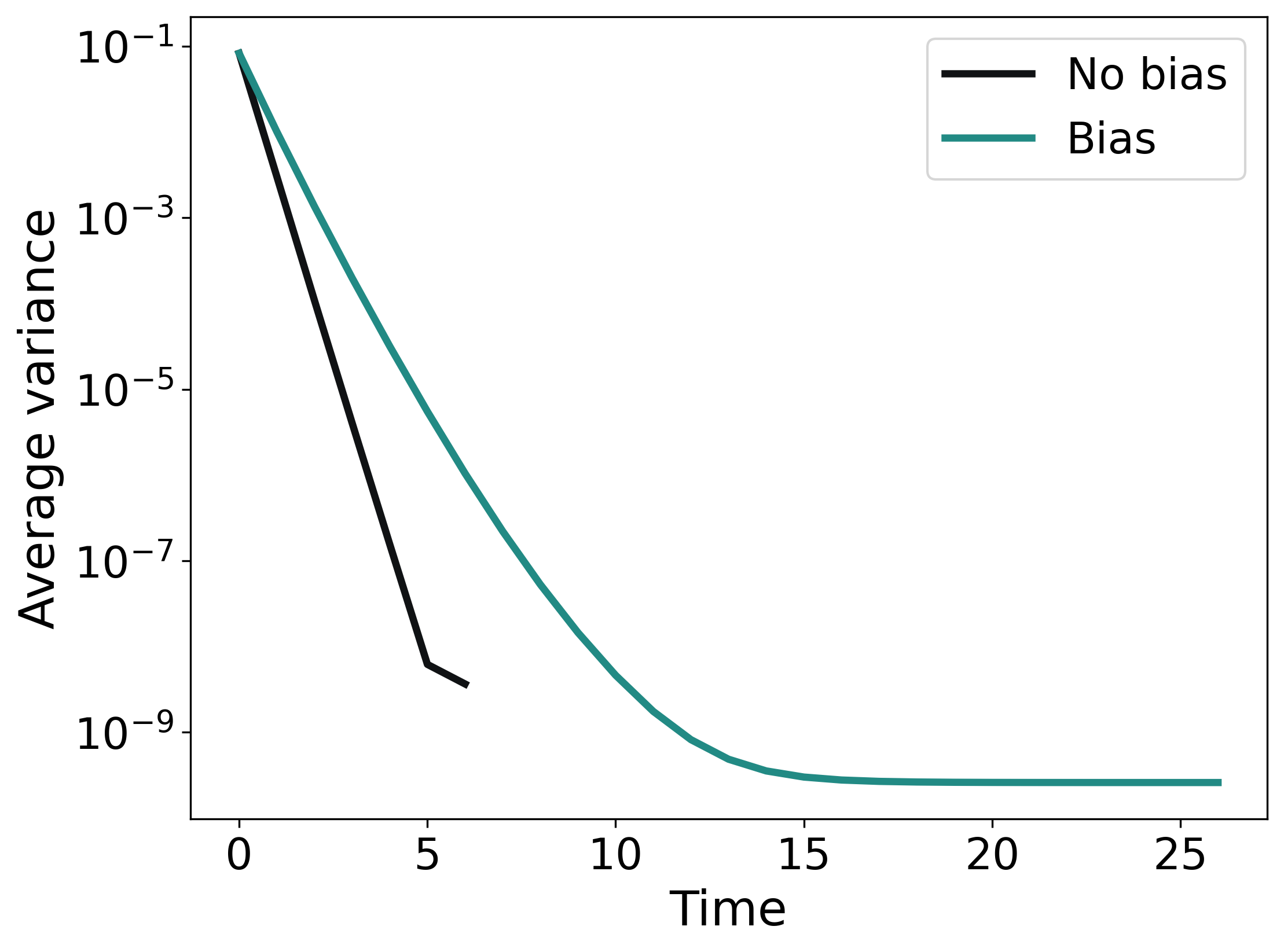}
\caption{ }
\end{subfigure}
\hfill
\begin{subfigure}[t]{0.49\linewidth}
\centering
\includegraphics[scale=0.45]{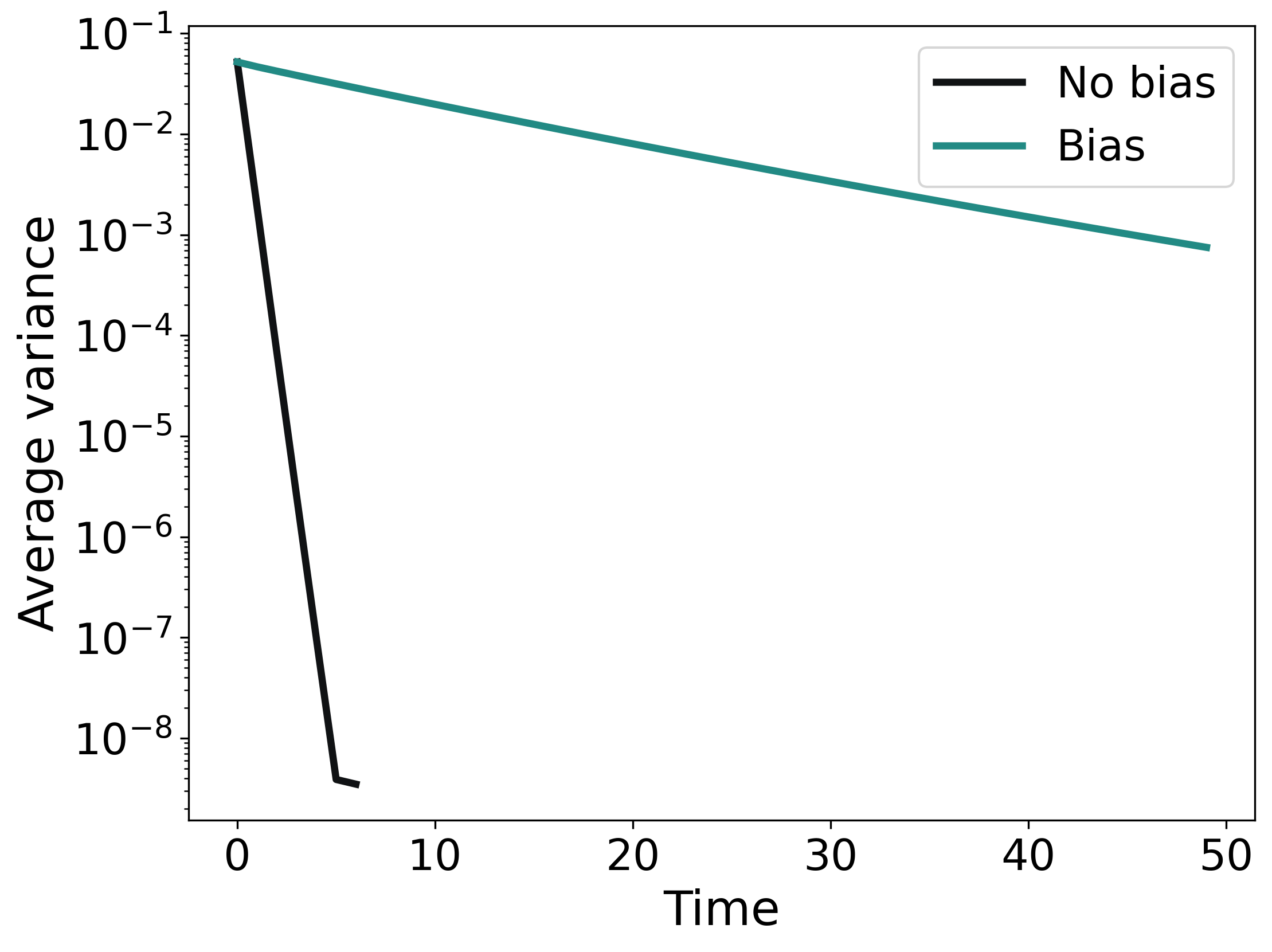}
\caption{ }
\end{subfigure}
\caption{Average variance of beliefs over time without (black) and with (teal) confirmation bias, for the first 50 periods of learning; once an instance has reached a consensus its last value is carried forward. Network constructed using the Jackson-Rogers algorithm. Panel (a): Uniform distribution of initial beliefs, $q=0.25$. Panel (b) Point mass distribution of initial beliefs, $q=0.78$.}
    \label{fig:polarisation}
\end{figure}

We therefore also check instance by instance: in each of the 10,000 instances, and in every period in which both networks are still learning (i.e., before the first of the two reaches a consensus), we compare the variance of beliefs with and without confirmation bias.\footnote{This amounts to 63,329 (instance, period) comparisons in the case of Figure \ref{fig:polarisation}(a), and 60,670 in the case of Figure \ref{fig:polarisation}(b).} We find no instances in which the variance of beliefs is higher without confirmation bias than with it (at $q=0.25, 0.5$ or $0.75$ with a uniform distribution of initial beliefs, or at $q=0.78$ with the point mass distribution of beliefs from Figure \ref{fig:voting}).

\begin{figure}[h!]
    \centering
    \includegraphics[width=0.6\linewidth]{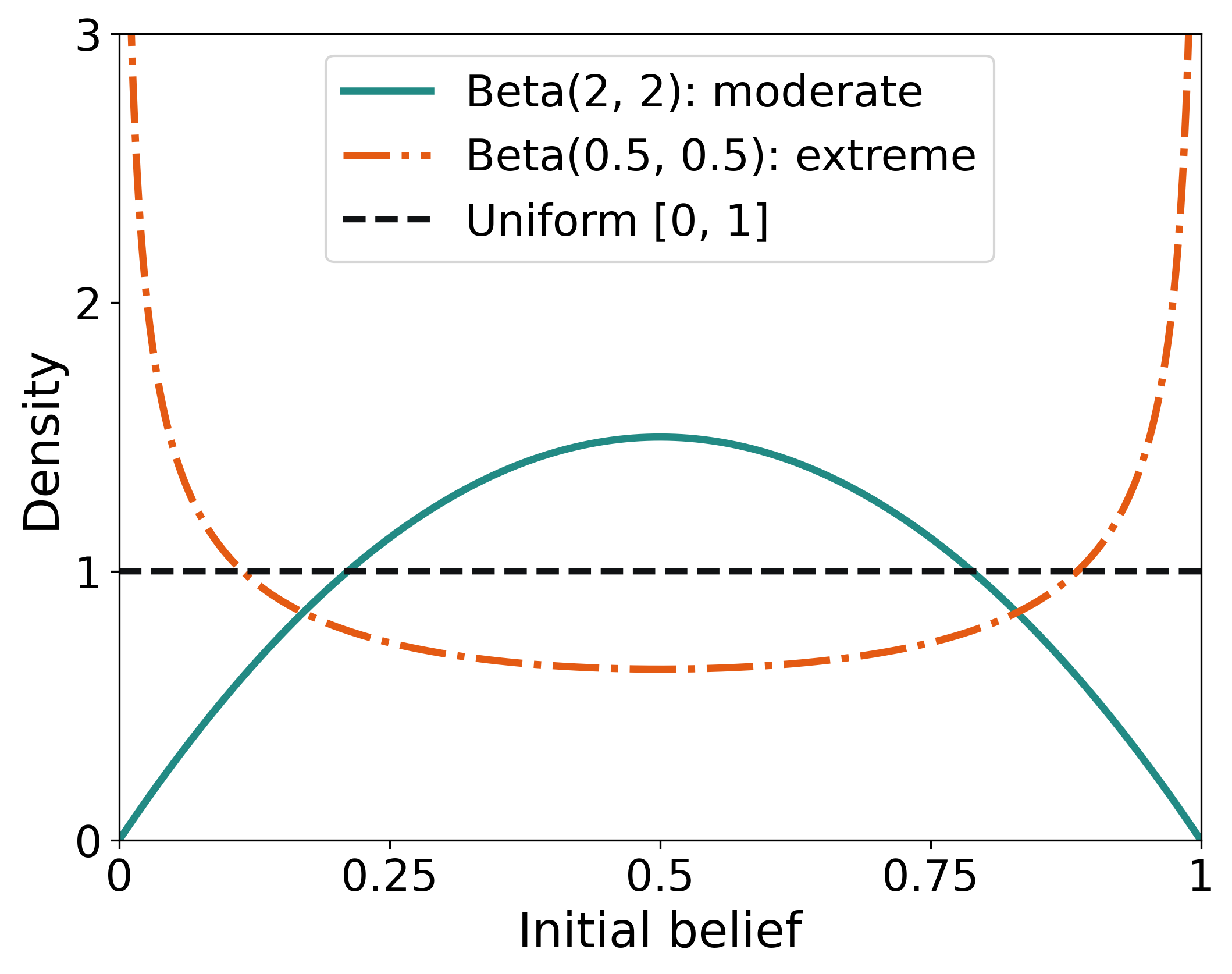}
    \caption{Probability densities of initial beliefs for Beta $(2,2)$ (moderate beliefs, teal), Uniform $[0,1]$ (evenly spread beliefs, dashed black), and Beta $(0.5, 0.5)$ (extreme beliefs, orange).}
    \label{fig:beta_dist}
\end{figure}

\paragraph{The role of initial beliefs.} In a similar vein, our simulation exercises can also provide a more fine-grained view of exactly how the initial beliefs affect learning. That is, they can suggest whether certain distributions of initial beliefs are more sensitive to the effects of confirmation bias. 

Two natural benchmarks are: (1) a distribution where initial beliefs are mostly moderate, with only a few at the extremes, and (2) a distribution where agents' initial beliefs are mostly extreme, with only a few moderate ones. Figure \ref{fig:beta_dist} displays two such distributions: (1) a Beta distribution with parameters $(2,2)$, yielding mostly moderate beliefs, and (2) a Beta distribution with parameters $(0.5, 0.5)$, yielding mostly extreme beliefs. 

\begin{figure}[h!]
  \centering
  \begin{subfigure}[b]{0.45\textwidth}
    \centering
    \includegraphics[width=\linewidth]{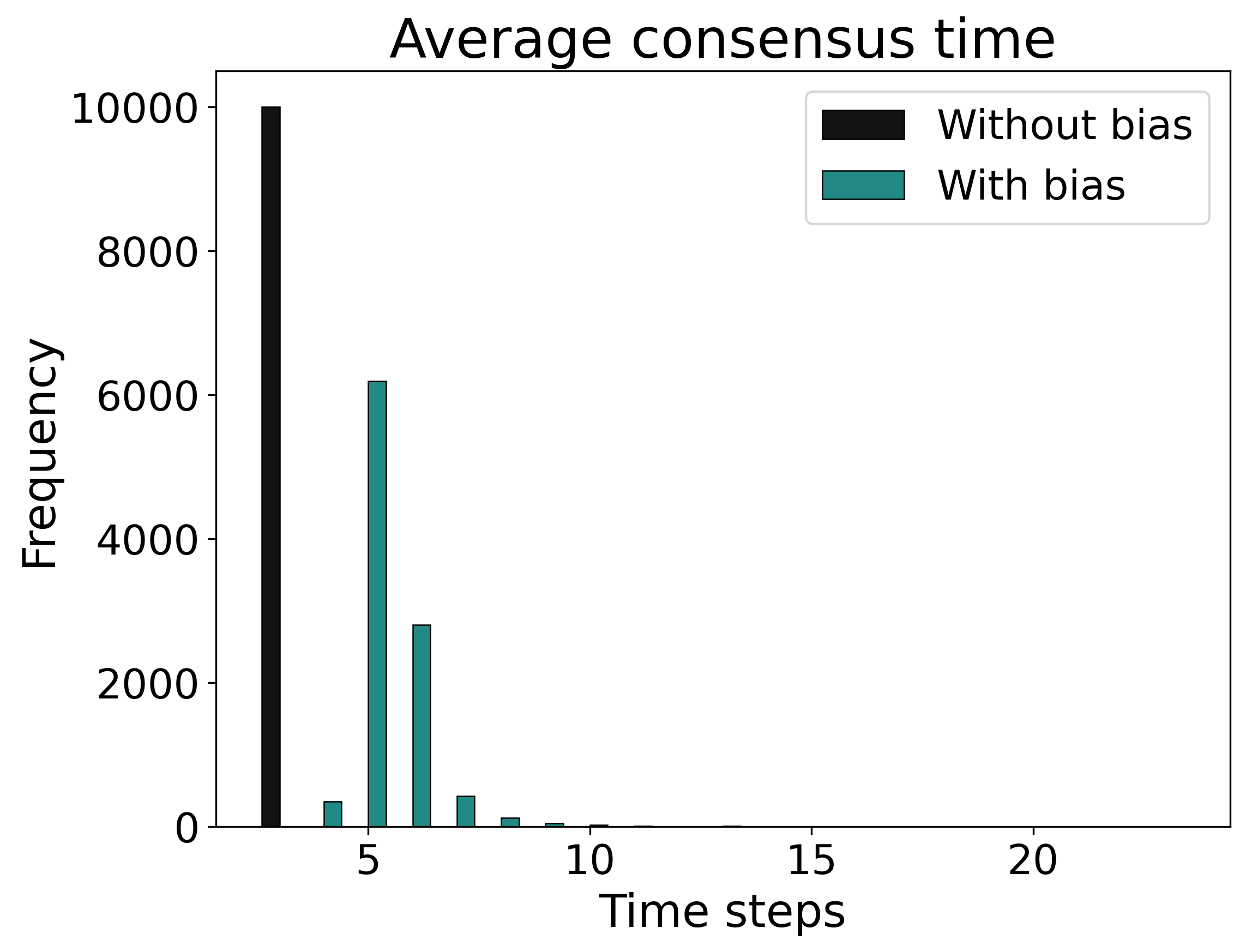}
  \end{subfigure}
  \hfill
  \begin{subfigure}[b]{0.45\textwidth}
    \centering
    \includegraphics[width=\linewidth]{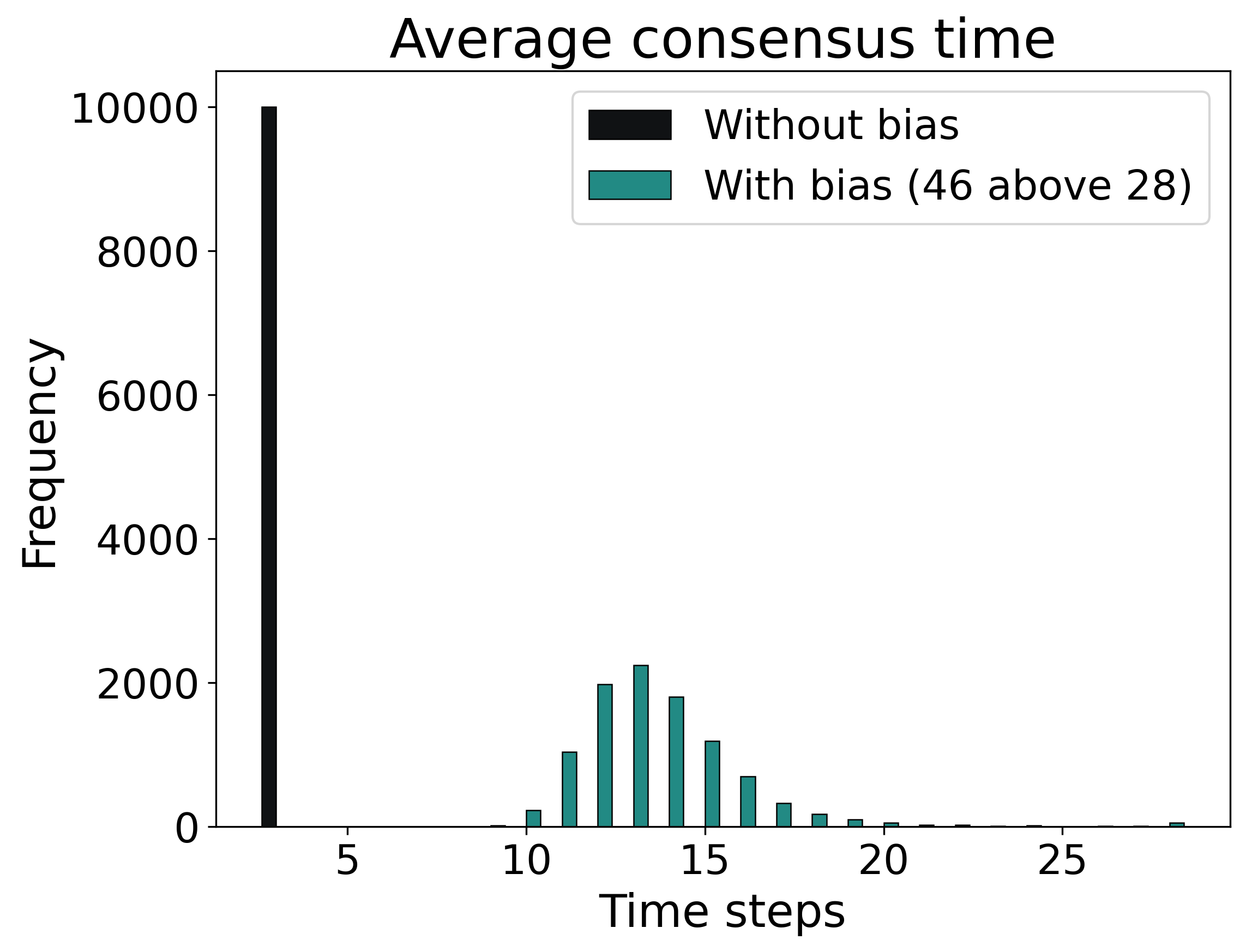}
  \end{subfigure}


  \centering
  \begin{subfigure}[b]{0.45\textwidth}
    \centering
    \includegraphics[width=\linewidth]{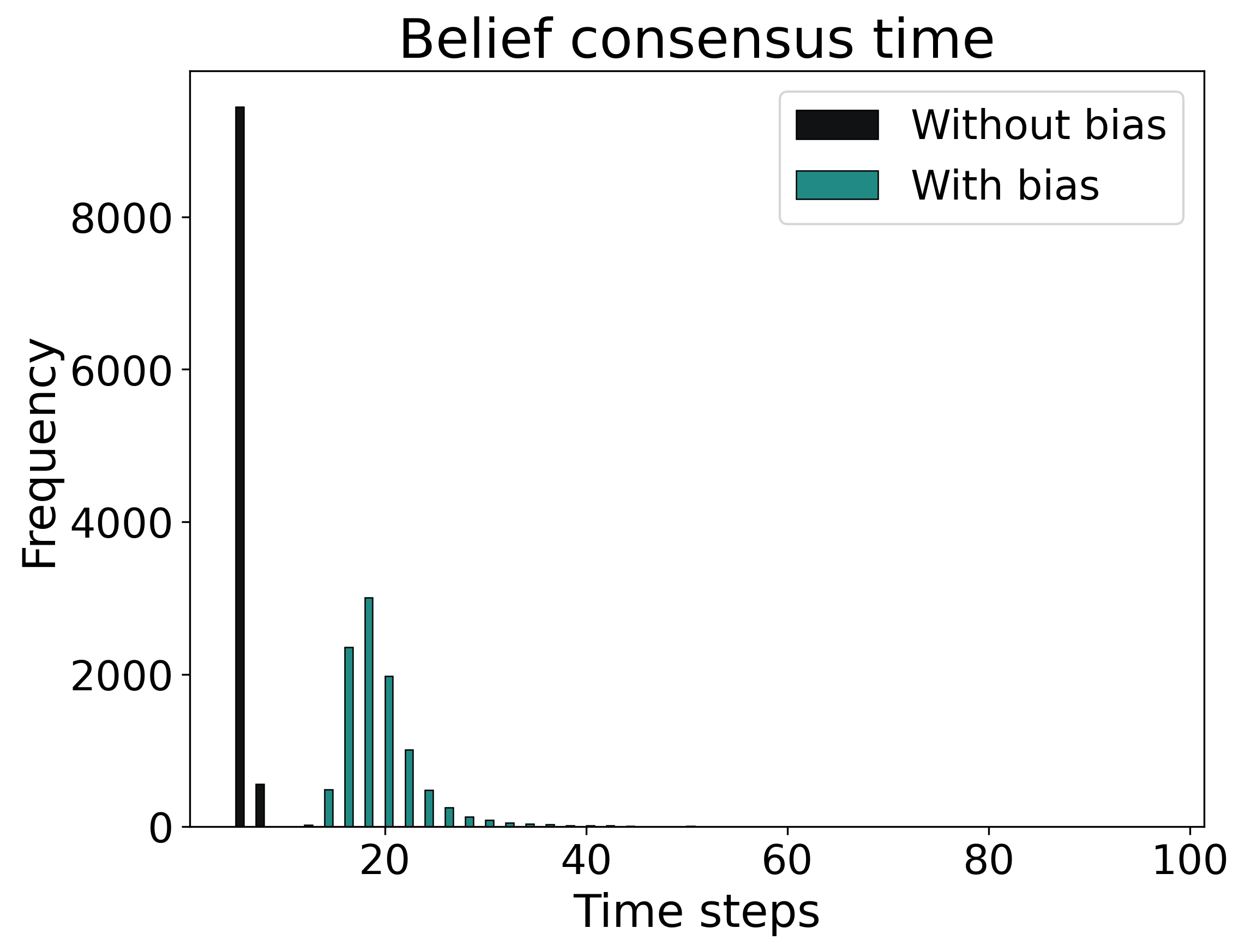}
  \end{subfigure}
  \hfill
  \begin{subfigure}[b]{0.45\textwidth}
    \centering
    \includegraphics[width=\linewidth]{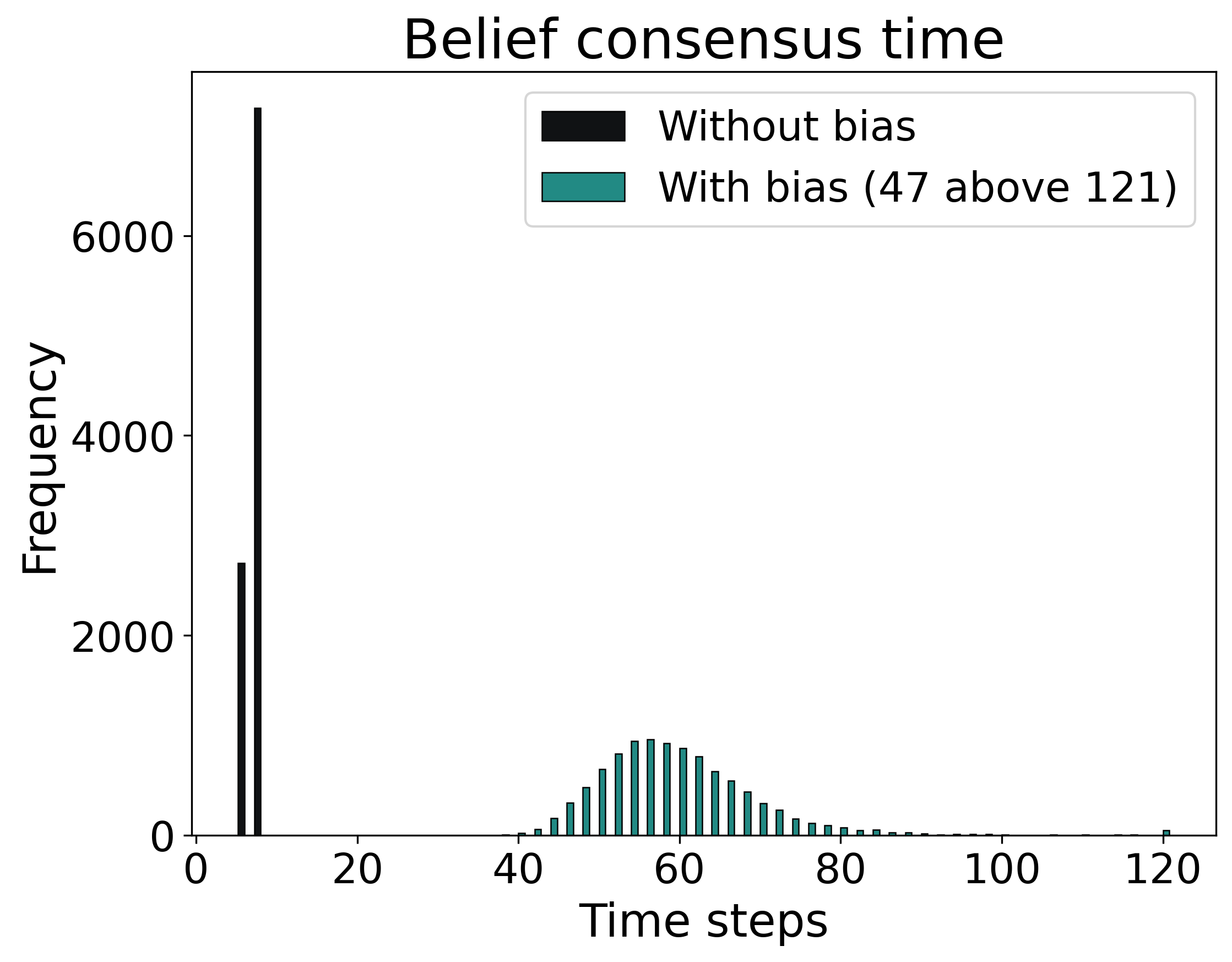}
  \end{subfigure}

  
  \begin{subfigure}[b]{0.45\textwidth}
    \centering
    \includegraphics[width=\linewidth]{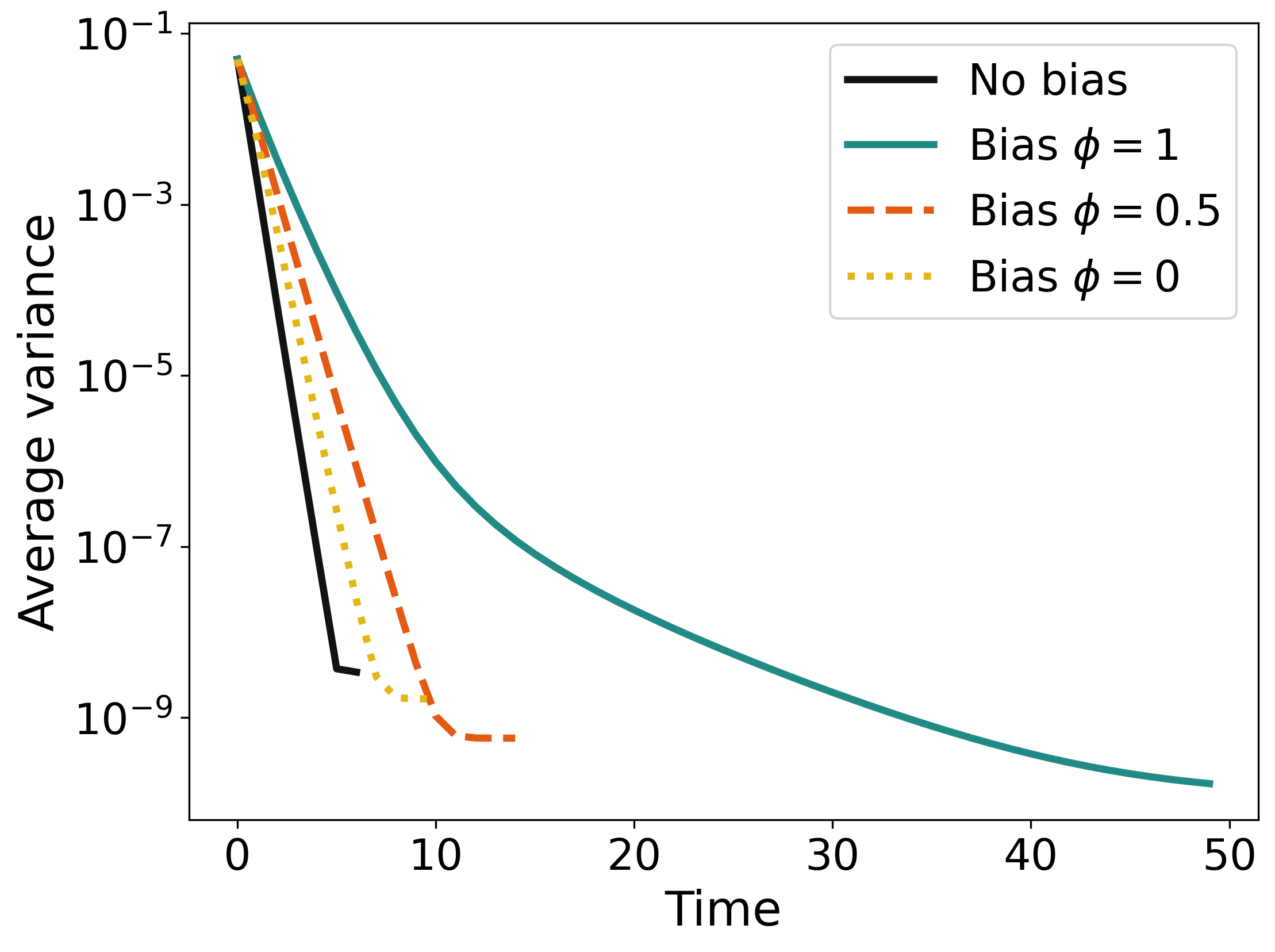}
  \end{subfigure}
  \hfill
  \begin{subfigure}[b]{0.45\textwidth}
    \centering
    \includegraphics[width=\linewidth]{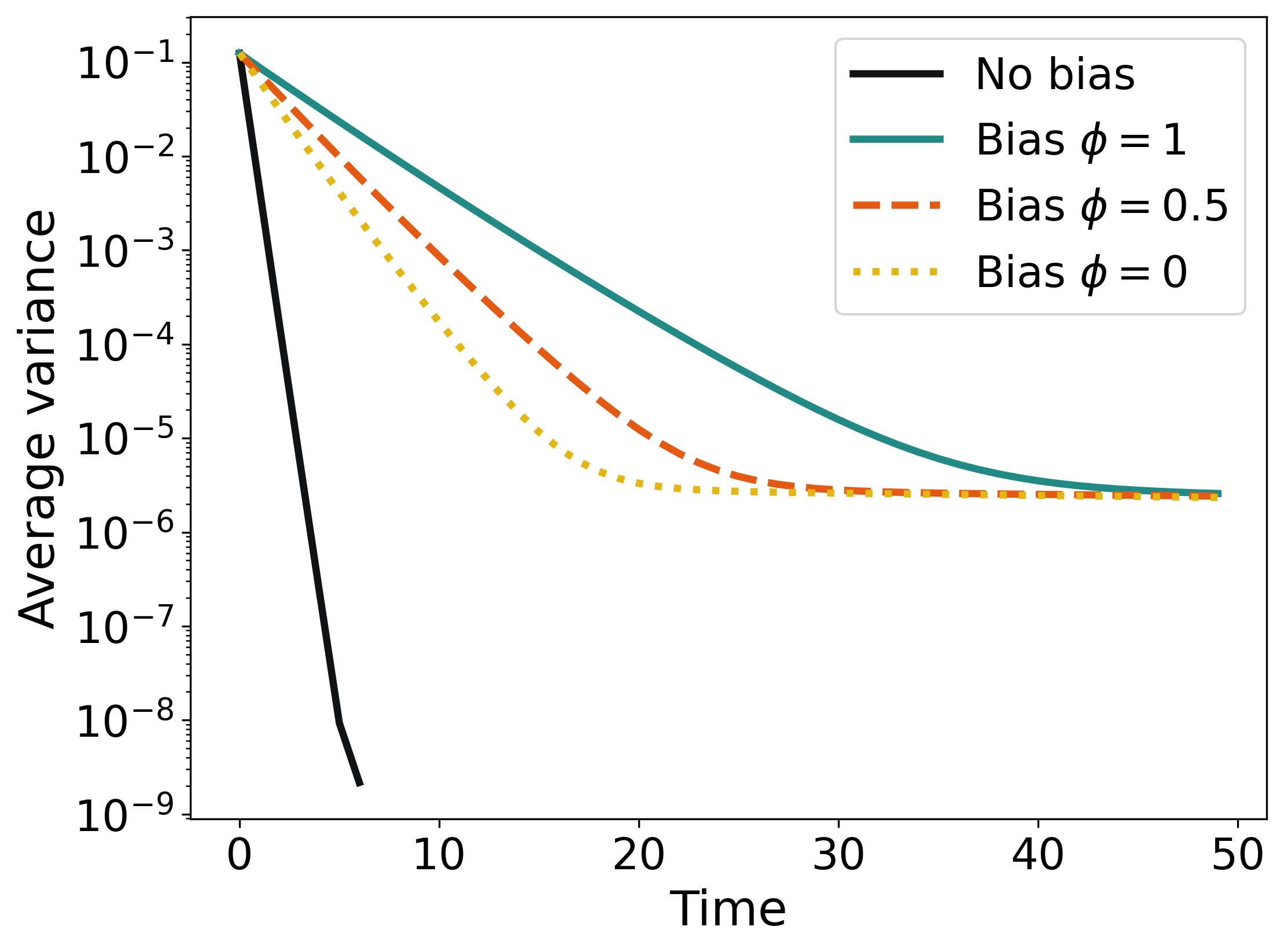}
  \end{subfigure}

  \caption{Left-hand column shows results for moderate beliefs: $Beta (2,2)$ distribution. Right-hand column shows results for extreme beliefs: $Beta (0.5,0.5)$ distribution. Top row: average consensus time; middle row: belief consensus time; bottom row: average variance of beliefs over time for the different values of $\phi$. All panels have $q=0.5$, network constructed using the Jackson-Rogers algorithm. }\label{fig:distribution_robustness}
\end{figure}

The comparison of the left- and right-hand columns of Figure \ref{fig:distribution_robustness} (top and middle rows) shows that the effects of confirmation bias on consensus time are much more pronounced when the initial beliefs are more extreme to begin with. Similarly, the impact on the variance of beliefs at any given point in time is much larger with extreme initial beliefs (bottom right) than with moderate initial beliefs (bottom left). 

This is because confirmation bias of any given strength leads to more link cutting when initial beliefs are more extreme. In turn, this link cutting leaves agents paying more attention to their own beliefs. This is exactly the mechanism that slows down learning that we discussed in Section \ref{main results}.

\section{Conclusion}\label{conclusion}
The advent of online social networks has dramatically expanded the number of people who shape our views. Another facet of this change is that it has become very easy to stop listening to someone; it suffices to ``unfollow'' them. A large number of studies in psychology show that confirmation bias is a powerful filter in how we process information, and, inevitably, the relevance of this filter increases when it becomes so easy to ignore discordant voices. This paper has incorporated confirmation bias into a model of social learning on a network to investigate its impact on the learning process and the political arena.

Confirmation bias has an unambiguously negative impact. An increase in confirmation bias slows down learning. In the political context, the presence of confirmation bias makes it possible for the worse candidate (given the available information) to win in a shock election result. Moreover, confirmation bias makes the ideology of fringe media organizations more extreme. At a methodological level, this paper showcased the potential of combining insights from behavioural economics and networks. These two methodologies share a common objective to incorporate realistic features of, respectively, the psychology of decision-making and social interactions into the standard economic framework. Our hope is that this contribution will encourage further studies at the intersection of these two literatures.

An interesting direction for future research is to explore how confirmation bias affects the impact of misinformation, where agents strategically misrepresent their initial beliefs to others. Absent confirmation bias, misinformation changes the consensus belief but not our measure of consensus time because it does not affect the network structure (see, e.g., \cite{acemoglu2010spread}). The introduction of confirmation bias adds a new strategic element. Because the initial beliefs that agents choose to share influence the network structure, agents have a novel incentive to misinform others of their initial belief in the `wrong direction'. An examination of these decisions may yield new insights about misinformation in society.

\clearpage
\singlespacing
\bibliographystyle{abbrvnat}
\addcontentsline{toc}{section}{References}
\bibliography{general_bib,maths_bib}
\onehalfspacing
\appendix
\numberwithin{equation}{section}
\numberwithin{thm}{section}
\numberwithin{prop}{section}
\numberwithin{defn}{section}
\numberwithin{lem}{section}

\newpage
\section{Proofs}\label{proofs}
We begin by setting out some notation. $\lambda_{k}$ is the $k$th eigenvalue of the Markov chain $T$, where eigenvalues have been ordered such that $\lambda_{1} > \lambda_{2} > ... > \lambda_{n}$. Similarly, $\lambda_{k}^{*}$ is the $k$th eigenvalue of the Markov chain $T^{*}$. $v_{k}$ is the right-hand eigenvector of the Markov chain $T$ that corresponds to the eigenvalue $\lambda_{k}$. $\mathcal{E}(f, T)$ is the Dirichlet energy functional of $T$. $s$ is the left-hand eigenvector that corresponds to the eigenvalue $\lambda_1 \equiv 1$, with $i$th element $s_i$. It is also called the \emph{influence} vector, and $s_i$, the influence of agent $i$. $f$ is an \emph{arbitrary} vector of length $n$.

\paragraph{Proof of Proposition \ref{critical q}.}
First, suppose that $q = 0$. Then $|x_{i0} - x_{j0}| \leq 1$ for all possible realizations of $x_{i0}$ and $x_{j0}$, so no links are re-routed. Hence $T^{*}(T,q=0) = T$, so is strongly connected and aperiodic by assumption, and therefore there is a consensus.
	
Assume that $\overline{q}$ is the smallest value of $q$ for which consensus is \emph{not} reached. Note that $\overline{q}$ must exist, since consensus is never reached when $q=1$, and $\overline{q} = 1 - \frac{1}{2} \epsilon$ is permitted, where $\epsilon = \max_i \{ \min_j \{x_{i0} - x_{j0} \} \}$. This follows from Definition \ref{core rule} and the assumption that $|x_{i0} - x_{j0}| > \epsilon$ for some $i,j$. So there must be at least one agent who will stop listening to anyone else for some $q<1$.
	
If $T^{*}(T, \overline{q})$ does not reach a consensus, then it is straightforward to see by contradiction that there must exist a non-empty subset of agents $B \subset N$ such that $T^{*}(\overline{q})_{ij} = T^{*}(\overline{q})_{ji} = 0$ for all $i \in B$ and $j \in N \setminus B$. Suppose this were not true, then either $T^{*}(T, \overline{q})$ is strongly connected, or every strongly connected component of $T^{*}(T, \overline{q})$ listens indirectly to a single strongly connected component. In either case, we get a contradiction because the assumption that $T$ is aperiodic ensures there is a consensus.
	
By definition, increasing the strength of confirmation bias can never restore (or otherwise create) inter-agent links. So $T^{*}(\overline{q})_{ij} =0 \implies T^{*}(q')_{ij} = 0$ for all $q' \geq \overline{q}$ and for all $i \neq j$. Therefore, $T^{*}(T,q')$ does not reach a consensus either.
\hfill \qed

\paragraph{Some existing results.} Before proving Theorem \ref{speed result}, we need to state the following definition and results from the literature.

\begin{defn}\citep[Definition 4.1]{newpart3notes}\label{dirichlet energy}
The Dirichlet energy function of a network $T$ is
	\begin{align}
	\mathcal{E}(f,T) = \frac{1}{2} \sum_{i,j} ( f_{i} - f_{j} )^{2} s_{i} T_{ij}
	\end{align}
\end{defn}

\begin{thm}\citep[Corollary 4.7]{newpart3notes}\label{dirichlet eigenvalues}
Let $T$ and $\tilde{T}$ be two transition matrices reversible with respect to the same invariant distribution $s$. Let $\mathcal{E}(f,T)$ and $\mathcal{E}(f,\tilde{T})$ be their Dirichlet forms and $(\lambda_i)_i$ and $(\tilde{\lambda}_i)_i$ their respective eigenvalues. If there exists a positive constant $A$ such that for all $f$: $E \to \mathbb{R}$ we have $\mathcal{E}(f,\tilde{T}) \leq A \mathcal{E}(f,T)$, then $1 - \tilde{\lambda}_j \leq A(1 - \lambda_j)$ for all $j$.
\end{thm}

Note that an \emph{invariant distribution} of a transition matrix is simply its first left-hand eigenvector (that corresponds to the eigenvalue $\lambda_1 \equiv 1$) -- which we call the influence vector throughout this paper.


\begin{lem}\citep[Lemma 3.5]{newpart3notes}\label{lemma3.5}
	\begin{align}
	|| T^{t}(i, \cdot) - s ||_{2}^{2} = \sum_{k = 2}^{n} v_{k}(i)^{2} \lambda_{k}^{2t}
	\end{align}
\end{lem}
To avoid confusion, we use $v_{k}(i)$ to denote the $i$th element of the eigenvector $v_{k}$ (which is itself the $k$th eigenvector of the matrix), and $T^{t}(i, \cdot)$ to denote the $i$th row of the matrix $(T)^{t}$.

\begin{thm}\citep[Theorem 1]{demarzo2003persuasion}\label{rem:positive influence}
	Suppose a Markov chain $T$ is strongly connected. Then there exists a unique $(1 \times n)$ row vector $s$ such that $s T = s$, and $s_i > 0$ for all $i$. 
\end{thm}

\paragraph{Two building lemmas.} Now we prove two lemmas regarding reversible networks.
\begin{lem}\label{lem:unweight equiv}
	Suppose $T$ is strongly connected and reversible. Then $T_{ij} > 0 \iff T_{ji} > 0$
\end{lem}
\begin{proof}[\emph{\textbf{Proof of Lemma \ref{lem:unweight equiv}.}}]
	Suppose that $T_{ij} > 0$. Theorem \ref{rem:positive influence} states that $s_i > 0$ for all $i$, so it must be the case that $s_i T_{ij} > 0$. In turn, the definition of reversibility then implies that $s_j T_{ji} > 0$. Since, by Theorem \ref{rem:positive influence}, $s_j >0$ then  $T_{ji} > 0$. This holds for all $i,j$.
	
	Instead, suppose that $T_{ij} = 0$. It clearly follows that $s_i T_{ij} = 0$. The definition of reversibility then requires $s_j T_{ji} = 0$. It follows from Theorem \ref{rem:positive influence} that $s_j > 0$, therefore $T_{ji} = 0$. This holds for all $i,j$.
\end{proof}

\begin{lem}\label{lem:reversibility influence}
	Suppose $q < \overline{q}(T,x_0)$, $T$ is strongly connected and reversible, $s$ is the influence vector. Then $T^*$ is reversible with an influence vector $s$.
\end{lem}
\begin{proof}[\emph{\textbf{Proof of Lemma \ref{lem:reversibility influence}.}}]
\textbf{Step 1.} We first show that re-routing one pair of reciprocal links (in line with confirmation bias) preserves reversibility and does not change influence. Let $\tilde{T} \equiv T - T_{ab} - T_{ba} + T'_{aa} + T'_{bb}$ with $|T_{ab}| = |T'_{aa}|$ and $|T_{ba}| = |T'_{bb}|$. Clearly (1) $\tilde{T}_{aa} = T_{aa} + T_{ab}$, (2) $\tilde{T}_{bb} = T_{bb} + T_{ba}$, (3) $\tilde{T}_{ab} = 0$, and (4) $\tilde{T}_{ba} = 0$. (5) For all other pairs of agents $\{i,j\} \neq \{a,b\}$, we have $\tilde{T}_{ij} = T_{ij}$. Note that (5) allows one agent from the pair $\{i,j\}$ to be agent $a$ [resp. $b$], so long as the other agent is not agent $b$ [resp. $a$]. (1)-(5) cover all links in the network. 

\textbf{Influence.} Re-routing the pair of reciprocal links between agents $a$ and $b$ does not affect the links of any agent $i \neq a,b$. This follows from the fact that $\tilde{T}_{ji} = T_{ji}$ for all $\{i,j\} \neq \{a,b\}$. Consequently, the influence vector $s$ solves the influence equations for all $i \neq a,b$. Formally: $s_i = \sum_j \tilde{T}_{ji} s_j$. 
Now consider agent $a$. We know that in network $T$: $s_a= \sum_j T_{ja} s_j = \sum_{j \neq a,b} T_{ja} s_j + T_{ba} s_b + T_{aa} s_a$. From the definition of reversibility we get: $T_{ba} = T_{ab} \frac{s_a}{s_b}$. Substituting this into the equation for $s_a$ yields: $s_a = \sum_{j \neq a,b} T_{ja} s_j + T_{ab} \frac{s_a}{s_b} s_b + T_{aa} s_a = \sum_{j \neq a,b} T_{ja} s_j + (T_{ab} + T_{aa}) s_a$. Substituting in facts (1)-(4) from above yields $s_a = \sum_{j \neq a} \tilde{T}_{ja} s_j + \tilde{T}_{aa} s_a = \sum_j \tilde{T}_{ja} s_j$. 
Therefore $s$ solves the influence equation for agent $a$ in network $\tilde{T}$. Identical reasoning shows that $s$ also solves the influence equation for agent $b$ in network $\tilde{T}$. Therefore $s_i = \sum_j \tilde{T}_{ji} s_j$ for all $i$. So $s \tilde{T} = s$, and $s$ is the unique influence vector by Theorem \ref{rem:positive influence}.

\textbf{Reversibility.} $T$ is reversible by assumption, and by Lemma \ref{lem:unweight equiv} $T_{ij} > 0 \iff T_{ji} > 0$ for all $i,j$. Therefore $\tilde{T}_{ij} = T_{ij}$ for all $i,j$ other than $\tilde{T}_{ab}$ and $\tilde{T}_{ba}$, where $\tilde{T}_{ab} = 0$ and $\tilde{T}_{ba} = 0$. Using the fact that influence is the same in $T$ and $\tilde{T}$ yields $s_i \tilde{T}_{ij} = s_j \tilde{T}_{ji}$ for all $i,j$. So $\tilde{T}$ is reversible. \\

\textbf{Step 2.} We now show that this extends to all links re-routed by confirmation bias. First note that $q < \overline{q}(T,x_0)$ and all links are reciprocated (by Lemma \ref{lem:unweight equiv}), although weights need not be the same in each direction. Therefore $T^*$ is strongly connected. Clearly it is possible to write $T^*$ as $T$ minus links that are cut, plus the additional strength of the self-links. Obviously a finite number of links are cut, so:
\begin{align*}
	T^* &= T - T_{a_1 b_1} - T_{b_1 a_1} + T'_{a_1 a_1} + T'_{b_1 b_1} - ... - T_{a_r b_r} - T_{b_r a_r} + T'_{a_r a_r} + T'_{b_r b_r} - ... \\
	&\qquad -  T_{a_R b_R} - T_{b_R a_R} + T'_{a_R a_R} + T'_{b_R b_R}
\end{align*}
with $|T_{a_r b_r}| = |T'_{a_r a_r}|$ and $|T_{b_r a_r}| = |T'_{b_r b_r}|$ for all $r \leq R$, for some $R \in \mathbb{Z}$. \\

By step 1, $T^{1} \equiv T - T_{a_1 b_1} - T_{b_1 a_1} + T'_{a_1 a_1} + T'_{b_1 b_1}$ is reversible with influence vector $s$. Therefore, by applying the logic from step 1 again, $T^{2} \equiv T^{1} - T_{a_2 b_2} - T_{b_2 a_2} + T'_{a_2 a_2} + T'_{b_2 b_2}$ is also reversible with influence vector $s$. Repeated application of this result finds that $T^{R}$ (defined recursively in the same fashion as $T^{2}$ etc.) is also reversible with influence vector $s$. Finally, notice that $T^* \equiv T^{R}$.

Since $T^* \equiv T^{R}$ is strongly connected, then $T^{r}$ for all $r <R$ must also be strongly connected -- as they involve adding inter-agent links relative to $T^* \equiv T^{R}$.
\end{proof}

\paragraph{Disambiguation:} In the proof to Theorem \ref{speed result}, we use subscripts to denote two different things. Subscripts on $T$, $s$, and $f$ (i.e., $T_{ij}$, $s_{i}$, $f_{i}$) denote the $i$th (or $ij$th) element of the vector/matrix. Subscripts on $v$ and $\lambda$ (i.e., $v_{k}$, $\lambda_{k}$) denote the $k$th ordered eigenvector / eigenvalue. So $v_{k}$ refers to a whole eigenvector (an $n \times 1$ vector), \textit{not} a single element of it.

\paragraph{Proof of Theorem \ref{speed result}.}
Using Definition \ref{dirichlet energy} we have: $\mathcal{E}(f,T^{*}) = \frac{1}{2} \sum_{i,j} ( f_{i} - f_{j} )^{2} s_{i}^{*} T_{ij}^{*}$. 
Since $T$ is reversible, by Lemma \ref{lem:reversibility influence}, then $s_{i}^{*} = s_{i}$ for all $i$, so $\mathcal{E}(f,T^{*}) = \frac{1}{2} \sum_{i,j} ( f_{i} - f_{j} )^{2} s_{i} T_{ij}^{*}$. 
By Definition \ref{core rule}: $T_{ij}^{*} \leq T_{ij}$ for all $i$, $j$ with $i \neq j$. Terms where $i=j$ drop out of the summation, as $(f_{i} - f_{j}) = (f_{i} - f_{i}) = 0$.
Therefore, $\mathcal{E}(f, T^{*}) \leq \mathcal{E}(f,T)$ for all $f$ and any reversible $T$ and associated $T^{*}$.

Since both $T$ and $T^*$ have the same influence vector (by Lemma \ref{lem:reversibility influence}), we can apply Theorem \ref{dirichlet eigenvalues}. From immediately above, we have $\mathcal{E}(f, T^{*}) \leq \mathcal{E}(f,T)$ for all $f$. Therefore, Theorem \ref{dirichlet eigenvalues} yields $1 - \lambda_{j}^{*} \leq 1 - \lambda_{j}$ for all $j$ (note that we are setting $A=1$ from Theorem \ref{dirichlet eigenvalues}). This establishes that all eigenvalues are weakly monotonically increasing in $q$.

Simple manipulation of Lemma \ref{lemma3.5} yields $\frac{1}{n} \sum_{i} || T^{t}(i, \cdot) - s ||_{2}^{2} = \frac{1}{n} \sum_{k = 2}^{n}  \sum_{i} v_{k}(i)^{2} \lambda_{k}^{2t}$. By our normalizations, the squared elements of each eigenvector sum to 1: $\sum_{i} v_{k}(i)^{2} = 1$ for all $k$. Therefore $\frac{1}{n} \sum_{i} || T^{t}(i, \cdot) - s ||_{2}^{2} = \frac{1}{n} \sum_{k = 2}^{n} \lambda_{k}^{2t}$. So we can write our convergence metric $\tau$ as:
	\begin{align}
\tau = \min \left\{ t > 0 : \frac{1}{n} \sum_{k = 2}^{n} \lambda_{k}^{2t} < \epsilon \right\}
	\end{align}
It is clearly the case that $\frac{1}{n} \sum_{k = 2}^{n} \lambda_{k}^{2t}$ is decreasing in $t$ and increasing in $\lambda_{k}$ for all $k \in \{ 2 , ... , n \}$. We know from earlier that $q' > q \implies \lambda_{k}^{'} \geq \lambda_{k}$ for all $k \in \{ 2 , ... , n \}$. Therefore $q' > q \implies \frac{1}{n} \sum_{k = 2}^{n} (\lambda_{k}^{'})^{2t}  \geq \frac{1}{n} \sum_{k = 2}^{n} \lambda_{k}^{2t}$. This follows from the fact that all eigenvalues are weakly positive. To see this, note that if $T_{ii} \geq \frac{1}{2}$ for all $i$, then there exists another Markov chain $\tilde{T}$ such that $T = \frac{1}{2} (I + \tilde{T})$, where $I$ is the identity matrix. All eigenvalues of $I$ are equal to $1$, and since $\tilde{T}$ is a Markov chain, all of its eigenvalues are weakly greater than $-1$.

So the minimum time taken to achieve $\frac{1}{n} \sum_{k = 2}^{n} \lambda_{k}^{2t} < \epsilon$ is weakly monotonically increasing in $q$ for any fixed $\epsilon > 0$.
\hfill \qed

Notice that this proof only relies on some off-diagonal elements of the matrix decreasing, and some corresponding diagonal elements increasing (and there being no other changes). Therefore, this proof clearly holds under the generalized model (set out in Definition \ref{generalized core rule 2}).

\paragraph{Proof of Remark \ref{equal influence no change}.}
This is a restatement of Lemma \ref{lem:reversibility influence}.
\hfill \qed 



\paragraph{Proof of Remark \ref{single agent influence}.}
\citet[Section 3, equation 8]{schweitzer1968} provides an equation for how the influence of an agent changes following a perturbation in a single row: $s_{i}^{*} = s_{i} \left( 1 + \frac{U_{ii}}{1 - U_{ii}} \right)$, where $U = (T^{*} - T) Z$, and $Z = (I - T - \mathbf{1}s)^{-1}$ is the fundamental matrix of the Markov chain $T$.

When only agent $i$ cuts links, then $T^{*}$ differs from $T$ only in the $i$th row, and so we can apply this formula. Further, when $T^{*}$ differs from $T$ only by entries in the $i$th row, then $$U_{ii} = \sum_{k \in N} (T^*_{ik} - T_{ik}) Z_{ki} = (T^{*}_{ii} - T_{ii}) Z_{ii} + \sum_{j \neq i} (T_{ij}^{*} - T_{ij}) Z_{ji}.$$ 

Let $D_i = \{ j : |x_{i0} - x_{j0}| > 1-q\}$ be the set of agents whose links from $i$ are cut. For $j \in D_{i}$, we have $T^*_{ij} = 0$ and for $j \notin D_i$ (and $j \neq i$), we have $T^*_{ij} = T_{ij}$. So the sum over $j \neq i$ is therefore restricted to $j \in D_i$. Furthermore, the weight of the severed links is transferred to the self-link. Therefore we have $T^*_{ii} - T_{ii} = \sum_{j \in D_i} T_{ij}$. 

Substituting these facts back into the expression for $U_{ii}$ yields:

\[ U_{ii} = \sum_{j \in D_i} T_{ij} Z_{ii} + \sum_ {j \in D_i} (- T_{ij}) Z_{ji} = \sum_{j \in D_i} T_{ij} (Z_{ii} - Z_{ji}). \]
Since $i$ cuts one or more links, the set $D_i$ is non-empty and $T_{ij} > 0$ for all $j \in D_i$. \citet[Section 2, equation 8]{conlisk1985comparative} shows that $Z_{ii} > Z_{ji}$ for all $j \neq i$. Therefore, each term in the summation is strictly positive, which implies $U_{ii} > 0$. Finally, recall that $s_{i}^{*} = s_{i} \left( 1 + \frac{U_{ii}}{1 - U_{ii}} \right)$. Hence $s^*_{i} > s_i$.
\hfill \qed

\paragraph{Simplified learning rule: a lemma.} Before proving Proposition \ref{voting result}, we need to state and prove the following Lemma.

\begin{lem}\label{eq with mean field}
Under Assumption \ref{ass:1}, the learning rule for an agent $i$ can be expressed as $x_{i t} = (1 - a^t) \, \mu + a^t \, x_{i 0} $ for all $i$, where $\mu = \frac{1}{n} \sum_{j} x_{j 0}$ and $a = \frac{n\phi - 1}{n-1} \in [0,1]$.
\end{lem}

\begin{proof}
Proof by induction. First, \textbf{show for $t = 1$.} By the definition of covariance: $cov(x_{j 0} , T_{ij}) = \frac{1}{n-1} \sum_{j \neq i } (x_{j 0} - \overline{x}_{j 0}) (T_{ij} - \overline{T}_{ij})$, where $\overline{x}_{j 0} = \frac{1}{n-1} \sum_{j \neq i} x_{j 0} = \frac{n \mu - x_{i 0}}{n-1}$ and $\overline{T}_{ij} = \frac{1}{n-1} \sum_{j \neq i} T_{ij} = \frac{1 - \phi}{n-1}$ are the means taken over all $j \neq i$. Which rearranges to
\begin{align*}
(n-1) cov(x_{j 0} , T_{ij}) = \left( \sum_{j \neq i } x_{j 0} T_{ij} \right) - (n-1) \, \overline{x}_{j 0} \overline{T}_{ij}.
\end{align*}
By Assumption \ref{ass:1}, $cov(x_{j 0} , T_{ij}) = 0$. Therefore
\begin{align*}
\sum_{j \neq i } x_{j 0} T_{ij} = (n-1) \, \overline{x}_{j 0} \overline{T}_{ij} = \frac{1 - \phi}{n-1} \left( n \mu - x_{i 0} \right).
\end{align*}
The learning rule in Section \ref{model} can trivially be decomposed to $x_{i 1} = \sum_{j \neq i } T_{ij} x_{j 0}	+ T_{ii} x_{i 0}$. Substituting our result from immediately above into the learning rule, and recalling that $T_{ii} = \phi$ for all $i$ by assumption, yields
\begin{align*}
x_{i 1} = \frac{n (1 - \phi)}{n-1} \, \mu + \left( \phi - \frac{1 - \phi}{n-1} \right) x_{i 0} = (1 - a) \mu + a \, x_{i 0},
\end{align*}
as required, since $1 - a = \frac{n (1 - \phi)}{n-1}$. Note that $a \geq 0$ because $\phi \geq \frac{1}{n}$, and $a \leq 1$ because $\phi \leq 1$.

\textbf{Now we assume that for $t=k$: $x_{i k} = (1 - a^k) \mu + a^k x_{i 0}$ for all $i$, and prove for $t=k+1$}. The learning rule for $t=k+1$ is $x_{i k+1} = \sum_{j} T_{ij} x_{j k}$. Substituting in our assumption for $t=k$, and using the fact that $\sum_{j} T_{ij} = 1$:
\begin{align*}
    x_{i k+1} = \sum_{j} T_{ij} \left( (1 - a^k) \mu + a^k x_{j 0} \right) = (1 - a^k) \mu + a^k \sum_{j} T_{ij} x_{j 0}.
\end{align*}
From the $t=1$ step, $\sum_{j} T_{ij} x_{j 0} = x_{i 1} = (1 - a) \mu + a \, x_{i 0}$. So we have:
\begin{align*}
    x_{i k+1} &= (1 - a^k) \mu + a^k (1 - a) \mu + a^{k+1} x_{i 0} \\
    \implies x_{i k+1} &= (1 - a^{k+1}) \mu + a^{k+1} x_{i 0}
\end{align*}
\end{proof}

\paragraph{Proof of Proposition \ref{voting result}.}
First, recall that $\mu < 0.5$ and $\# [x_{i 0} < 0.5] > \# [x_{i 0} > 0.5]$ by assumption. The Left wins immediately and in the long run.
By Lemma \ref{eq with mean field}, when $q=0$ we have the simplified learning rule: $ x_{i t} =  (1 - a^{t}) \mu + a^{t} x_{i 0}$ for all $i$, where $a = \frac{n\phi - 1}{n-1} \in [0,1]$. Let $\Delta x_{i t} = x_{i t} - x_{i t-1}$ for $t \geq 1$. Substituting in the simplified learning rule and rearranging yields $\Delta x_{i t}	= a^{t-1} (1-a)(\mu - x_{i 0})$.
$t$ only enters as the exponent on a weakly positive number.\footnote{In the special case where $\phi = \frac{1}{n}$ (so that $a = 0$), agent $i$ will arrive at the long-run social consensus after only one period of learning. This is because, by assumption, their neighbours are representative of society as a whole.} Therefore, the sign of $\Delta x_{i,t}$ is the same for all $i,t$.

Therefore, $|x_{i t} - \mu|$ is weakly monotonically decreasing in $t$, for all $i$. This means that no agent who votes for the Left-wing candidate can ever change their vote (and all agents who vote for the Right-wing candidate must switch their vote to the Left-wing candidate at some point). Therefore the Left-wing candidate would win for any $t$, and so there cannot be a shock election.

To prove that a shock election is possible with $q>0$ we show sufficient conditions for the Right-wing candidate to win at $t=1$. Since the Left-wing candidate wins at $t=0$ and $t=\infty$ by assumption, this is sufficient to prove that a shock election is possible.

Choose $q \in ( 1 - \max \{ k_{m+1} - k_{m} \} , 1 - \min \{ k_{m+2} - k_{m} \} )$. Then we have:
\begin{align*}
x_{EL 1} 	&= f_{CL} x_{CL 0}   + (1 - f_{CL}) x_{EL 0} 	 \\
x_{CL 1} 	&= f_{EL} x_{EL 0}   +  f_{S} x_{S 0} + (1 - f_{EL} - f_{S}) x_{CL 0} \\
x_{S 1}	 	&= f_{CL} x_{CL 0}   +  f_{CR} x_{CR 0} + (1 - f_{CL} - f_{CR}) x_{S 0} 	\\
x_{CR 1} 	&= f_{S} x_{S 0}   +  f_{ER} x_{ER 0} + (1 - f_{S} - f_{ER}) x_{CR 0} 	 \\
x_{ER 1} 	&= f_{CR} x_{CR 0}   + (1 - f_{CR}) x_{ER 0}
\end{align*}
Now suppose that $f_{EL} + f_{CL} < 0.5 $ (the Left do not form a majority on their own) and that $f_{CR} (x_{CR 0} - 0.5) > f_{CL} (0.5 - x_{CL 0})$ (the centre right can persuade the swing voters more than the centre left). Then $x_{S 1} > 0.5$, and the Right-wing candidate wins at $t=1$.
\hfill \qed

\paragraph{Proof of Proposition \ref{media monotone m}.} The proof is an application of \citet[p.~31]{eaton1975}.

\paragraph{Proof of Proposition \ref{media monotone q}.}
\textbf{Case 1: `low' confirmation bias.} From \citet{eaton1975} we have $\rho_{fr}(q=0, M) = \frac{1}{2M-4}$. When $(1-q) > \frac{1}{2M-4}$ then a fringe media organization does not lose any (potential) agents due to confirmation bias. For $q \in [0 , \frac{2M-5}{2M-4} ]$, the fringe media organization is unaffected by confirmation bias, and so its ideology is unaffected by confirmation bias in this range.

\textbf{Case 2: `high' confirmation bias.} When $q \in [\frac{2M-5}{2M-4} , 1]$, we have $\rho_{fr}(q, M) = (1-q)$. \textbf{Observation 1:} $\rho_{fr}(q, M) > (1-q)$ is not possible. The fringe media organization can shift leftwards and pick up at least as many listeners to its left as it would lose to its right, meaning that it weakly prefers to shift to the left. Since this leftward shift does not reduce the audience of any other media organization, then it remains an equilibrium. This logic (implicitly) requires that the distribution of agents' initial opinions is uniform (or close to uniform).

\textbf{Observation 2:} $\rho_{fr}(q, M) < (1-q)$ is also not possible. First, a media organization could only choose $\rho < (1-q)$ if another organization also made the same choice (otherwise it would choose to shift inwards). Supposing a pair of media organizations chose $\rho_{1} = \rho_{2} = \chi < (1-q)$, then there must be another organization choosing $\rho_{3} = 3 \chi$ -- otherwise one of the fringe media organizations would deviate to $\rho = \chi + \epsilon$, for some $\epsilon > 0$.
The media organization $m=3$ must choose this ideology to prevent $m=1,2$ from deviating. In turn, this requires that $\rho_{4} = 5 \chi$ -- to prevent $m=3$ from deviating. This continues across the spectrum of media organizations: $\rho_{m} = (2m - 3)\chi$ for $m \in \{ 3 , ... , M-2 \}$.

Note that for $\rho_{fr}(q, M) < (1-q)$ to be sustainable as an equilibrium, there must also be a pair of media organizations choosing $\rho = 1 - \chi$ (an equally extreme ideology at the other end of the spectrum) -- otherwise a fringe media organization could profitably deviate to $\rho = 1 - \chi + \epsilon$, for some $\epsilon > 0$. By identical logic to above, we require $\rho_{M-2} = 1 - 3 \chi$.
Immediately above we found that $\rho_{M-2} = 1 - 3\chi$, but in the paragraph above found that $\rho_{M-2} = (2M - 7)\chi$. For these to equate, we must have $\chi = \frac{1}{2M - 4}$. But this is a contradiction -- since we assumed that $\chi < \frac{1}{2M - 4}$.\footnote{We did this by setting $\chi < (1-q) \leq  \frac{1}{2M - 4}$ in the third paragraph.}

\textbf{Observation 3:} It now remains to note that $\rho_m^* = q + (m-1) \frac{1-2q}{M-1}$ for $m \in \{1,..,M\}$ \emph{is} an equilibrium. This places the most extreme organizations at $\rho_1^* = q, \rho_M^* = 1 - q$, and has the $M$ media organizations space themselves evenly across the available ideological spectrum $[1-q, q]$. Hence we have $\rho_{fr}(q,M) = \min\{\rho_m^*\} = \rho^*_M = 1-q$, and such an equilibrium exists.
\hfill \qed

\clearpage
\section*{Social networks, confirmation bias and shock elections: Online Appendix}
\addtocounter{section}{1}

This appendix contains generalizations and extensions to the main model. Section \ref{generalised model} generalizes the model and shows that our main result holds in this case. Section \ref{speed of learning appendix} provides further discussion of the main result and shows that it still holds using alternative definitions of consensus time. Section \ref{influence appendix} characterizes a special case where we can determine changes in influence. Section \ref{media appendix} provides a characterization of the media landscape when there are fewer than 6 organizations. Finally, section \ref{simulations appendix} shows additional graphs from the simulations presented in section \ref{simulations results}.

\subsection{Generalized model}\label{generalised model}

The model in Section \ref{model} assumes that confirmation bias causes agents to cut links \emph{completely}. This section shows that the main result, Theorem \ref{speed result}, is robust to relaxing this assumption. The following is a generalized version of the definition of confirmation bias. \\

\begin{defn}\label{generalized core rule 2}
	A society on a network $T$ in which agents have confirmation bias $q$ communicates according to a network $T^{*}$ such that:
	\begin{align*}
		\text{if} \quad |{x_{i 0} - x_{j 0}}| > (1-q) \quad &\text{then} & & T^{*}_{i j}  = ( 1 - \alpha_{ij} ) T_{i j} & & T^{*}_{i i} = T_{i i} + \sum_{j : |{x_{i 0} - x_{j 0}}| > (1-q)} \alpha_{ij} T_{i j} \\
		&\text{otherwise}  & & T^{*}_{i j}  = T_{i j}  	& & T^{*}_{i i} = T_{i i}
	\end{align*}
	where $\alpha_{ij} \in [0,1]$ and $\alpha_{ij} = \alpha_{ji}$ for all $i$ and $j$.
\end{defn}

Confirmation bias no longer reduces link weights to zero, but rather reduces them by a link-specific factor $(1 - \alpha_{ij})$, such that $\alpha_{ij} = \alpha_{ji}$. The other assumptions from Section \ref{model} remain unchanged. With this generalized definition, we can now state the generalized result.

\begin{thm}
	If $T_{ii} \geq \frac{1}{2}$ for all $i$, then, for any reversible network $T$ and any $\epsilon > 0$, the average consensus time $\tau(\epsilon)$ is (weakly) monotonically increasing in $q$, and in $\alpha_{ij}$ for all $i,j$.
\end{thm}

\begin{proof}
	Follows straightforwardly from the proof to Theorem \ref{speed result}. Notice that the extent of link-weakening must be the same for both $T_{ij}$ and $T_{ji}$ to ensure that reversibility is preserved. This means agents' influences are not affected by confirmation bias -- a key simplification required for tractability.
\end{proof}

The proof to Theorem \ref{speed result} shows that the \emph{average consensus time}, $\tau(\epsilon)$, increases when link weight is moved from the between-agent links to agents' self-links (mathematically, this is moving weight from the off-diagonal elements to the diagonal elements of the matrix $T$) and there are \emph{no other changes} to the network. It does not rely on any other structure to the changes to the network. Theorem \ref{speed result} is therefore effectively proved as a special case, and the result holds under the generalization here without any change to the proof. The intuition is also unchanged -- society converges to a consensus more slowly when agents are less willing to pay attention to those with different beliefs, and instead pay more attention to themselves.


\subsection{Speed of Learning}\label{speed of learning appendix}

Theorem \ref{speed result} relies on a characterization of the full spectrum of eigenvalues of the network. This provides an exact result in terms of average consensus time, $\tau(\epsilon)$. As part of the proof, we show that the second largest eigenvalue modulus (SLEM) is increasing in the strength of confirmation bias. Results in both the economics and mathematics literature provide approximations of both mixing times and convergence times (in terms of beliefs) in terms of the SLEM. Therefore the link between confirmation bias and the SLEM is of interest in its own right. Before stating the result, it is helpful to introduce a formal definition for the SLEM.

\begin{defn}[SLEM]
	Consider an $n$-node network $T$. Order its eigenvalues such that $1 = \lambda_1 \geq \lambda_2 \geq ... \geq \lambda_n$. The second largest eigenvalue modulus (SLEM) of $T$ is $\lambda^* = \max \{ | \lambda_2 | , | \lambda_n | \}$.
\end{defn}

\begin{prop}\label{SLEM result}
	When $q < \overline{q}(T,x_0)$ and $T_{ii} \geq \frac{1}{2}$ for all $i$, then for any reversible network $T$, the second largest eigenvalue modulus, $\lambda^*$, is weakly monotonically increasing in the amount of confirmation bias $q$.
\end{prop}

\begin{proof}
	This follows directly from the proof to Theorem \ref{speed result}.
\end{proof}

We now restate a pair of theorems from \cite{demarzo2003persuasion} and \cite{golub2012homophily} respectively. Both link the SLEM to convergence time. It is clear each of these results can be coupled with Proposition \ref{SLEM result} to link confirmation bias to convergence time. Proposition \ref{SLEM result} links confirmation bias to the SLEM, and then each of these results links the SLEM to convergence time.

\begin{thm}\citep[Theorem 5]{demarzo2003persuasion}
	Consider a matrix $T$ with eigenvalues $\lambda_1,...,\lambda_n$. Then $|| x_{it} - s x_{0}|| = \mathcal{O}\big((\lambda^{*})^{t-1}\big)$.
\end{thm}

The result in \cite{golub2012homophily} first requires the formal definition of the convergence time metric they use, called \emph{consensus time}. The consensus time is the minimum time $t$ at which all agents have beliefs within $\epsilon$ of each other in the \emph{worst-case} distribution of initial signals. 

\begin{defn}\citep[Definition 2]{golub2012homophily}\label{defn convergence speed}
	The consensus time of the network $T$ is:
	\begin{equation*}
		CT(\frac{\epsilon}{n}) = \sup\limits_{ x_{0} \in [0,1]^{n} } \min \{ t : || T^{t} x_{0} - T^{\infty} x_{0} ||^{2}_{s(i)} < \frac{\epsilon}{n} \} \\
	\end{equation*}
	where $||.||_{s(i)}$ is a weighted version of the $\boldsymbol\ell_{2}$ norm.
\end{defn}

\begin{thm}\citep[Theorem 1]{golub2012homophily}
	Consider a sequence of multi-type random networks with $\lambda^* \in (0,1)$. Then for any $\epsilon > 0$:
	\begin{equation*}
		CT(\frac{\epsilon}{n}) \approx \frac{\log(n)}{\log\left(\frac{1}{\lambda^*}\right)}
	\end{equation*}
\end{thm}

We have restated the results in a way that keeps notation consistent with that used throughout this paper. Notation in the original theorem statements differs.\footnote{In particular, \cite{demarzo2003persuasion} consider a setting where the network is allowed to change in a particular manner over time, whereas we do not. Therefore the notation used here is materially simplified compared to the original statement.}

\subsection{Influencers and Listeners}\label{influence appendix}

Remark \ref{multiple agent influence} shows that agents' influence can either rise or fall due to confirmation bias. Here, we characterize two types of agents: \emph{influencers} (whose influence rises) and \emph{listeners} (whose influence falls) due to confirmation bias. Definition \ref{influencer} identifies an \emph{influencer} as an agent who (1) cuts links with other agents, but (2) no other agent cuts a link with her, and (3) continues listening only to agents who satisfy (1) and (2).

\begin{defn} \label{influencer}
	An agent $i$ is an \emph{influencer} if $d_{i, \text{out}}(T^{*}) < d_{i, \text{out}}(T)$ and $d_{i, \text{in}}(T^{*}) = d_{i, \text{in}}(T)$ and, for every $j\in N_{out}^{*}(i)$, we have that $d_{j, \text{out}}(T^{*}) < d_{j, \text{out}}(T)$ and $d_{j, \text{in}}(T^{*}) =d_{j, \text{in}}(T)$.
\end{defn}

The opposite of an influencer is a \emph{listener} -- an agent who (1) does not cut any links with other agents, but (2) others cut links with her, and (3) continues to be listened to only by other agents who satisfy (1) and (2) after confirmation bias has changed the network.

\begin{defn} \label{listener}
	An agent $i$ is a \emph{listener} if $d_{i, \text{out}}(T^{*}) = d_{i, \text{out}}(T)$ and $d_{i, \text{in}}(T^{*}) < d_{i, \text{in}}(T)$ and, for every $j\in N_{in}^{*}(i)$, we have that $d_{j, \text{in}}(T^{*}) < d_{j, \text{in}}(T)$ and $d_{j, \text{out}}(T^{*}) = d_{j, \text{out}}(T)$.
\end{defn}

A final restriction needed to understand shifts of influence at the individual level is to focus on a society that is ``wise'', as defined by \citet{golub2010naive}. Informally, a network is wise when each agent has a negligible influence on the outcome the society converges to. The following definition states this formally.

\begin{defn} \label{wisdom}
	A large network $T$ is \emph{wise} if and only if $s_{i} \sim \mathcal{O}(\frac{1}{n})$ for all $i$, and therefore $s_i \approx 0$ as $n$ gets large.
\end{defn}

Armed with these definitions, we can characterize how confirmation bias changes the influence of influencers and listeners in a wise society.

\begin{prop}\label{influencer listener result}
	Assume $T$ and $T^*$ are wise networks. The influence of influencers rises due to confirmation bias and the influence of listeners declines due to confirmation bias.
\end{prop}
This result shows that confirmation bias increases the influence of a particular subgroup of agents -- the influencers -- who are embedded in a neighbourhood of other agents like them. Confirmation bias has the opposite effect on listeners. Given the restrictions imposed by definitions \ref{influencer} and \ref{listener}, a large network would typically have few (if any) influencers/listeners. An implication is that we are able to characterize changes in influence at the individual level only for a small subset of agents. Technically, the wisdom assumption allows us to ignore the effect on an agent's influence of changes in the neighbours' neighbours' influence, which helps with tractability.

\begin{proof}[\emph{\textbf{Proof of Proposition \ref{influencer listener result}.}}]
	We prove the result for an \emph{influencer}, $i$. Denote the set of influencers $\Theta$, and the set of agents that influencers listen to $\widehat{\Theta}$.
	
	Taking the definition of influence for agent $i$: $(1 - T_{ii}) s_i = \sum_{ j \in \mathbb{N}(i) } T_{ji} s_j$, and for agent $j$: $(1 - T_{jj}) s_j = \sum_{ k \in \mathbb{N}(j) \backslash \{ i \} } T_{kj} s_k + T_{ij} s_i$, and substituting the second into the first, yields
	\begin{align*}
		(1 - T_{ii}) s_i	= \sum_{ j \in \mathbb{N}(i) } T_{ji} \left( \left[ \sum_{ k \in \mathbb{N}(j) \backslash \{ i \} } T_{kj} s_k + T_{ij} s_i \right] (1 - T_{jj})^{-1} \right)
		\intertext{Rearranging yields }
		(1 - T_{ii} - [ \sum_{ j \in \mathbb{N}(i) } T_{ji} T_{ij} (1 - T_{jj})^{-1} ] ) s_i = 		\sum_{ j \in \mathbb{N}(i) } \sum_{ k \in \mathbb{N}(j) \backslash \{ i \} } T_{ji} T_{kj} s_k (1 - T_{jj})^{-1}
	\end{align*}
	And by identical logic we can get an equivalent equation for $s_{i}^{*}$.
	By definition: (a) $i \in \Theta, \ T_{ji} = T_{ji}^{*}$ for all $j \neq i \in \mathbb{N}(i)$, and (b) $j \in \widehat{\Theta}, \ T_{kj} = T_{kj}^{*}$ for all $k \neq j,i \in \mathbb{N}(j) \backslash i$. Now let (c) $T_{ii}^{*} = T_{ii} + \Delta_{i}, \ \Delta_{i} > 0$, and (d) $T_{jj}^{*} = T_{jj} + \Delta_{j}, \ \Delta_{j} > 0$.
	Substituting (a)-(d) into the equivalent equation for $s_{i}^{*}$, and then subtracting the equation for $s_i$ yields
	\begin{align*}
		\left(1 - T_{ii} - \Delta_{i} - \big[ \sum_{ j \in \mathbb{N}(i) }  T_{ji} T_{ij}^{*}  (1 - T_{jj} - \Delta_{j})^{-1}  \big] 	\right) 	s_{i}^{*}  		
		-
		\left( 1 - T_{ii} - \big[ \sum_{ j \in \mathbb{N}(i) } T_{ji} T_{ij}  (1 - T_{jj})^{-1} \big] \right) s_{i} 		
		=								\\		
		\left( \sum_{ j \in \mathbb{N}(i) } \sum_{ k \in \mathbb{N}(j) \backslash \{ i \} } T_{ji} T_{kj} s_{k}^{*} (1 - T_{jj} - \Delta_{j})^{-1} \right)		
		-
		\left( \sum_{ j \in \mathbb{N}(i) } \sum_{ k \in \mathbb{N}(j) \backslash \{ i \} } T_{ji} T_{kj} s_k (1 - T_{jj})^{-1} \right)		
	\end{align*}
	Rearranging the Right Hand Side (RHS) yields
	\begin{align*}
		\text{RHS} = \sum_{ j \in \mathbb{N}(i) } \sum_{ k \in \mathbb{N}(j) \backslash \{ i \} }
		\frac{T_{ji} T_{kj} (s_{k}^{*} - s_{k})}{1 - T_{jj} - \Delta_{j} }
		+
		\sum_{ j \in \mathbb{N}(i) } \sum_{  k \in \mathbb{N}(j) \backslash \{ i \}  }
		\frac{\Delta_{j} T_{ji} T_{kj} s_{k} }{(1 - T_{jj})(1 - T_{jj} - \Delta_{j})}
	\end{align*}
	We assume that third-round effects are small: $T_{ji} T_{kj} (s_{k}^{*} - s_{k} ) \approx 0$. This yields
	\begin{align*}
		\text{RHS} = \sum_{ j \in \mathbb{N}(i) } \sum_{  k \in \mathbb{N}(j) \backslash \{ i \}  }
		\frac{\Delta_{j} T_{ji} T_{kj} s_{k} }{(1 - T_{jj})(1 - T_{jj} - \Delta_{j})}
	\end{align*}
	By the definition of an \emph{influencer}, if $i$ cuts a link with $j$ ($T_{ij}^{*} < T_{ij}$) then $j$ cannot listen to $i$ ($T_{ji} = 0$). Otherwise, $j$ would cut her link with $i$ -- due to the homogeneity of confirmation bias -- which is not permitted by the definition of an influencer. Therefore, $T_{ij} T_{ji} = T_{ij}^{*} T_{ji}$.
	
	We know that in any instance where $T_{ij} \neq T_{ij}^{*} \ , \ T_{ji} = 0$, and so that component of the summation is zero. Therefore, it is appropriate to equate $T_{ij}^{*}$ to $T_{ij}$ in this summation, as any instances where they differed would be irrelevant (the summation component would be zero). Using this result, we can rearrange the Left Hand Side (LHS) to find:
	\begin{align*}
		\text{LHS} =	 \left(1 - T_{ii} - \sum_{ j \in \mathbb{N}(i) }
		\frac{T_{ji} T_{ij}}{1 - T_{jj} - \Delta_{j}} \right)	(s_{i}^{*} - s_{i})
		- \Delta_{i} s_{i}^{*}
		- \sum_{ j \in \mathbb{N}(i) } s_i \frac{  (T_{ji} T_{ij}) \Delta_{j} }{(1 - T_{jj}) (1 - T_{jj} - \Delta_{j})  }
	\end{align*}
	
	Aside: show that for $i \in \Theta$: $ 1 - T_{ii} - \sum_{ j \in \mathbb{N}(i) }  \frac{  T_{ji} T_{ij}  }{  1 - T_{jj} - \Delta_{j}  }  > 0 $. First, note that $T_{jj} + \Delta_{j} + \sum_{k} T_{jk}^{*} = 1$ by definition, and that $T_{ji}^{*} = T_{ji}$ for all $j \neq i \in \mathbb{N}(i)$. This implies that $ 1- T_{jj} - \Delta_{j} - \sum_{k \neq i  \in \mathbb{N}(j)} T_{jk}^{*} = T_{ji} $. Since $\sum_{k \neq i  \in \mathbb{N}(j)} T_{jk}^{*} \geq 0$, then $ 1- T_{jj} - \Delta_{j} \geq T_{ji}$. This rearranges to $\frac{ T_{ji} T_{ij} }{1 - T_{jj} - \Delta_{j} } \leq \ T_{ij}$.
	
	If $\frac{ T_{ji} T_{ij} }{1 - T_{jj} - \Delta_{j} } = \ T_{ij}$ for all $j$, then $1 - T_{ii} - \sum_{ j \in \mathbb{N}(i) } \frac{T_{ji} T_{ij}}{1 - T_{jj} - \Delta_{j}} = 0$. However, there must exist at least one $j \in \mathbb{N}(i)$ where $T_{ji} = 0$, as the definition of an influencer requires that $i$ cuts links with some $j$, but no $j$ cuts links with $i$. This in turn requires that there is some $j$ who does not listen to $i$, but is listened to by $i$. Therefore:
	\begin{align*}
		1 - T_{ii} - \sum_{j\in \mathbb{N}(i)} \frac{T_{ji} T_{ij}}{1 - T_{jj} - \Delta_{j}} > 0
	\end{align*}
	
	Now equate LHS and RHS (from above) and rearrange:
	\begin{align*}
    \hspace*{-20mm}
		\left(1 - T_{ii} - \sum_{ j \in \mathbb{N}(i) }
		\frac{T_{ji} T_{ij}}{1 - T_{jj} - \Delta_{j}} \right)	(s_{i}^{*} - s_{i})
		=
		\Delta_{i} s_{i}^{*}
		+
		\sum_{j\in \mathbb{N}(i)}\frac{T_{ji} \Delta_{j}}{(1 - T_{jj}) (1 - T_{jj} - \Delta_{j})}
		\left( s_i T_{ij} + \sum_{k \in \mathbb{N}(j) \backslash \{ i \} } T_{kj} s_{k} \right)
	\end{align*}
	All Right Hand Side terms are weakly greater than zero, and we have established that $1 - T_{ii} - \sum_{ j \in \mathbb{N}(i) } \frac{T_{ji} T_{ij}}{1 - T_{jj} - \Delta_{j}} > 0$. Therefore:
	\begin{align*}
		(s_{i}^{*} - s_i) \geq 0
	\end{align*}
	
	The proof for a \emph{listener} follows the same logic as that for an \emph{influencer} and does not provide any additional insight. It is available from the authors upon request.
\end{proof}

\subsection{Media}\label{media appendix}

In this section we find the exact ideology of the fringe media organization for each value of $M \leq 5$, and for all values of $q \in [0,1]$.

\textbf{M=1:} $\rho_{fr}(q,M=1) \in [\min(q, 1-q) , \max(q, 1-q)]$. Where there is a monopoly media organization, it chooses any ideology that will capture the largest possible audience (when $q \leq 0.5$ this encompasses all agents).

\textbf{M=2:} $\rho_{fr}(q\leq 0.5, M=2) = 0.5$ and $\rho_{fr}(q > 0.5, M=2) = 1-q$. When $q \leq 0.5$, all agents are still willing to listen to an ideology of $\mu =0.5$, and so the duopoly equilibrium is unaffected by confirmation bias in this range (see \citet{eaton1975} for a proof of the equilibrium). When $q > 0.5$ then the fringe media organization is willing to choose a sufficiently extreme ideology to attract the most extreme agent as a listener, \emph{but no more}. Clearly a more extreme ideology would not be optimal, as the media organization could choose a more centrist ideology, and not lose any listeners on the fringe, but would gain some listeners to the centre.

\textbf{M=3:} When $q \leq 0.75$ and there are three media organizations, an equilibrium does not exist. Therefore $\rho_{fr}(q \leq 0.75, M=3)$ does not exist. $\rho_{fr}(q > 0.75, M=3) = 1-q$. With 3 media organizations and low confirmation bias, there is no equilibrium -- peripheral organizations want to move inwards (towards the centre), which increases their audience, but this causes the organization in the middle to want to change its ideology discontinuously to become a peripheral organization. This cycle never ends -- preventing equilibrium existence. When $q>0.75$ however, the peripheral organizations no longer want to move inwards past $\rho = 1-q$ (or $\rho = q$), as they would lose very extreme agents from their audience by doing so. Given this, the organization in the middle no longer wants to switch its ideology to become a peripheral organization.

\textbf{M=4 and M=5:} $\rho_{fr}(q \leq 0.75, M=4) = 0.25$ and $\rho_{fr}(q > 0.75, M=4) = 1-q$.  $\rho_{fr}(q \leq \frac{5}{6}, M=5) = \frac{1}{6}$ and $\rho_{fr}(q > \frac{5}{6}, M=5) = 1-q$. The logic here is identical to the two-organization case above.

We can see from this that Propositions \ref{media monotone m} and \ref{media monotone q} hold for two or more media organizations (aside from the existence issue when $M=3$), but that neither holds for $M=1$. At least some competition is required to sustain our results.

\subsection{Simulations}\label{simulations appendix} 
In Section \ref{simulations results} we conducted a number of simulation exercises. For all of those exercises, we constructed the networks using the \cite{jackson2007meeting} algorithm (which mixes both random meetings and preferential attachment). Here, we report analogous results constructed using Erd\H{o}s-R\'enyi networks and scale-free networks. Throughout this appendix, the left panels show results for Erd\H{o}s-R\'enyi networks, and the right panels show results for scale-free networks. We also follow the order of Section \ref{simulations results}. Each run again consists of 100 networks of 200 agents and 100 draws of initial beliefs, and the same thresholds ($\epsilon = 10^{-4}$) are used throughout.

\paragraph{Consensus Time.} It is clear from Figure \ref{fig:robustness_consensus_time} that it takes longer to reach a consensus with confirmation bias than without it, whether consensus is measured by the average consensus time (Definition \ref{def:time}, top row) or by the belief consensus time (bottom row). This does not depend on which type of network we use (Erd\H{o}s-R\'enyi or scale-free). Figure \ref{fig:robustness_effect_of_phi} then repeats the exercise of Figure \ref{fig:effect_of_phi}: for both network types, and for both measures, consensus time increases in the strength of confirmation bias, $q$, and in the fraction $\phi$ of the weight from cut links that is re-routed to oneself. Interestingly, the effect of an increase in the strength of confirmation bias does not appear to depend on which type of network we use (Erd\H{o}s-R\'enyi or scale-free).

\begin{figure}[p]
  \centering
  \begin{subfigure}{0.48\textwidth}
    \centering
    \includegraphics[width=\linewidth,height=0.30\textheight,keepaspectratio]{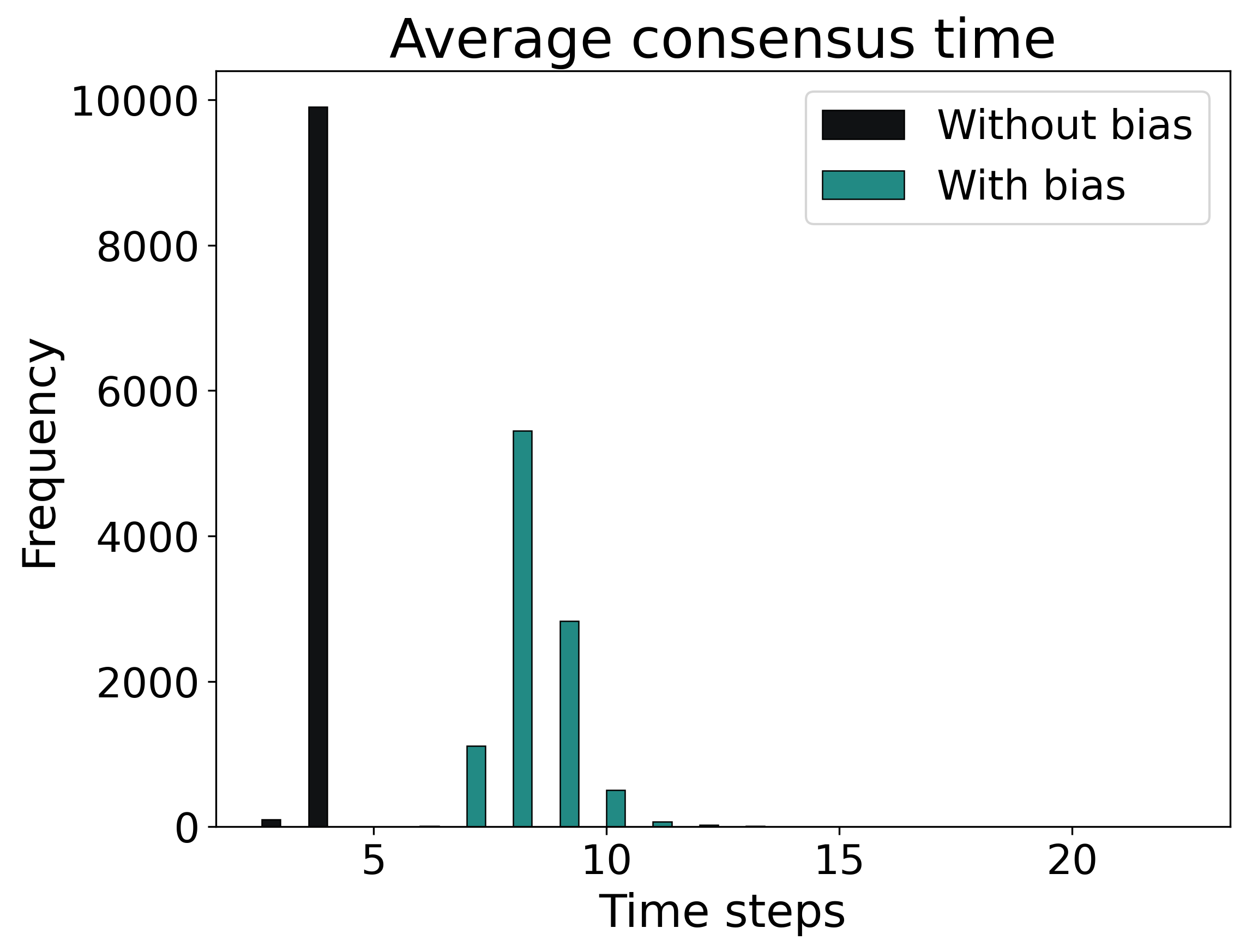}
    \caption{ }
  \end{subfigure}\hfill
  \begin{subfigure}{0.48\textwidth}
    \centering
    \includegraphics[width=\linewidth,height=0.30\textheight,keepaspectratio]{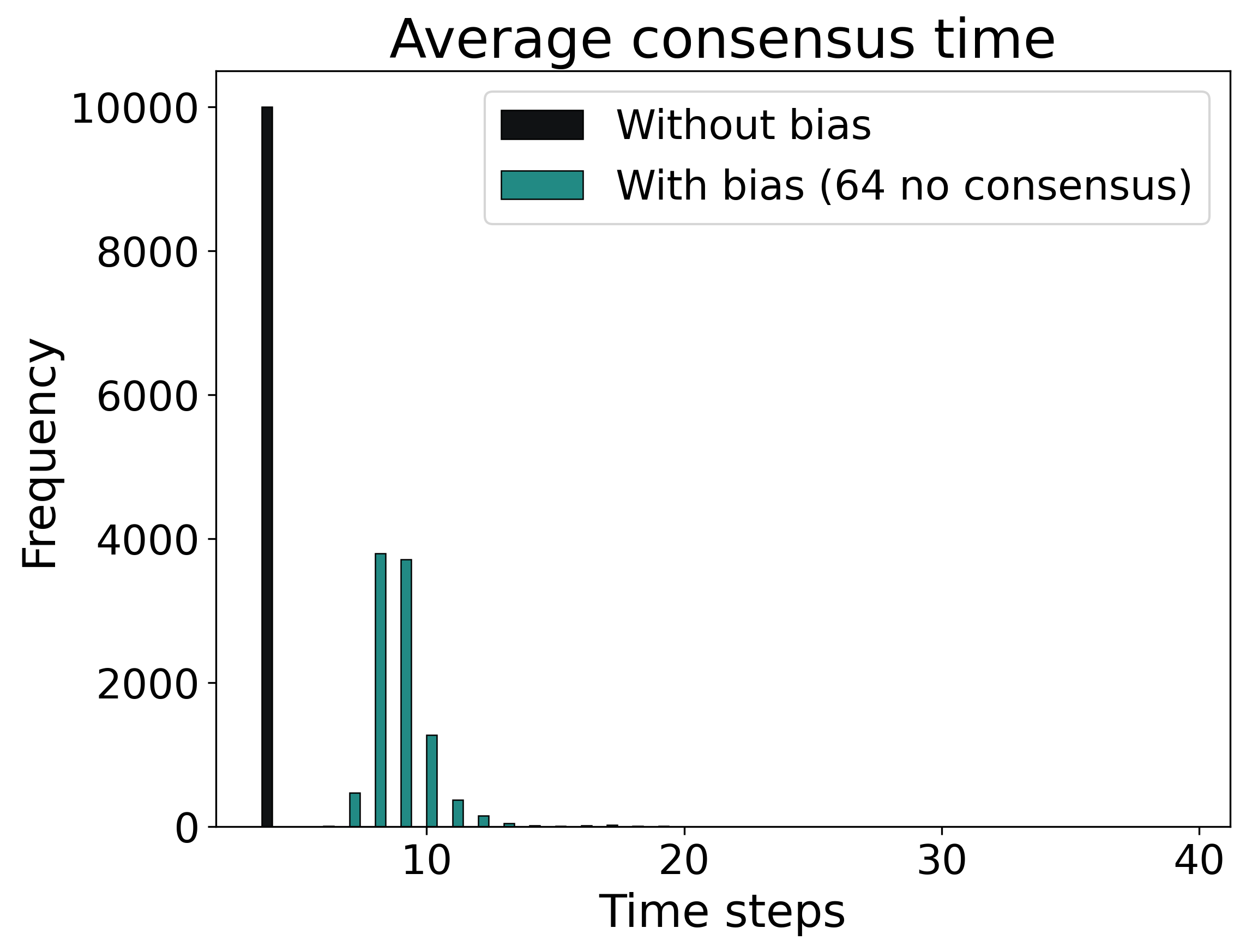}
    \caption{ }
  \end{subfigure}

  \begin{subfigure}{0.48\textwidth}
    \centering
    \includegraphics[width=\linewidth,height=0.30\textheight,keepaspectratio]{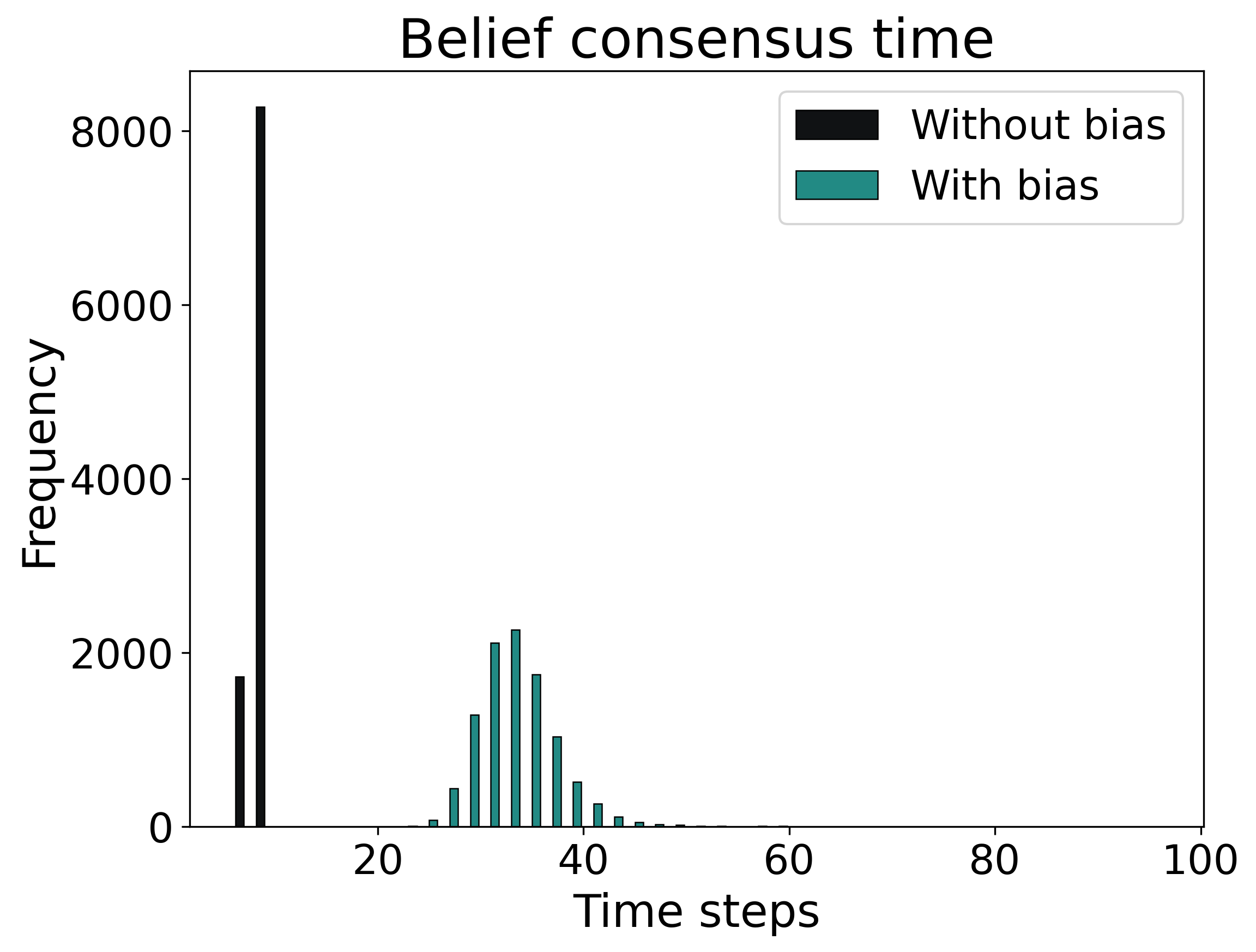}
    \caption{ }
  \end{subfigure}\hfill
  \begin{subfigure}{0.48\textwidth}
    \centering
    \includegraphics[width=\linewidth,height=0.30\textheight,keepaspectratio]{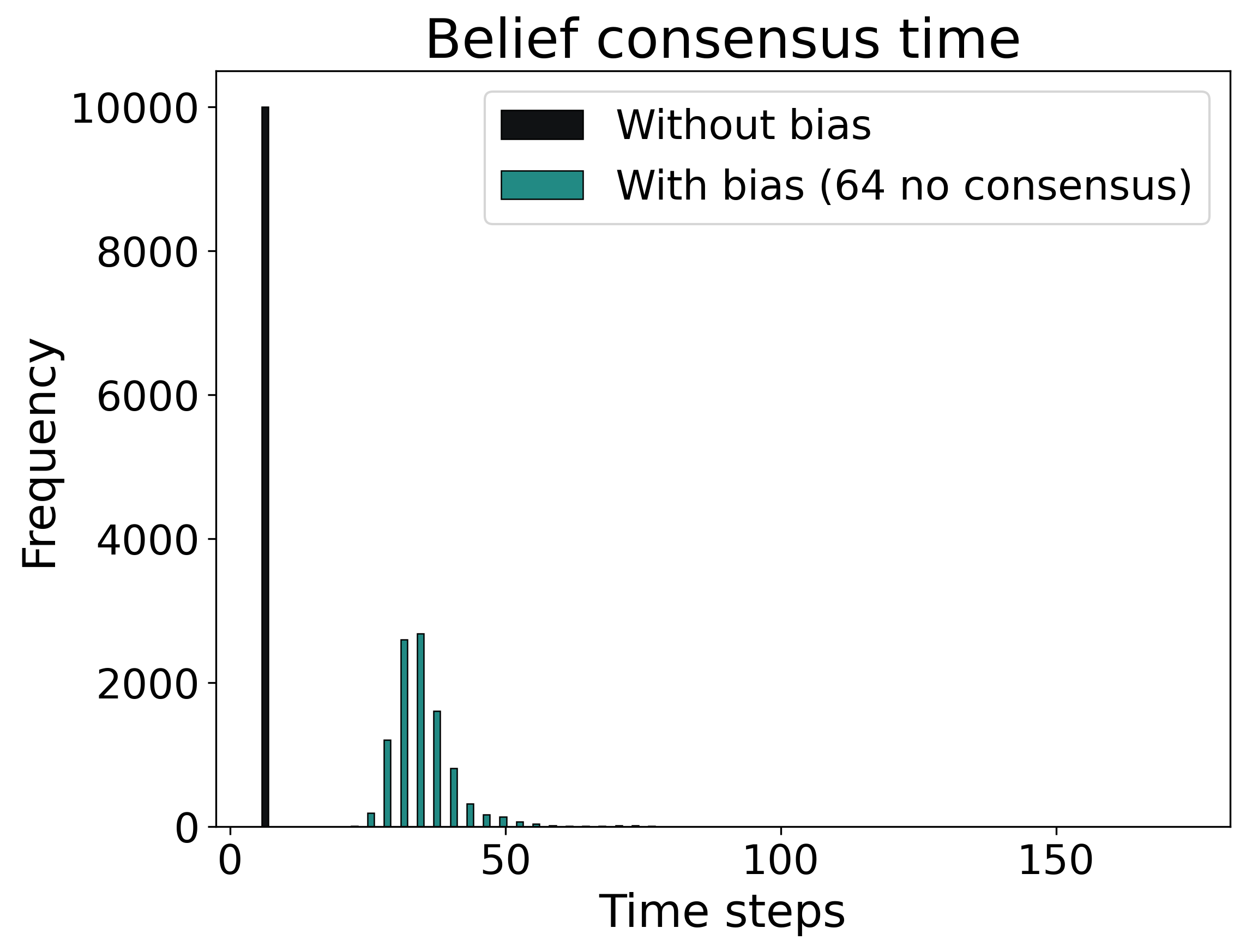}
    \caption{ }
  \end{subfigure}

  \caption{Average consensus time (top row) and belief consensus time (bottom row) without (black) and with (teal) confirmation bias. Uniform distribution of initial beliefs, $q=0.5$. Left panels: Erd\H{o}s-R\'enyi networks. Right panels: scale-free networks. Compare with Figure \ref{fig:consensus_time_sims}.}\label{fig:robustness_consensus_time}
\end{figure}

\begin{figure}[p]
  \centering
  \begin{subfigure}{0.48\textwidth}
    \centering
    \includegraphics[width=\linewidth,height=0.30\textheight,keepaspectratio]{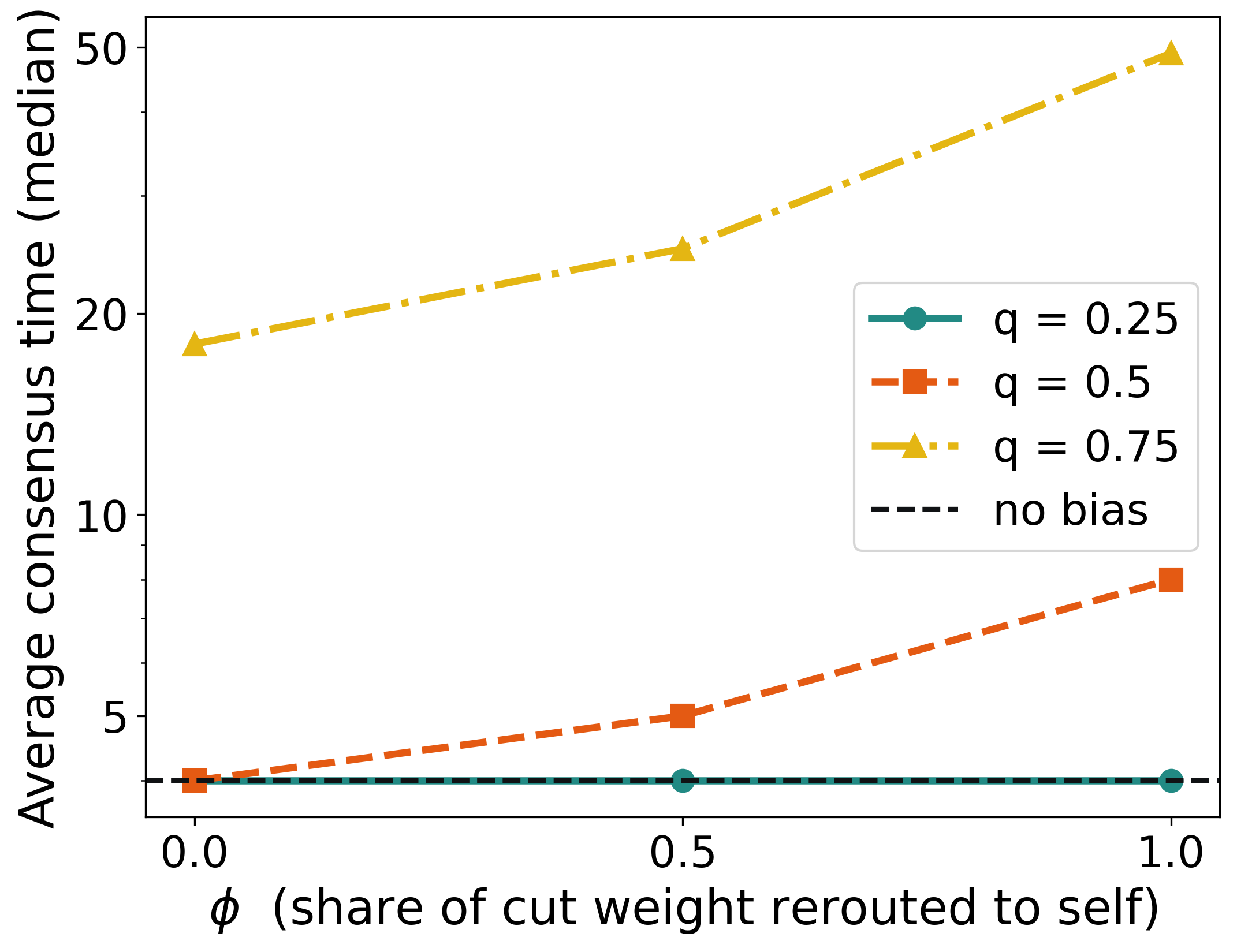}
    \caption{ }
  \end{subfigure}\hfill
  \begin{subfigure}{0.48\textwidth}
    \centering
    \includegraphics[width=\linewidth,height=0.30\textheight,keepaspectratio]{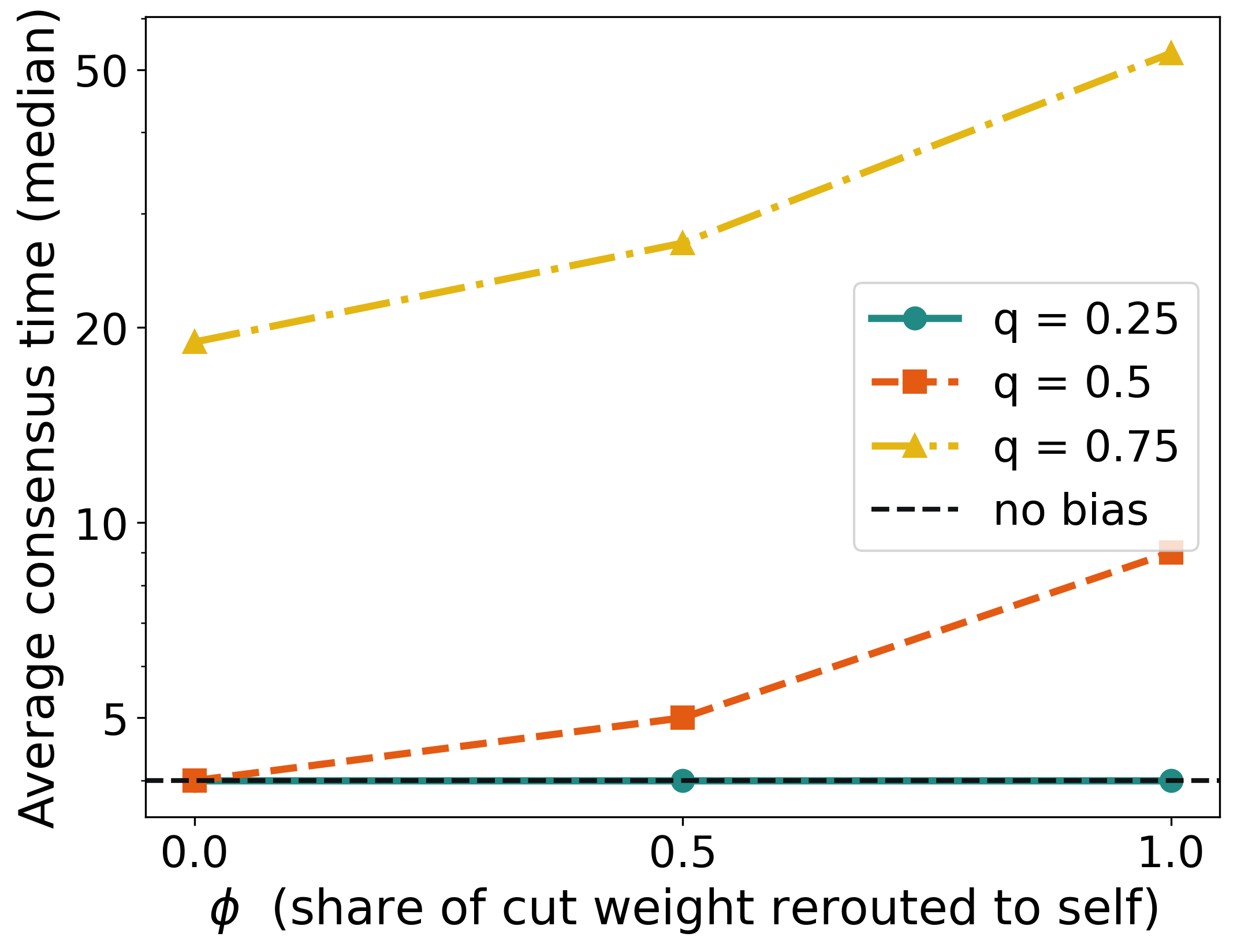}
    \caption{ }
  \end{subfigure}

  \begin{subfigure}{0.48\textwidth}
    \centering
    \includegraphics[width=\linewidth,height=0.30\textheight,keepaspectratio]{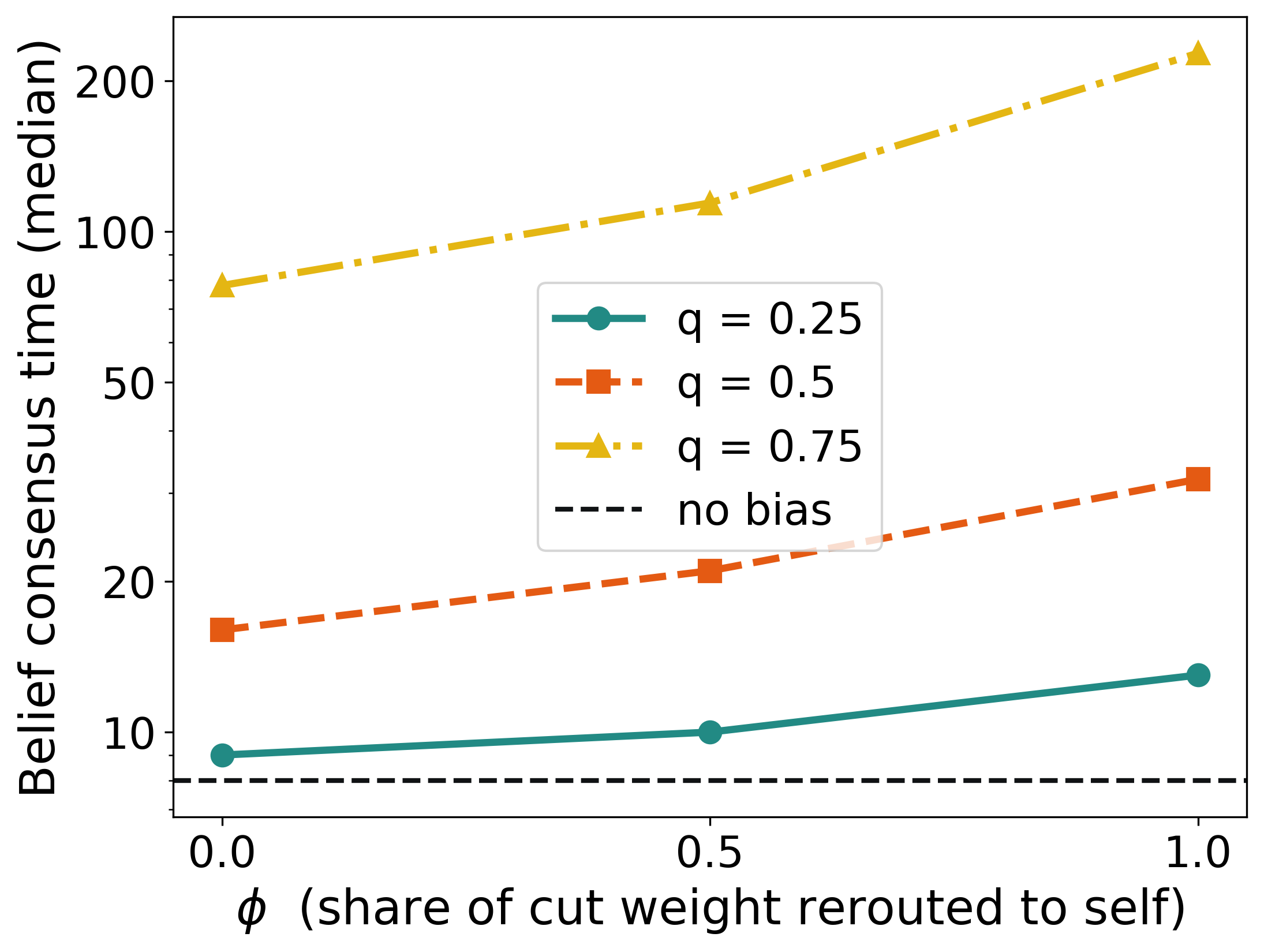}
    \caption{ }
  \end{subfigure}\hfill
  \begin{subfigure}{0.48\textwidth}
    \centering
    \includegraphics[width=\linewidth,height=0.30\textheight,keepaspectratio]{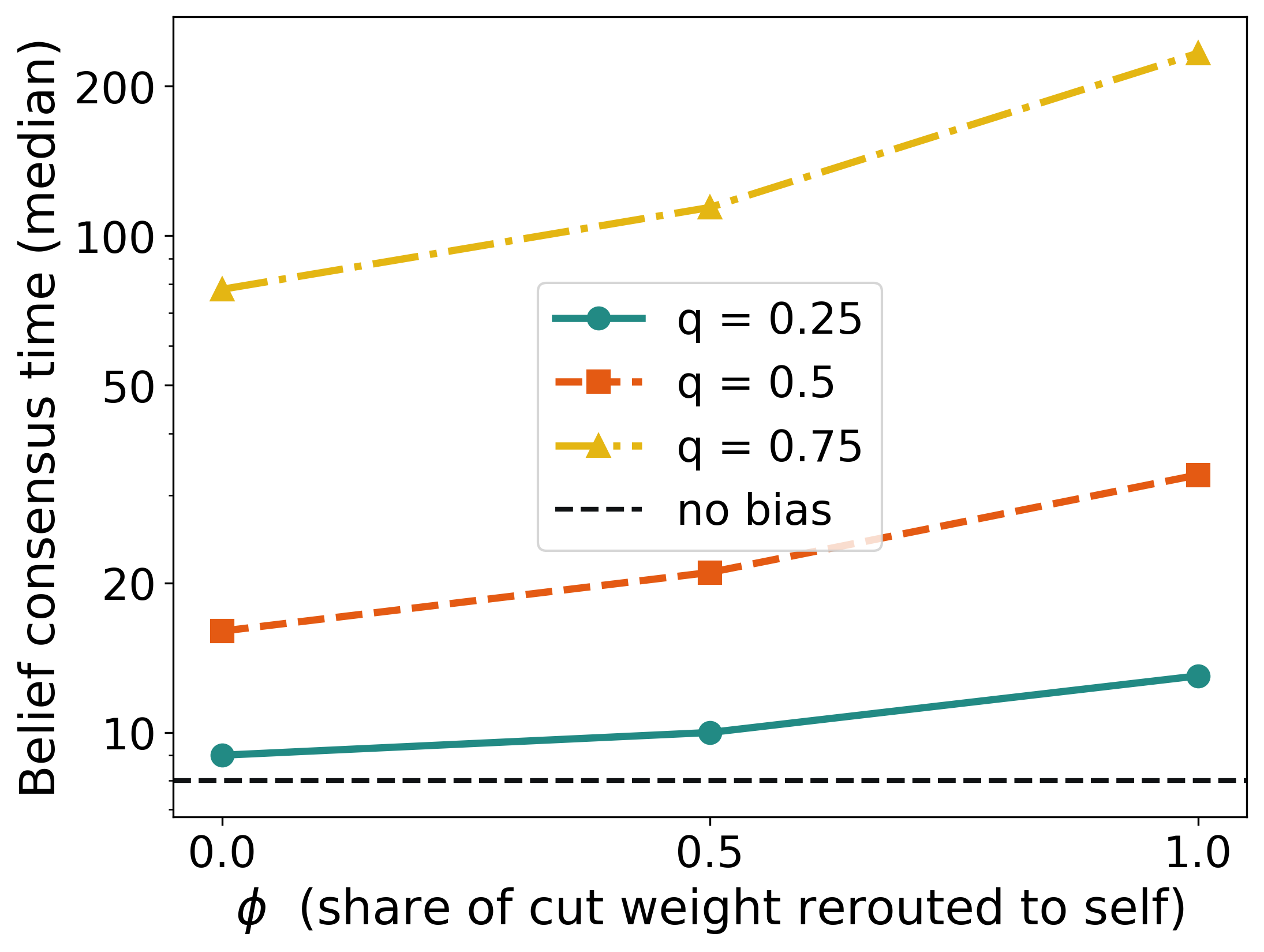}
    \caption{ }
  \end{subfigure}

  \caption{Average consensus time (top row) and belief consensus time (bottom row) for $q \in \{0.25, 0.5, 0.75\}$ and $\phi \in \{0, 0.5, 1\}$, with the no-bias level shown as a dashed line; median across the 10,000 instances, on a logarithmic scale. Uniform distribution of initial beliefs. Left panels: Erd\H{o}s-R\'enyi networks. Right panels: scale-free networks. Compare with Figure \ref{fig:effect_of_phi}.}\label{fig:robustness_effect_of_phi}
\end{figure}

\paragraph{Shock elections.} When using Erd\H{o}s-R\'enyi and scale-free networks, we find no shock elections at all without confirmation bias. And with confirmation bias, shock elections are no less common than when networks are constructed using the Jackson-Rogers algorithm: a shock election occurs in roughly 60\% of instances for both network types (62\% for Erd\H{o}s-R\'enyi and 60\% for scale-free networks, against 60\% for Jackson-Rogers). Figure \ref{fig:robustness_shock_elections} shows the incidence of shock elections without (top panels) and with (bottom panels) confirmation bias, and is the analogue of Figure \ref{fig:shock_election_sims}. As in Section \ref{sec:simulation_tests}, we use the point mass distribution of initial beliefs and set $q=0.78$.

\begin{figure}[H]
\centering
\begin{subfigure}[t]{0.49\linewidth}
\centering
\includegraphics[width=\linewidth]{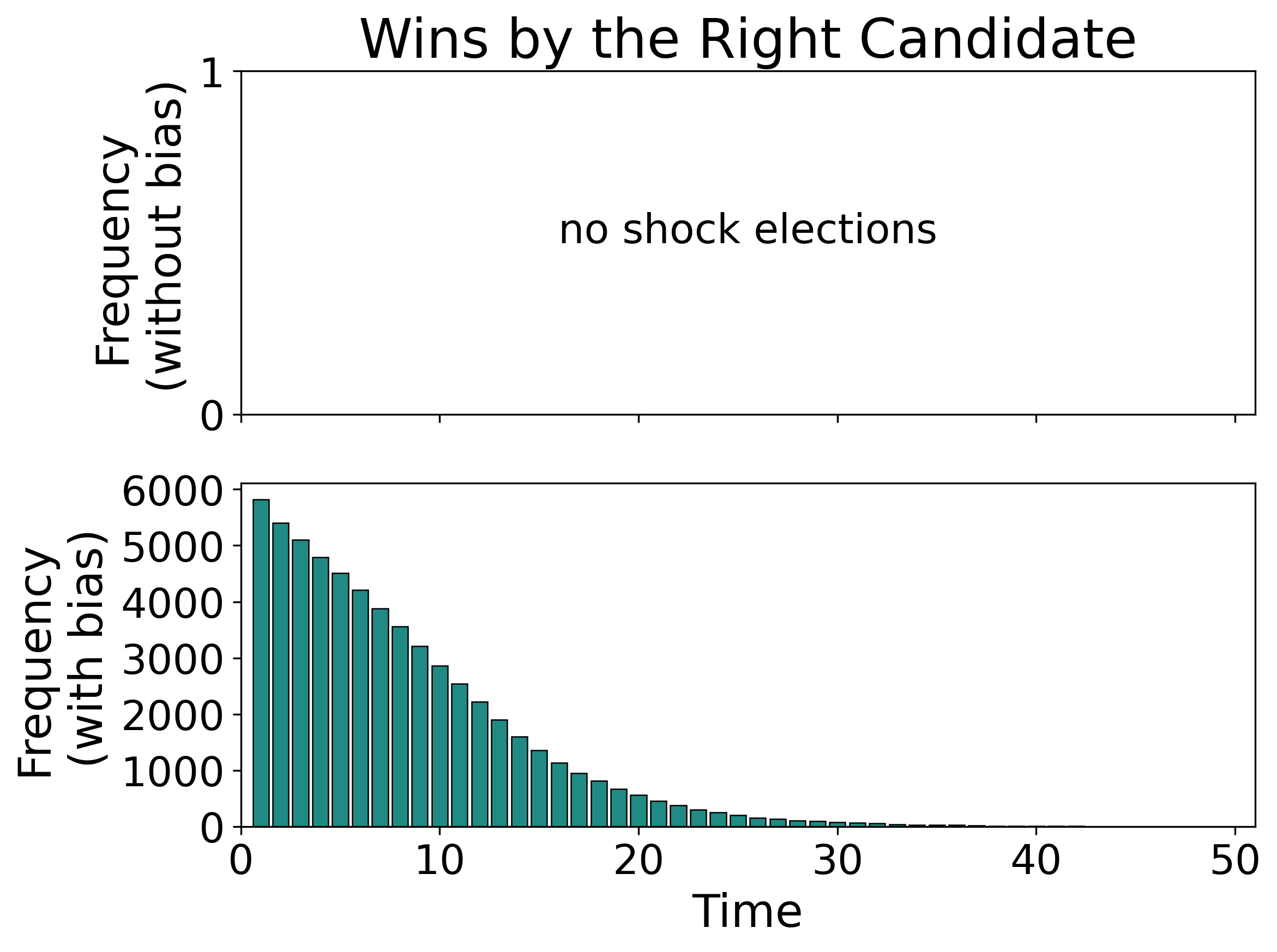}
\caption{ }
\end{subfigure}
\hfill
\begin{subfigure}[t]{0.49\linewidth}
\centering
\includegraphics[width=\linewidth]{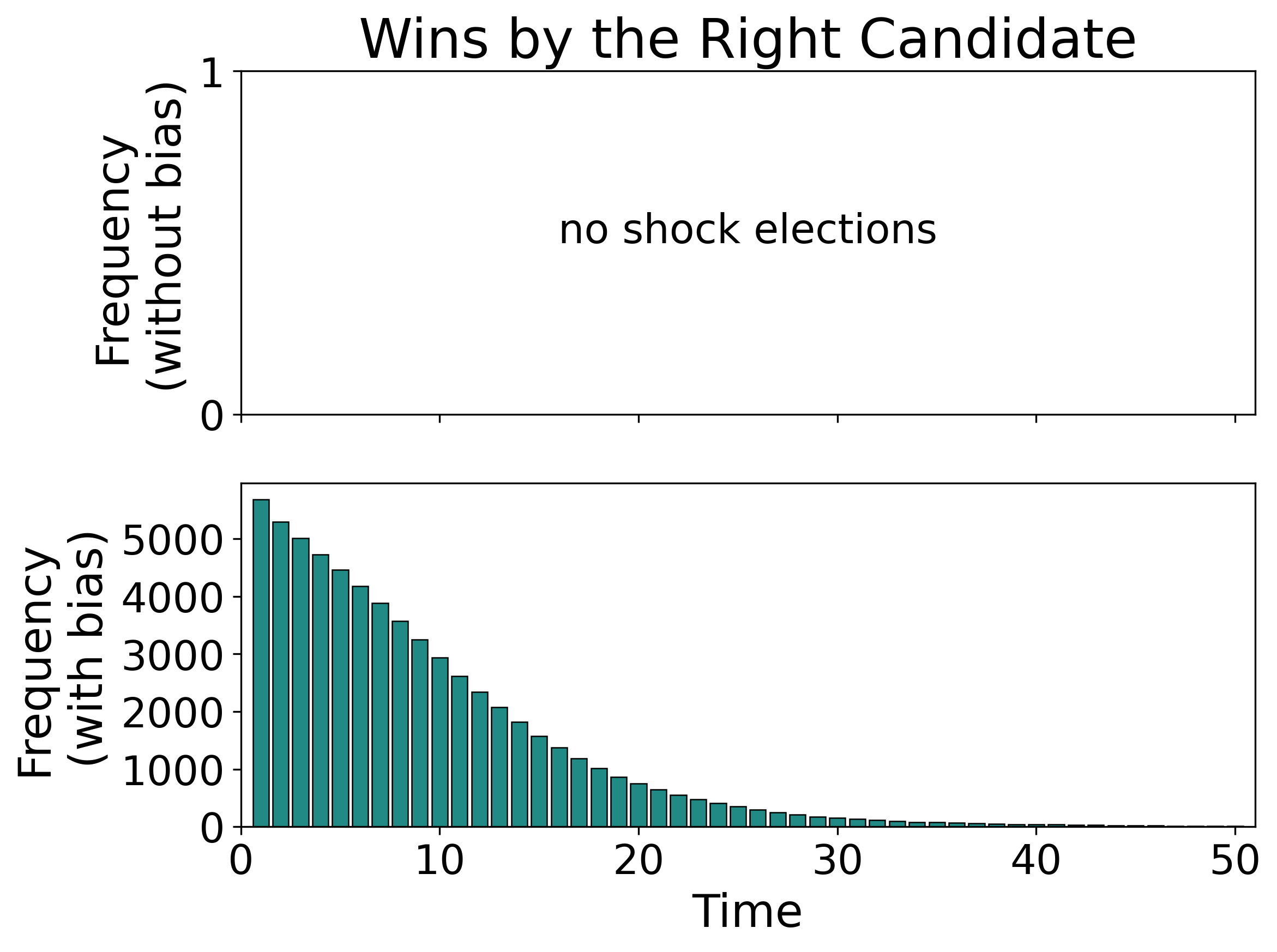}
\caption{ }
\end{subfigure}
\caption{Wins by Right Candidate without (black, top panels) and with (teal, bottom panels) confirmation bias. Each bar counts the number of instances in which the Right candidate would win an election held at that date. Point mass distribution of initial beliefs, $q=0.78$. Left panel: Erd\H{o}s-R\'enyi networks. Right panel: scale-free networks.}
    \label{fig:robustness_shock_elections}
\end{figure}

\paragraph{Evolution of beliefs.} Our finding from Section \ref{sec:further results} that average variance is higher with confirmation bias than without in every single period extends unchanged to Erd\H{o}s-R\'enyi and scale-free networks: as in the Jackson-Rogers case, there is not a single instance, in any period, in which the variance of beliefs without confirmation bias exceeds the variance with it. Figure \ref{fig:robustness_polarisation} is very similar to Figure \ref{fig:polarisation}. The only notable difference is that in panel (b) the average variance with confirmation bias levels off at a small positive value rather than falling all the way to zero. This is because in a small number of instances (0.1\%) confirmation bias disconnects the scale-free network, so that agents never reach a consensus; as the last value of each instance is carried forward, these instances leave the average at a small positive level.

\begin{figure}[h!]
  \centering
  \begin{subfigure}[b]{0.45\textwidth}
    \centering
    \includegraphics[width=\linewidth]{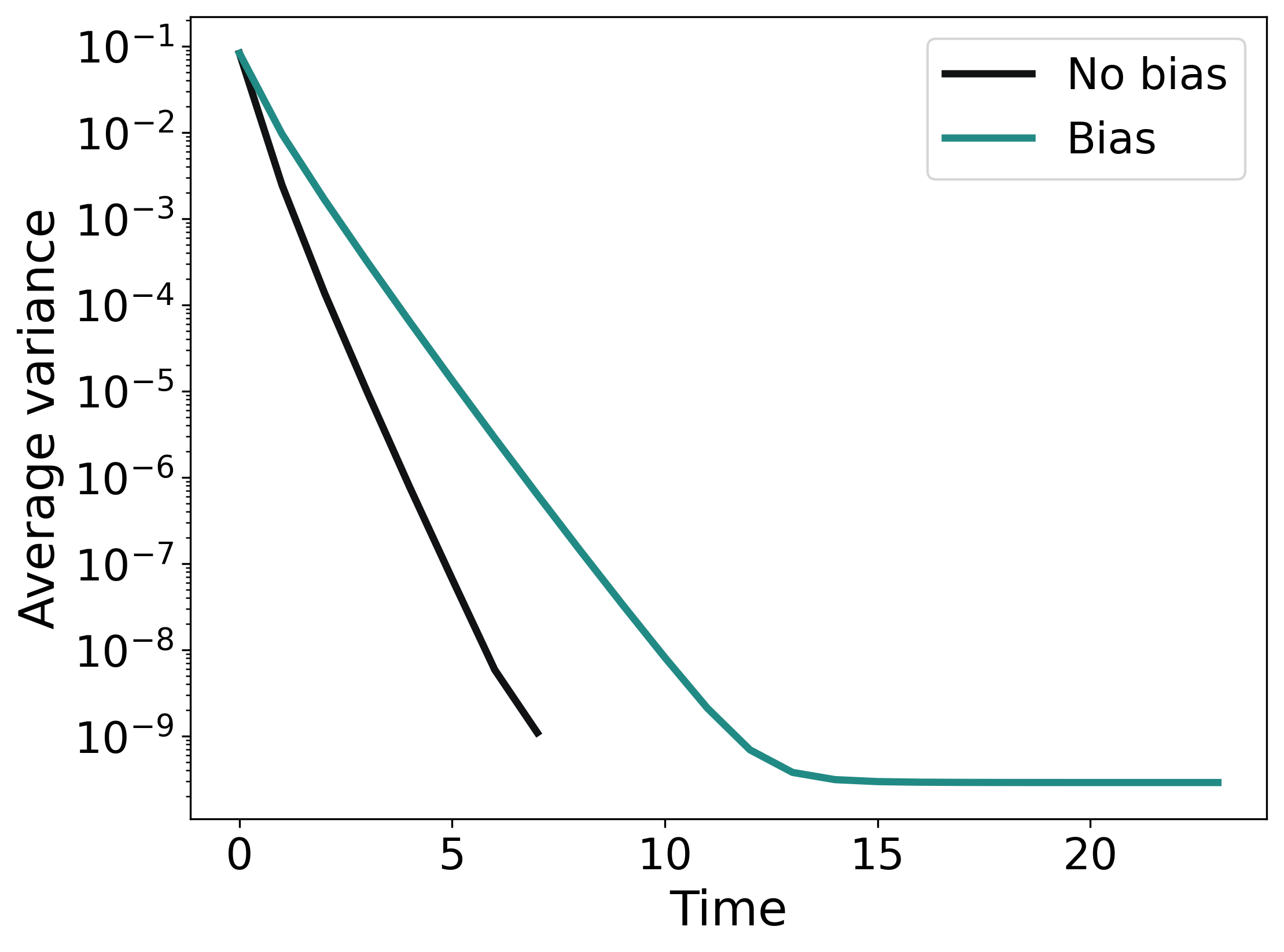}
    \caption{ }
    \label{fig:a}
  \end{subfigure}
  \hfill
  \begin{subfigure}[b]{0.45\textwidth}
    \centering
    \includegraphics[width=\linewidth]{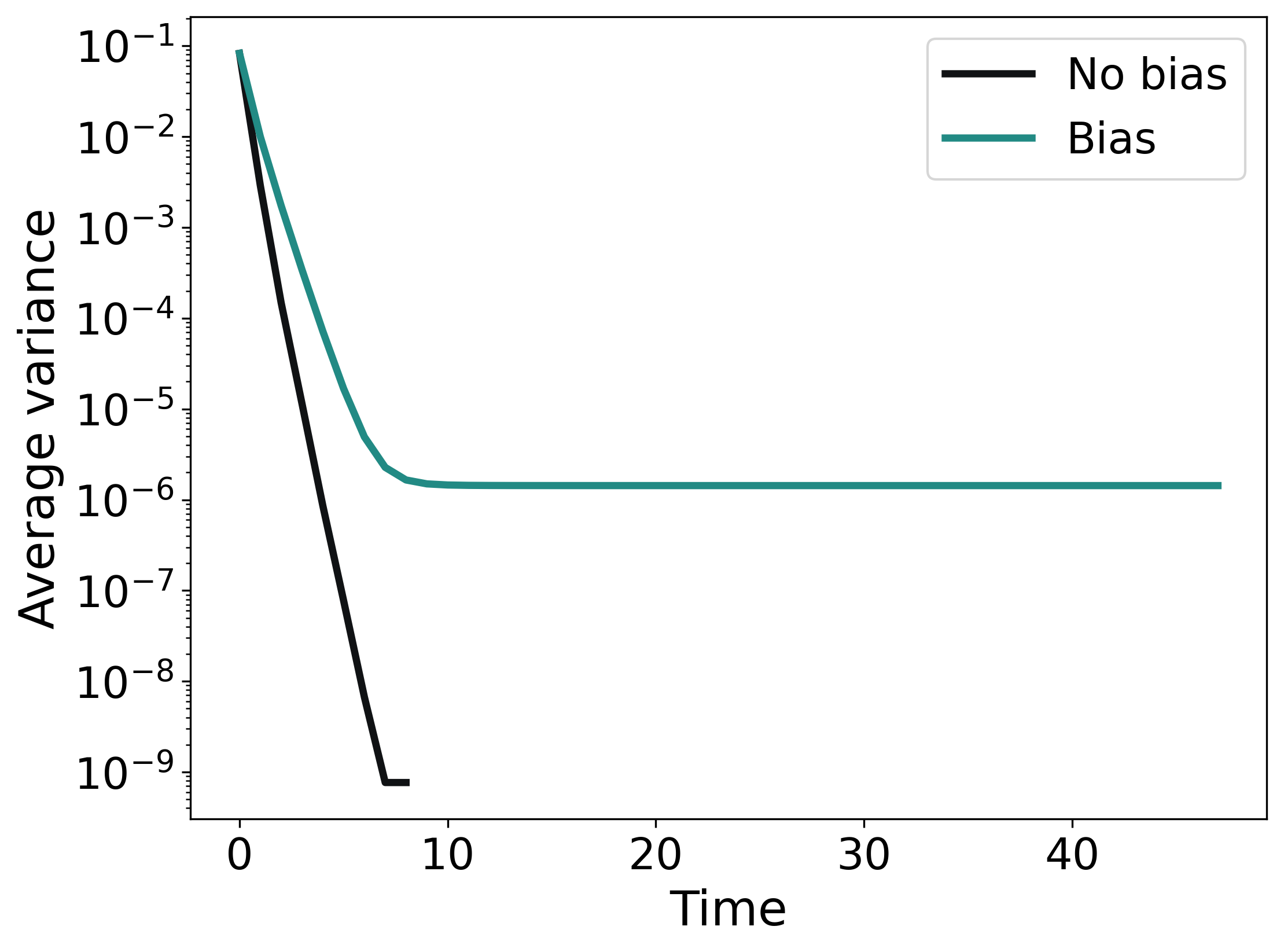}
    \caption{ }
    \label{fig:b}
  \end{subfigure}

  \vspace{0.5cm}

  \begin{subfigure}[b]{0.45\textwidth}
    \centering
    \includegraphics[width=\linewidth]{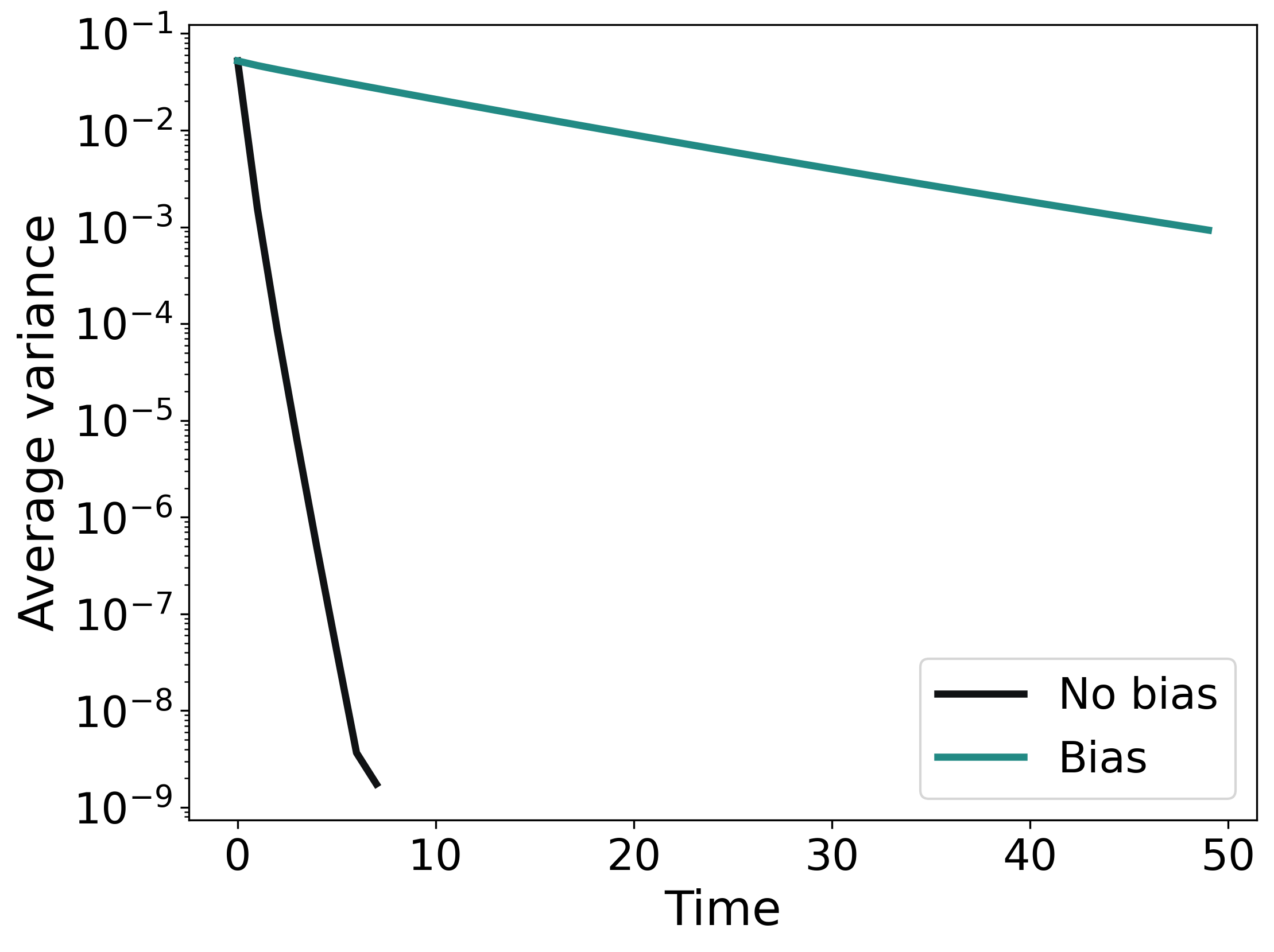}
    \caption{ }
    \label{fig:c}
  \end{subfigure}
  \hfill
  \begin{subfigure}[b]{0.45\textwidth}
    \centering
    \includegraphics[width=\linewidth]{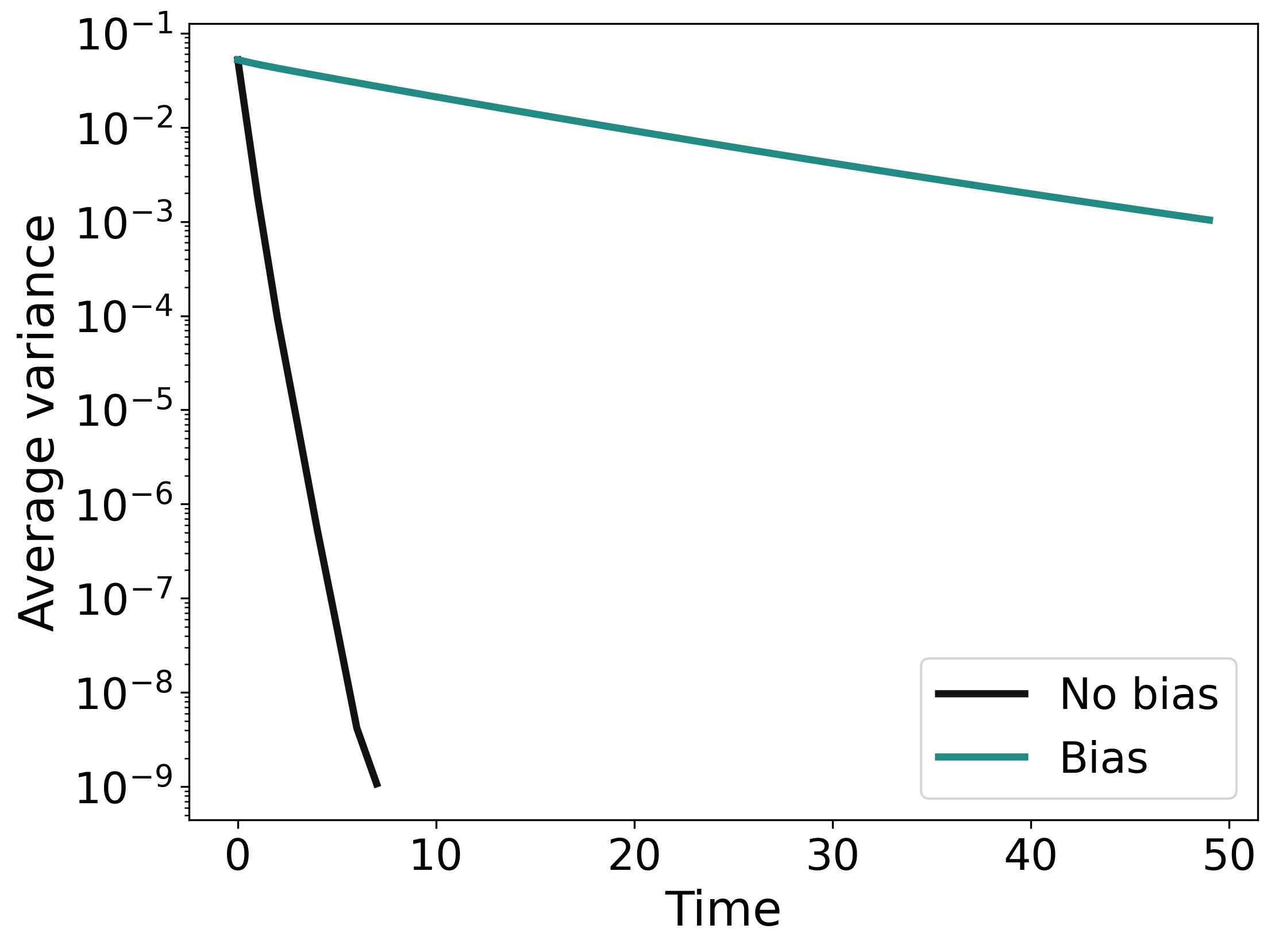}
    \caption{ }
    \label{fig:d}
  \end{subfigure}

  \caption{Average variance of beliefs over time without (black) and with (teal) confirmation bias, for the first 50 periods of learning; once an instance has reached a consensus its last value is carried forward. Left panels: Erd\H{o}s-R\'enyi networks. Right panels: scale-free networks. Top row: Uniform distribution of initial beliefs, $q=0.25$. Bottom row: Point mass distribution of initial beliefs, $q=0.78$. }\label{fig:robustness_polarisation}
\end{figure}

\newpage
\subsection{Network construction processes}
\paragraph{\citet{jackson2007meeting} algorithm.} We start at step $k=0$ from an initial cluster $M$ of $m=40$ nodes such that every node is connected to every other, with the direction of the link randomly determined with $50\%$ probability, i.e., if $T_{ij}=1$ then $T_{ji}=0$ and if $T_{ij}=0$ then $T_{ji}=1$ for all $i,j \in M$. In step $k=1$, a new node is introduced and randomly ``meets'' $m_r=20$ other nodes. In each random meeting there is a $p_r=0.75$ probability of forming a link, with the direction of the link randomly determined with $50\%$ probability. After these connections through random meetings have been formed, the new node meets $m_n=20$ neighbours of her new connections. There is a $p_n=0.75$ probability that she forms a link with each of these connections, and the direction of the link is randomly determined with $50\%$ probability. In step $k=2$ a new node is introduced and the process repeats itself. The network formation process stops at step $K=160$ when we obtain a directed, weighted network formed by 200 agents. In case the resulting network is not strongly connected, we disregard it and restart the process to construct a new one. An agent listens to all of their neighbours equally, but the number of neighbours differs across agents. Therefore, link weights depend on the agent who is listening, and so the network is \emph{weighted}.

\paragraph{Erd\H{o}s-R\'enyi random network.} We start with 200 nodes and for each pair of nodes $(i,j)$, form an undirected link between that pair uniformly at random with probability $p=0.15$. Then we normalize entries of the resulting adjacency matrix such that the rows sum to one (that is we set $G_{ij} = 1 / d_i$ for all $i,j$, where $d_i$ is the number of distinct neighbours that $i$ has).\footnote{In principle, it is possible that an agent has no neighbours, and so this would involve dividing by zero. But this would happen with probability $(1-p)^n$, which here is $0.85^{200} \approx 7.6 \times 10^{-15}$. In any case, as with the other network formation processes, we disregard any network that is not connected and draw a new one.}

\paragraph{Scale-free random network.} We start at step $k=0$ with a star network on $17$ nodes.\footnote{This is equivalent to starting with an empty network on $16$ nodes, because the next node created in the process is given links with $16$ nodes.}
Then at each step $k\geq1$, add a new node, and form $16$ links with existing nodes. The probability a new node $i$ forms a link with an existing node $j$ is proportional to the current degree of node $j$. The network formation process stops at step $K=183$ when we obtain an undirected, unweighted network formed by 200 agents. Then we normalize entries of the resulting adjacency matrix such that the rows sum to one (that is we set $G_{ij} = 1 / d_i$ for all $i,j$, where $d_i$ is the number of distinct neighbours that $i$ has).

\end{document}